\RequirePackage[l2tabu,orthodox]{nag}
\documentclass{article}

\pdfoutput=1

\usepackage[T1]{fontenc}
\usepackage{lmodern}
\usepackage{fixltx2e}
\usepackage{microtype} 
\usepackage{xspace}

\makeatletter
\newcommand\etc{etc\@ifnextchar.{}{.\@}\xspace}
\newcommand\ie{i.e.\@\xspace}  
\newcommand\eg{e.g.\@\xspace}
\makeatother

\usepackage{graphicx}

\newcommand{\inlinegraphic}[2]{
  \dimendef\grafheight=255\dimendef\grafvshift=254
  \grafheight=#1
  \grafvshift=-0.5\grafheight
  \advance\grafvshift by 0.5ex
  \raisebox{\grafvshift}{\includegraphics[height=\grafheight]{images/#2}\xspace}
}

\newcommand{\ninlinegraphic}[2][1.0]{
  \dimendef\grafheight=255\dimendef\grafvshift=254
  \setbox0 = \hbox{\scalebox{#1}{\includegraphics{images/#2}}}
  \grafheight=\the\ht0
  \grafvshift=-0.5\grafheight
  \advance\grafvshift by 0.5ex
  \raisebox{\grafvshift}{\includegraphics[height=\grafheight]{images/#2}\xspace}
}

\newcommand{\inline}[1]{
  \raisebox{0.5ex}{\;#1\;}
}

\usepackage{latexsym}
\usepackage{amssymb}
\usepackage{amsmath} 
\usepackage{stmaryrd}
\usepackage{mathtools}

\usepackage{amsthm}
\newtheorem{theorem}{Theorem}[section]
\newtheorem{proposition}[theorem]{Proposition}
\newtheorem{lemma}[theorem]{Lemma}
\newtheorem{corollary}[theorem]{Corollary}
\theoremstyle{definition}\newtheorem{example}[theorem]{Example}
\theoremstyle{definition}
\theoremstyle{definition}\newtheorem{definition}[theorem]{Definition}
\theoremstyle{definition}
\theoremstyle{definition}\newtheorem{remark}[theorem]{Remark}
\theoremstyle{definition}

\newcommand{\denote}[1]{
\llbracket #1 \rrbracket}

\newcommand{\sizeof}[1]{
  \left|#1\right|}

\newcommand{\bra}[1]{
    \ensuremath{\left\langle #1 \right|}\xspace}
\newcommand{\ket}[1]{
    \ensuremath{\left|  #1 \right\rangle}\xspace}
\newcommand{\innp}[2]{
    \ensuremath{\langle #1 \!\mid\!#2 \rangle}}
\newcommand{\outp}[2]{
    \ensuremath{\left|\left. #1 \rangle\right.\!\!\langle\left. #2 \right|\right. }}

\newcommand{\CZ}{\ensuremath{\wedge Z}\xspace}
\newcommand{\CX}{\ensuremath{\wedge X}\xspace}

\usepackage{diagrams} 
\newarrow{Equals}{}{=}{}{=}{}

\ifx\numericids\undefined
\newcommand{\id}[1]{\ensuremath{\mathrm{id}_{#1}}}
\else
\newcommand{\id}[1]{\ensuremath{1_{#1}}}
\fi

\newcommand{\fdhilb}{
\ensuremath{\mathbf{fdHilb}}\xspace}
\newcommand{\fdHilb}{\fdhilb}

\newcommand{\superop}{
\ensuremath{\mathbf{SuperOp}}\xspace}

\newcommand{\dsmc}{$\dag$-\textsc{smc}\xspace}

\newcommand{\dsmcs}{$\dag$-\textsc{smc}s\xspace}
\newcommand{\smcs}{\textsc{smc}s\xspace}

\newcommand{\zxcalculus}{\textsc{zx}-calculus\xspace}

\newcommand{\DD}{\ensuremath{\mathbb{D}}\xspace}
\newcommand{\DDV}{\ensuremath{\mathbb{D}(\mathcal{V})}\xspace}

\newcommand{\denoteV}[1]{\denote{#1}_\mathcal{V}}

\newcommand{\patP}{\ensuremath{\mathfrak{P}}\xspace}

\newcommand{\G}[1]{\ensuremath{\Gamma_{#1}}\xspace}
\newcommand{\Gp}{\G\patP}

\newcommand{\D}[1]{\ensuremath{D(#1)}\xspace}
\newcommand{\DG}{\D{\Gamma}}

\newcommand{\Dp}{\D{\patP}}
\newcommand{\DGp}{\D{\Gp}}
\newcommand{\DGpp}{\ensuremath{\DGp^*}\xspace}

\newcommand{\mbqc}{\textsc{mbqc}\xspace}
\newcommand{\owqc}{\textsc{1wqc}\xspace}

\usepackage{hyperref}

\usepackage{tikz}
\usepackage{pgfplots}
\usetikzlibrary{trees}
\usetikzlibrary{topaths}
\usetikzlibrary{decorations.pathmorphing}
\usetikzlibrary{decorations.markings}
\usetikzlibrary{matrix,backgrounds,folding}
\usetikzlibrary{chains,scopes,positioning,fit}
\usetikzlibrary{arrows,shadows}
\usetikzlibrary{calc} 
\usetikzlibrary{chains}
\usetikzlibrary{shapes,shapes.geometric,shapes.misc}

\usepackage{tikzquanto}
\usepackage{tikzcircuits}

\pgfdeclarelayer{edgelayer}
\pgfdeclarelayer{nodelayer}
\pgfsetlayers{background,edgelayer,nodelayer,main}

\tikzstyle{every picture}=[baseline=(current bounding box).east,scale=0.5,node distance=5mm]
\tikzstyle{none}=[inner sep=0mm]
\tikzstyle{every loop}=[]

\tikzstyle{(null)}=[]
\tikzstyle{plain}=[]

\newcommand{\greenspider}[1]{\inline{%
\begin{tikzpicture}[quanto]
  \spider{green vertex}{spideri}{0,0}
  \node [green angle] at (spideri) {#1};
\end{tikzpicture}
}}

\newcommand{\redspider}[1]{\inline{%
\begin{tikzpicture}[quanto]
  \spider{red vertex}{spideri}{0,0}
  \node [red angle] at (spideri) {#1};
\end{tikzpicture}
}}

\newcommand{\hgate}{\inline{%
\begin{tikzpicture}[quanto]
  \node [boundary vertex] (a) at (1.05,-1.0) {};
  \node [hadamard vertex] (c) at (1.05,-2.25) {};
  \node [boundary vertex] (d) at (1.05,-3.525) {};
  \draw [] (a) to (c);
  \draw [] (c) to (d);
\end{tikzpicture}
}}

\newcommand{\greendelta}{\inline{%
\begin{tikzpicture}[quanto]
  \node [style=boundary vertex] (0) at (0, 1) {};
  \node [style=green vertex] (1) at (0, 0) {};
  \node [style=boundary vertex] (2) at (-1, -1) {};
  \node [style=boundary vertex] (3) at (1, -1) {};
  \draw  (1) to (2);
  \draw  (1) to (3);
  \draw  (0) to (1);
\end{tikzpicture}
}}

\newcommand{\greenmu}{\inline{%
\begin{tikzpicture}[quanto]
  \node [style=boundary vertex] (0) at (0, -1) {};
  \node [style=green vertex] (1) at (0, 0) {};
  \node [style=boundary vertex] (2) at (-1, 1) {};
  \node [style=boundary vertex] (3) at (1, 1) {};
  \draw  (1) to (2);
  \draw  (1) to (3);
  \draw  (0) to (1);
\end{tikzpicture}
}}

\newcommand{\greenunit}{\inline{%
  \begin{tikzpicture}[quanto]
    \node [boundary vertex] (a) at (1.0,-2.0) {};
    \node [green vertex] (b) at (1.0,-1.0) {};
    \draw [] (b) to (a);
  \end{tikzpicture}
}}

\newcommand{\greencounit}{\inline{%
  \begin{tikzpicture}[quanto]
    \node [boundary vertex] (a) at (1.0,2.0) {};
    \node [green vertex] (b) at (1.0,1.0) {};
    \draw [] (b) to (a);
  \end{tikzpicture}
}}

\newcommand{\greenpoint}[1]{\inline{%
  \begin{tikzpicture}[quanto]
    \node [boundary vertex] (a) at (1.0,-2.0) {};
    \node [green vertex] (b) at (1.0,-1.0) {};
    \draw [] (b) to (a);
    \node [green angle] at (b) {$#1$};
  \end{tikzpicture}
}}

\newcommand{\redpoint}[1]{\inline{%
  \begin{tikzpicture}[quanto]
    \node [boundary vertex] (a) at (1.0,-2.0) {};
    \node [red vertex] (b) at (1.0,-1.0) {};
    \draw [] (b) to (a);
    \node [red angle] at (b) {$#1$};
  \end{tikzpicture}
}}

\newcommand{\reddelta}{\inline{%
\begin{tikzpicture}[quanto]
  \node [style=boundary vertex] (0) at (0, 1) {};
  \node [style=red vertex] (1) at (0, 0) {};
  \node [style=boundary vertex] (2) at (-1, -1) {};
  \node [style=boundary vertex] (3) at (1, -1) {};
  \draw  (1) to (2);
  \draw  (1) to (3);
  \draw  (0) to (1);
\end{tikzpicture}
}}

\newcommand{\redunit}{\inline{%
  \begin{tikzpicture}[quanto]
    \node [boundary vertex] (a) at (1.0,-2.0) {};
    \node [red vertex] (b) at (1.0,-1.0) {};
    \draw [] (b) to (a);
  \end{tikzpicture}
}}

\newcommand{\redcounit}{\inline{%
  \begin{tikzpicture}[quanto]
    \node [boundary vertex] (a) at (1.0,2.0) {};
    \node [red vertex] (b) at (1.0,1.0) {};
    \draw [] (b) to (a);
  \end{tikzpicture}
}}

\newcommand{\greenphase}[1]{\inline{%
  \begin{tikzpicture}[quanto]
    \node [green vertex] (c) at (1.175,-2.0) {};
    \node [boundary vertex] (b) at (1.175,-3.0) {};
    \node [boundary vertex] (a) at (1.175,-1.0) {};
    \draw [] (a) to (c);
    \draw [] (c) to (b);
    \node [green angle] at (c) {$#1$};
  \end{tikzpicture}
}}

\newcommand{\twogreenphases}[2]{\inline{%
  \begin{tikzpicture}[quanto]
    \node [boundary vertex] (a) at (1.175,0) {};
    \node [green vertex] (b) at (1.175,-1.0) {};
    \node [green vertex] (c) at (1.175,-2.25) {};
    \node [boundary vertex] (d) at (1.175,-3.25) {};
    \draw [] (a) to (b);
    \draw [] (c) to (b);
    \draw [] (c) to (d);
    \node [green angle] at (b) {$#1$};
    \node [green angle] at (c) {$#2$};
  \end{tikzpicture}
}}

\newcommand{\greenphasesig}[2]{\inline{%
  \begin{tikzpicture}[quanto]
    \node [green vertex] (c) at (1.175,-2.0) {};
    \node [boundary vertex] (b) at (1.175,-3.0) {};
    \node [boundary vertex] (a) at (1.175,-1.0) {};
    \draw [] (a) to (c);
    \draw [] (c) to (b);
    \node [green angle] at (c) {$#1,#2$};
  \end{tikzpicture}
}}

\newcommand{\redphase}[1]{\inline{%
  \begin{tikzpicture}[quanto]
    \node [red vertex] (c) at (1.175,-2.0) {};
    \node [boundary vertex] (b) at (1.175,-3.0) {};
    \node [boundary vertex] (a) at (1.175,-1.0) {};
    \draw [] (a) to (c);
    \draw [] (c) to (b);
    \node [red angle] at (c) {$#1$};
  \end{tikzpicture}
}}

\newcommand{\redphasesig}[2]{\inline{%
  \begin{tikzpicture}[quanto]
    \node [red vertex] (c) at (1.175,-2.0) {};
    \node [boundary vertex] (b) at (1.175,-3.0) {};
    \node [boundary vertex] (a) at (1.175,-1.0) {};
    \draw [] (a) to (c);
    \draw [] (c) to (b);
    \node [red angle] at (c) {$#1,#2$};
  \end{tikzpicture}
}}

\newcommand{\czed}{\inline{%
  \begin{tikzpicture}[quanto]
    \node [boundary vertex] (a) at (1.0,-1.0) {};
    \node [boundary vertex] (d) at (3.0,-3.0) {};
    \node [boundary vertex] (c) at (1.0,-3.0) {};
    \node [boundary vertex] (b) at (3.0,-1) {};
    \node [hadamard vertex] (g) at (2.0,-2.0) {};
    \node [green vertex] (f) at (1.0,-2.0) {};
    \node [green vertex] (e) at (3.0,-2.0) {};
    \draw [] (f) to (g);
    \draw [] (e) to (d);
    \draw [] (f) to (c);
    \draw [] (b) to (e);
    \draw [] (a) to (f);
    \draw [] (g) to (e);
  \end{tikzpicture}
}}

\newcommand{\cex}{\inline{%
  \begin{tikzpicture}[quanto]
    \node [boundary vertex] (a) at (1.0,-1.0) {};
    \node [boundary vertex] (d) at (3.0,-3.0) {};
    \node [boundary vertex] (c) at (1.0,-3.0) {};
    \node [boundary vertex] (b) at (3.0,-1) {};
    \node [green vertex] (f) at (1.0,-2.0) {};
    \node [red vertex] (e) at (3.0,-2.0) {};
    \draw [] (f) to (e);
    \draw [] (e) to (d);
    \draw [] (f) to (c);
    \draw [] (b) to (e);
    \draw [] (a) to (f);
  \end{tikzpicture}
}}

\newcommand{\meas}{\inline{%
  \begin{tikzpicture}[quanto]
    \node [green vertex] (c) at (1.3,-2.25) {};
    \node [boundary vertex] (d) at (1.3,-1.0) {};
    \node [green vertex] (b) at (1.3,-3.5) {};
    \draw [] (c) to (b);
    \draw [] (d) to (c);
    \node [green angle] at (c) {$\pi,\{i\}$};
    \node [green angle] at (b) {$-\alpha$};
  \end{tikzpicture}
}}

\newcommand{\etapic}{\inline{%
\inline{\begin{tikzpicture}[quanto]
  \node [boundary vertex] (0) at (0, 0) {};
  \node [boundary vertex] (1) at (2, 0) {};
  \draw [bend left=90, looseness=2.00](0) to (1);
\end{tikzpicture}}
}}

\newcommand{\etaproofi}{\inline{%
  \inline{\begin{tikzpicture}[quanto]
      \node [green vertex] (0) at (1, 2) {};
      \node [green vertex] (1) at (1, 1) {};
      \node [boundary vertex] (2) at (0, 0) {};
      \node [boundary vertex] (3) at (2, 0) {};
      \draw  (0) to (1);
      \draw [bend left=45] (2) to (1);
      \draw [bend right=45] (3) to (1);
    \end{tikzpicture}}
}}

\newcommand{\etaproofii}{\inline{%
\inline{\begin{tikzpicture}[quanto]
      \node [green vertex] (1) at (1, 1) {};
      \node [boundary vertex] (2) at (0, 0) {};
      \node [boundary vertex] (3) at (2, 0) {};
      \draw [bend left=45] (2) to (1);
      \draw [bend right=45] (3) to (1);
\end{tikzpicture}}
}}

\newcommand{\greeneta}{\etaproofii}
\newcommand{\redeta}{\inline{%
\inline{\begin{tikzpicture}[quanto]
      \node [red vertex] (1) at (1, 1) {};
      \node [boundary vertex] (2) at (0, 0) {};
      \node [boundary vertex] (3) at (2, 0) {};
      \draw [bend left=45] (2) to (1);
      \draw [bend right=45] (3) to (1);
\end{tikzpicture}}
}}

\newcommand{\cnoti}[1][1.0]{\inline{\scalebox{#1}{%
  \begin{tikzpicture}[quanto]
    \node [boundary vertex] (a) at (-3.5,2.0) {};
    \node [green vertex] (d) at (0.0,2.0) {};
    \node [green vertex] (t) at (3.5,2.0) {};
    \node [boundary vertex] (w) at (4.5,2.0) {};
    \node [green vertex] (z) at (-3.5,3.0) {};
    \node [green vertex] (r) at (0.0,3.0) {};
    \node [green vertex] (v) at (3.5,3.0) {};
    \node [green vertex] (h) at (-3.5,4.0) {};
    \node [hadamard vertex] (j) at (-1.75,4.0) {};
    \node [green vertex] (e) at (0.0,4.0) {};
    \node [red vertex] (k) at (4.0,5.0) {};
    \node [green vertex] (s) at (4.0,6.0) {};
    \node [green vertex] (m) at (2.0,6.0) {};
    \node [hadamard vertex] (p) at (1.0,6.0) {};
    \node [green vertex] (o) at (0.0,6.0) {};
    \node [hadamard vertex] (n) at (3.0,7.0) {};
    \node [green vertex] (l) at (4.0,7.0) {};
    \node [green vertex] (i) at (2.0,7.0) {};
    \node [boundary vertex] (b) at (-3.5,8.0) {};
    \node [boundary vertex] (c) at (0.0,8.0) {};
    \node [green vertex] (g) at (2.0,8.0) {};
    \node [green vertex] (f) at (4.0,8.0) {};
    \draw [] (g) to (i);
    \draw [] (l) to (s);
    \draw [] (e) to (r);
    \draw [] (i) to (n);
    \draw [] (c) to (o);
    \draw [] (p) to (m);
    \draw [] (n) to (l);
    \draw [] (b) to (h);
    \draw [] (j) to (e);
    \draw [] (h) to (z);
    \draw [] (z) to (a);
    \draw [] (m) to (i);
    \draw [] (s) to (k);
    \draw [] (k) to (w);
    \draw [] (f) to (l);
    \draw [] (o) to (v);
    \draw [] (o) to (p);
    \draw [] (v) to (t);
    \draw [] (r) to (d);
    \draw [] (h) to (j);
    \draw [] (m) to (e);
    \node [red angle] at (k) {$\pi,\{3\}$};
    \node [green angle] at (v) {$\pi,\{2\}$};
    \node [green angle] at (s) {$\pi,\{2\}$};
    \node [green angle] at (r) {$\pi,\{3\}$};
    \node [green angle] at (z) {$\pi,\{2\}$};
  \end{tikzpicture}
}}}

\newcommand{\cnotii}[1][1.0]{\inline{\scalebox{#1}{
\begin{tikzpicture}[quanto]
\node [boundary vertex] (a) at (-1.5,-1.0) {};
\node [boundary vertex] (w) at (2.0,-1.0) {};
\node [red vertex] (k) at (2.0,0.0) {};
\node [green vertex] (z) at (-1.5,0.0) {};
\node [green vertex] (spider2) at (2.0,1.0) {};
\node [green vertex] (spider1) at (1.0,3.0) {};
\node [green vertex] (guff2) at (2.0,3.0) {};
\node [green vertex] (guff22) at (2.0,4.0) {};
\node [hadamard vertex] (n) at (2.0,2.0) {};
\node [hadamard vertex] (j) at (-0.25,3.0) {};
\node [green vertex] (h) at (-1.5,3.0) {};
\node [hadamard vertex] (p) at (1.0,4.0) {};
\node [green vertex] (spider0) at (1.0,5.0) {};
\node [boundary vertex] (b) at (-1.5,6.0) {};
\node [boundary vertex] (c) at (1.0,6.0) {};
\draw [] (spider0) to (p);
\draw [] (k) to (w);
\draw [] (spider2) to (k);
\draw [] (n) to (spider2);
\draw [] (h) to (z);
\draw [] (z) to (a);
\draw [] (h) to (j);
\draw [] (c) to (spider0);
\draw [] (spider1) to (n);
\draw [] (j) to (spider1);
\draw [] (b) to (h);
\draw [] (p) to (spider1);
\draw [] (guff2) to (spider1);
\draw [] (guff2) to (guff22);
\node [green angle] at (guff2) {$\pi,\{3\}$};
\node [green angle] at (spider0) {$\pi,\{2\}$};
\node [green angle] at (spider2) {$\pi,\{2\}$};
\node [green angle] at (z) {$\pi,\{2\}$};
\node [red angle] at (k) {$\pi,\{3\}$};
\end{tikzpicture}
}}}

\newcommand{\cnotiii}[1][1.0]{\inline{\scalebox{#1}{%
\begin{tikzpicture}[quanto]
\node [boundary vertex] (a) at (-1.5,0.0) {};
\node [boundary vertex] (w) at (2.0,0.0) {};
\node [green vertex] (d) at (-1.5,1.0) {};
\node [red vertex] (k) at (2.0,1.0) {};
\node [red vertex] (kk) at (2.0,2.0) {};
\node [green vertex] (spider2) at (2.0,3.0) {};
\node [hadamard vertex] (n) at (2.0,4.0) {};
\node [green vertex] (spider1) at (2.0,5.0) {};
\node [hadamard vertex] (j) at (1.0,5.0) {};
\node [green vertex] (h) at (-1.5,5.0) {};
\node [hadamard vertex] (p) at (2.0,6.0) {};
\node [green vertex] (spider0) at (2.0,7.0) {};
\node [boundary vertex] (b) at (-1.5,8.0) {};
\node [boundary vertex] (c) at (2.0,8.0) {};
\draw [] (c) to (spider0);
\draw [] (h) to (j);
\draw [] (n) to (spider2);
\draw [] (k) to (kk);
\draw [] (k) to (w);
\draw [] (b) to (h);
\draw [] (p) to (spider1);
\draw [] (spider1) to (n);
\draw [] (h) to (d);
\draw [] (spider0) to (p);
\draw [] (j) to (spider1);
\draw [] (spider2) to (kk);
\draw [] (d) to (a);
\node [red angle] at (k) {$\pi,\{3\}$};
\node [red angle] at (kk) {$\pi,\{3\}$};
\node [green angle] at (spider0) {$\pi,\{2\}$};
\node [green angle] at (d) {$\pi,\{2\}$};
\node [green angle] at (spider2) {$\pi,\{2\}$};
\end{tikzpicture}
}}}

\newcommand{\cnotiv}[1][1.0]{\inline{\scalebox{#1}{%
\begin{tikzpicture}[quanto]
\node [boundary vertex] (a) at (0.0,2.0) {};
\node [boundary vertex] (w) at (2.0,2.0) {};
\node [green vertex] (d) at (0.0,3.0) {};
\node [green vertex] (spider2) at (2.0,4.0) {};
\node [red vertex] (o) at (2.0,5.0) {};
\node [green vertex] (h) at (0.0,5.0) {};
\node [green vertex] (spider0) at (2.0,6.0) {};
\node [boundary vertex] (b) at (0.0,7.0) {};
\node [boundary vertex] (c) at (2.0,7.0) {};
\draw [] (b) to (h);
\draw [] (c) to (spider0);
\draw [] (h) to (d);
\draw [] (w) to (spider2);
\draw [] (o) to (spider2);
\draw [] (d) to (a);
\draw [] (h) to (o);
\draw [] (spider0) to (o);
\node [green angle] at (spider0) {$\pi,\{2\}$};
\node [green angle] at (d) {$\pi,\{2\}$};
\node [green angle] at (spider2) {$\pi,\{2\}$};
\end{tikzpicture}
}}}

\newcommand{\cnotv}[1][1.0]{\inline{\scalebox{#1}{%
\begin{tikzpicture}[quanto]
\node [boundary vertex] (a) at (0.0,2.0) {};
\node [boundary vertex] (w) at (3.25,2.0) {};
\node [green vertex] (d) at (0.0,3.0) {};
\node [green vertex] (spider0) at (3.25,3.0) {};
\node [green vertex] (d2) at (0.0,4.0) {};
\node [green vertex] (spider2) at (3.25,4.0) {};
\node [red vertex] (o) at (3.25,5.0) {};
\node [green vertex] (h) at (0.0,5.0) {};
\node [boundary vertex] (b) at (0.0,7.0) {};
\node [boundary vertex] (c) at (3.25,7.0) {};
\draw [] (b) to (h);
\draw [] (c) to (o);
\draw [] (h) to (d2);
\draw [] (d) to (d2);
\draw [] (spider0) to (spider2);
\draw [] (spider0) to (w);
\draw [] (o) to (spider2);
\draw [] (d) to (a);
\draw [] (h) to (o);
\node [green angle] at (spider0) {$\pi,\{2\}$};
\node [green angle] at (d) {$\pi,\{2\}$};
\node [green angle] at (d2) {$\pi,\{2\}$};
\node [green angle] at (spider2) {$\pi,\{2\}$};
\end{tikzpicture}
}}}

\newcommand{\cnotvi}[1][1.0]{\inline{\scalebox{#1}{%
\begin{tikzpicture}[quanto]
\node [boundary vertex] (a) at (0.0,-1.0) {};
\node [boundary vertex] (w) at (2.0,-1.0) {};
\node [green vertex] (spider2) at (0.0,1.0) {};
\node [red vertex] (j) at (2.0,1.0) {};
\node [boundary vertex] (b) at (0.0,3.0) {};
\node [boundary vertex] (c) at (2.0,3.0) {};
\draw [] (spider2) to (a);
\draw [] (c) to (j);
\draw [] (w) to (j);
\draw [] (b) to (spider2);
\draw [] (j) to (spider2);
\end{tikzpicture}
}}}

\usetikzlibrary{fit}

\newcommand{\ctrlXskewed}{\inline{%
    \begin{tikzpicture}[quanto]
      \node [boundary vertex] (a) at (1.05,-1.0) {};
      \node [boundary vertex] (d) at (2.5,-4.8) {};
      \node [boundary vertex] (c) at (1.05,-4.8) {};
      \node [boundary vertex] (b) at (2.5,-1) {};
      \node [green vertex] (f) at (1.05,-2.25) {};
      \node [red vertex] (e) at (2.5,-3.525) {};
      \draw [] (a) to (f);
      \draw [] (f) to (c);
      \draw [] (f) to (e);
      \draw [] (b) to (e);
      \draw [] (e) to (d);
    \end{tikzpicture}}
}

\newcommand{\ctrlXskewedbis}{\inline{%
    \begin{tikzpicture}[quanto]
      \node [boundary vertex] (a) at (1.05,-1) {};
      \node [boundary vertex] (d) at (2.5,-4.8) {};
      \node [boundary vertex] (c) at (1.05,-4.8) {};
      \node [boundary vertex] (b) at (2.5,-1.0) {};
      \node [green vertex] (f) at (1.05,-3.525) {};
      \node [red vertex] (e) at (2.5,-2.25) {};
      \draw [] (e) to (d);
      \draw [] (a) to (f);
      \draw [] (e) to (f);
      \draw [] (b) to (e);
      \draw [] (f) to (c);
    \end{tikzpicture}}
}

\newcommand{\CZviaCX}{\inline{%
\begin{tikzpicture}[quanto]
  \node [style=boundary vertex] (0) at (-1, 4) {};
  \node [style=boundary vertex] (1) at (1, 4) {};
  \node [style=hadamard vertex] (2) at (1, 3) {};
  \node [style=green vertex] (3) at (-1, 2) {};
  \node [style=red vertex] (4) at (1, 2) {};
  \node [style=hadamard vertex] (5) at (1, 1) {};
  \node [style=boundary vertex] (6) at (-1, 0) {};
  \node [style=boundary vertex] (7) at (1, 0) {};
  \draw  (3) to (4);
  \draw  (2) to (1);
  \draw  (7) to (5);
  \draw  (0) to (3);
  \draw  (3) to (6);
  \draw  (4) to (2);
  \draw  (5) to (4);
\end{tikzpicture}
}}

\newcommand{\CZselfinvi}{\inline{%
\begin{tikzpicture}[quanto]
  \node [style=boundary vertex] (0) at (-1, 3) {};
  \node [style=boundary vertex] (1) at (1, 3) {};
  \node [style=green vertex] (2) at (-1, 2) {};
  \node [style=hadamard vertex] (3) at (0, 2) {};
  \node [style=green vertex] (4) at (1, 2) {};
  \node [style=green vertex] (5) at (-1, 0.5) {};
  \node [style=hadamard vertex] (6) at (0, 0.5) {};
  \node [style=green vertex] (7) at (1, 0.5) {};
  \node [style=boundary vertex] (8) at (-1, -0.5) {};
  \node [style=boundary vertex] (9) at (1, -0.5) {};
  \draw  (6) to (5);
  \draw  (9) to (7);
  \draw  (4) to (7);
  \draw  (8) to (5);
  \draw  (0) to (2);
  \draw  (2) to (3);
  \draw  (2) to (5);
  \draw  (4) to (1);
  \draw  (3) to (4);
  \draw  (7) to (6);
\end{tikzpicture}
}}

\newcommand{\CZselfinvii}{\inline{%
\begin{tikzpicture}[quanto]
  \node [style=boundary vertex] (0) at (-1, 3) {};
  \node [style=boundary vertex] (1) at (1, 3) {};
  \node [style=hadamard vertex] (2) at (0, 2) {};
  \node [style=green vertex] (3) at (-1, 1.25) {};
  \node [style=green vertex] (4) at (1, 1.25) {};
  \node [style=hadamard vertex] (5) at (0, 0.5) {};
  \node [style=boundary vertex] (6) at (-1, -0.5) {};
  \node [style=boundary vertex] (7) at (1, -0.5) {};
  \draw [bend right] (5) to (4);
  \draw  (3) to (6);
  \draw  (0) to (3);
  \draw [bend left] (3) to (2);
  \draw [bend left] (5) to (3);
  \draw  (4) to (7);
  \draw  (4) to (1);
  \draw [bend left] (2) to (4);
\end{tikzpicture}
}}

\newcommand{\CZselfinviii}{\inline{%
\begin{tikzpicture}[quanto]
  \node [style=boundary vertex] (0) at (-1, 3) {};
  \node [style=boundary vertex] (1) at (1, 3) {};
  \node [style=hadamard vertex] (2) at (1, 2.25) {};
  \node [style=green vertex] (3) at (-1, 1.25) {};
  \node [style=red vertex] (4) at (1, 1.25) {};
  \node [style=hadamard vertex] (5) at (1, 0.25) {};
  \node [style=boundary vertex] (6) at (-1, -0.5) {};
  \node [style=boundary vertex] (7) at (1, -0.5) {};
  \draw  (1) to (2);
  \draw  (5) to (7);
  \draw  (3) to (6);
  \draw  (0) to (3);
  \draw [bend left=60] (3) to (4);
  \draw  (2) to (4);
  \draw [bend right=60] (3) to (4);
  \draw  (4) to (5);
\end{tikzpicture}
}}

\newcommand{\CZselfinviv}{\inline{%
\begin{tikzpicture}[quanto]
  \node [style=boundary vertex] (0) at (-1, 3) {};
  \node [style=boundary vertex] (1) at (1, 3) {};
  \node [style=hadamard vertex] (2) at (1, 2.25) {};
  \node [style=green vertex] (3) at (-1, 1.25) {};
  \node [style=red vertex] (4) at (1, 1.25) {};
  \node [style=hadamard vertex] (5) at (1, 0.25) {};
  \node [style=boundary vertex] (6) at (-1, -0.5) {};
  \node [style=boundary vertex] (7) at (1, -0.5) {};
  \draw  (1) to (2);
  \draw  (5) to (7);
  \draw  (3) to (6);
  \draw  (0) to (3);
  \draw  (2) to (4);
  \draw  (4) to (5);
\end{tikzpicture}
}}

\newcommand{\CZselfinvv}{\inline{%
\begin{tikzpicture}[quanto]
  \node [style=boundary vertex] (0) at (-1, 3) {};
  \node [style=boundary vertex] (1) at (1, 3) {};
  \node [style=hadamard vertex] (2) at (1, 2) {};
  \node [style=hadamard vertex] (3) at (1, 0.5) {};
  \node [style=boundary vertex] (4) at (-1, -0.5) {};
  \node [style=boundary vertex] (5) at (1, -0.5) {};
  \draw  (3) to (5);
  \draw  (2) to (3);
  \draw  (1) to (2);
  \draw  (0) to (4);
\end{tikzpicture}
}}

\newcommand{\CZselfinvvi}{\inline{%
\begin{tikzpicture}[quanto]
  \node [style=boundary vertex] (0) at (-1, 3) {};
  \node [style=boundary vertex] (1) at (1, 3) {};
  \node [style=boundary vertex] (2) at (-1, -0.5) {};
  \node [style=boundary vertex] (3) at (1, -0.5) {};
  \draw  (1) to (3);
  \draw  (0) to (2);
\end{tikzpicture}
}}

\newcommand{\makeBellPfi}{\inline{%
\begin{tikzpicture}[quanto]
  \node [style=red vertex] (0) at (-1, 2) {};
  \node [style=red vertex] (1) at (1, 2) {};
  \node [style=hadamard vertex] (2) at (-1, 1) {};
  \node [style=green vertex] (3) at (-1, 0) {};
  \node [style=red vertex] (4) at (1, 0) {};
  \node [style=boundary vertex] (5) at (-1, -1) {};
  \node [style=boundary vertex] (6) at (1, -1) {};
  \draw  (4) to (1);
  \draw  (0) to (2);
  \draw  (2) to (3);
  \draw  (3) to (5);
  \draw  (3) to (4);
  \draw  (4) to (6);
\end{tikzpicture}
}}

\newcommand{\makeBellPfii}{\inline{%
\begin{tikzpicture}[quanto]
  \node [style=green vertex] (0) at (-1, 1.5) {};
  \node [style=green vertex] (1) at (-1, 0) {};
  \node [style=red vertex] (2) at (1, 0) {};
  \node [style=boundary vertex] (3) at (-1, -1) {};
  \node [style=boundary vertex] (4) at (1, -1) {};
  \draw  (0) to (1);
  \draw  (1) to (3);
  \draw  (1) to (2);
  \draw  (2) to (4);
\end{tikzpicture}
}}

\newcommand{\makeBellPfiii}{\inline{%
\begin{tikzpicture}[quanto]
  \node [style=green vertex] (0) at (0, 0) {};
  \node [style=boundary vertex] (1) at (-1, -1) {};
  \node [style=boundary vertex] (2) at (1, -1) {};
  \draw [bend right=45] (0) to (1);
  \draw [bend left=45] (0) to (2);
\end{tikzpicture}
}}

\newcommand{\exnotcircuitlike}{\inline{%
\begin{tikzpicture}[quanto]
  \node [style=boundary vertex] (0) at (-1, 2) {};
  \node [style=boundary vertex] (1) at (1, 2) {};
  \node [style=green vertex] (2) at (-1, 1) {};
  \node [style=red vertex] (3) at (1, 1) {};
  \node [style=hadamard vertex] (4) at (-1, 0) {};
  \node [style=hadamard vertex] (5) at (1, 0) {};
  \node [style=green vertex] (6) at (-1, -1) {};
  \node [style=red vertex] (7) at (1, -1) {};
  \node [style=boundary vertex] (8) at (-1, -2) {};
  \node [style=boundary vertex] (9) at (1, -2) {};
  \draw  (6) to (8);
  \draw  (2) to (4);
  \draw  (7) to (9);
  \draw  (2) to (7);
  \draw  (3) to (5);
  \draw  (3) to (6);
  \draw  (1) to (3);
  \draw  (2) to (0);
  \draw  (6) to (4);
  \draw  (5) to (7);
\end{tikzpicture}
}}

\newcommand{\exnotcircuitlikebis}{\inline{%
\begin{tikzpicture}[quanto]
  \node [style=boundary vertex] (0) at (-1, 2) {};
  \node [style=boundary vertex] (1) at (1, 2) {};
  \node [style=green vertex,fill=gray!10] (2) at (-1, 1) {};
  \node [style=red vertex,fill=gray!10] (3) at (1, 1) {};
  \node [style=hadamard vertex,fill=gray!10] (4) at (-1, 0) {};
  \node [style=hadamard vertex,fill=gray!10] (5) at (1, 0) {};
  \node [style=green vertex,fill=gray!10] (6) at (-1, -1) {};
  \node [style=red vertex,fill=gray!10] (7) at (1, -1) {};
  \node [style=boundary vertex] (8) at (-1, -2) {};
  \node [style=boundary vertex] (9) at (1, -2) {};
  \begin{scope}[blue!50,line width=2.5pt]
    \coordinate  (2a) at (-1, 1) {};
    \coordinate  (3a) at (1, 1) {};
    \coordinate  (4a) at (-1, 0) {};
    \coordinate  (5a) at (1, 0) {};
    \coordinate  (6a) at (-1, -1) {};
    \coordinate  (7a) at (1, -1) {};
    \draw (2a) -- (4a); 
    \draw (4a) -- (6a);
    \draw (6a) -- (3a);
    \draw (3a) -- (5a);
    \draw (5a) -- (7a);
    \draw (7a) -- (2a);
  \end{scope}
  \draw  (6) to (8);
  \draw  (2) to (4);
  \draw  (7) to (9);
  \draw  (2) to (7);
  \draw  (3) to (5);
  \draw  (3) to (6);
  \draw  (1) to (3);
  \draw  (2) to (0);
  \draw  (6) to (4);
  \draw  (5) to (7);
\end{tikzpicture}
}}

\newcommand{\excircuitlikei}{\inline{%
\begin{tikzpicture}[quanto]
  \node [style=boundary vertex] (0) at (0, 2) {};
  \node [style=green vertex] (1) at (0, 1) {};
  \node [style=red vertex] (2) at (-1.5, 0) {};
  \node [style=hadamard vertex] (3) at (0, 0) {};
  \node [style=green vertex] (4) at (0, -1) {};
  \node [style=boundary vertex] (5) at (-1.5, -2) {};
  \node [style=boundary vertex] (6) at (0, -2) {};
  \draw  (0) to (1);
  \draw  (2) to (5);
  \draw  (1) to (3);
  \draw  (2) to (4);
  \draw  (1) to (2);
  \draw  (4) to (6);
  \draw  (3) to (4);
\end{tikzpicture}
}}

\newcommand{\excircuitlikeii}{\inline{%
\begin{tikzpicture}[quanto]
  \node [style=red vertex] (0) at (-1.5, 2) {};
  \node [style=boundary vertex] (1) at (0, 2) {};
  \node [style=red vertex] (2) at (-1.5, 1) {};
  \node [style=green vertex] (3) at (0, 1) {};
  \node [style=hadamard vertex] (4) at (0, 0) {};
  \node [style=red vertex] (5) at (-1.5, -1) {};
  \node [style=green vertex] (6) at (0, -1) {};
  \node [style=boundary vertex] (7) at (-1.5, -2) {};
  \node [style=boundary vertex] (8) at (0, -2) {};
  \draw  (1) to (3);
  \draw  (5) to (7);
  \draw  (3) to (4);
  \draw  (5) to (6);
  \draw  (2) to (0);
  \draw  (6) to (8);
  \draw  (4) to (6);
  \draw  (2) to (5);
  \draw  (3) to (2);
\end{tikzpicture}
}}

\newcommand{\teleportLabelled}{\inline{%
\begin{tikzpicture}[quanto]
  \node [style=boundary vertex] (0) at (-4, 3) {};
  \node [style=red vertex] (1) at (-2, 3) {};
  \node [style=red vertex] (2) at (1.5, 3) {};
  \node [style=hadamard vertex] (3) at (-2, 2) {};
  \node [style=green vertex] (4) at (-2, 1) {};
  \node [style=red vertex] (5) at (1.5, 1) {};
  \node [style=red vertex] (6) at (-4, 0) {};
  \node [style=green vertex] (7) at (-2, 0) {};
  \node [style=green vertex] (8) at (1.5, 0) {};
  \node [style=hadamard vertex] (9) at (-4, -1) {};
  \node [style=green vertex] (10) at (-2, -1) {};
  \node [style=red vertex] (11) at (1.5, -1) {};
  \node [style=green vertex] (12) at (-4, -3) {};
  \node [style=green vertex] (13) at (-2, -2) {};
  \node [style=green vertex] (14) at (-4, -4) {};
  \node [style=boundary vertex] (15) at (1.5, -3) {};
  \draw  (11) to (15);
  \draw  (5) to (2);
  \draw  (12) to (14);
  \draw  (7) to (10);
  \draw  (5) to (8);
  \draw  (0) to (6);
  \draw  (3) to (4);
  \draw  (10) to (13);
  \draw  (1) to (3);
  \draw  (4) to (5);
  \draw  (6) to (9);
  \draw  (9) to (12);
  \draw  (4) to (7);
  \draw  (8) to (11);
  \draw  (6) to (7);
  \node [green angle] at (12) {$\pi,\{v_1\}$};
  \node [green angle] at (10) {$\pi,\{v_2\}$};
  \node [red angle] at (11) {$\pi,\{v_1\}$};
  \node [green angle] at (8) {$\pi,\{v_2\}$};
  \begin{pgfonlayer}{background}
    \node [draw=gray!40,fill=gray!30,inner ysep=-10.33mm,inner
    xsep=-6mm,fit=(6) (7) (9) (10) (12) (13) (14),label={[gray!85]left:A}] {}; 
    \node [draw=gray!40,fill=gray!30,inner ysep=-5.5mm,inner
    xsep=-10mm,fit=(1) (2) (3) (4) (5),label={[gray!85]right:B}] {};
  \end{pgfonlayer}
\end{tikzpicture}
}}

\newcommand{\teleportProofi}{\inline{%
\begin{tikzpicture}[quanto]
  \node [style=boundary vertex] (0) at (-4, 3) {};
  \node [style=red vertex] (1) at (-2, 3) {};
  \node [style=red vertex] (2) at (1.5, 3) {};
  \node [style=hadamard vertex] (3) at (-2, 2) {};
  \node [style=green vertex] (4) at (-2, 1) {};
  \node [style=red vertex] (5) at (1.5, 1) {};
  \node [style=red vertex] (6) at (-4, 0) {};
  \node [style=green vertex] (7) at (-2, 0) {};
  \node [style=green vertex] (8) at (1.5, 0) {};
  \node [style=hadamard vertex] (9) at (-4, -1) {};
  \node [style=green vertex] (10) at (-2, -1) {};
  \node [style=red vertex] (11) at (1.5, -1) {};
  \node [style=green vertex] (12) at (-4, -3) {};
  \node [style=green vertex] (13) at (-2, -2) {};
  \node [style=green vertex] (14) at (-4, -4) {};
  \node [style=boundary vertex] (15) at (1.5, -3) {};
  \draw  (11) to (15);
  \draw  (5) to (2);
  \draw  (12) to (14);
  \draw  (7) to (10);
  \draw  (5) to (8);
  \draw  (0) to (6);
  \draw  (3) to (4);
  \draw  (10) to (13);
  \draw  (1) to (3);
  \draw  (4) to (5);
  \draw  (6) to (9);
  \draw  (9) to (12);
  \draw  (4) to (7);
  \draw  (8) to (11);
  \draw  (6) to (7);
  \node [green angle] at (12) {$\pi,\{v_1\}$};
  \node [green angle] at (10) {$\pi,\{v_2\}$};
  \node [red angle] at (11) {$\pi,\{v_1\}$};
  \node [green angle] at (8) {$\pi,\{v_2\}$};
\end{tikzpicture}
}}

\newcommand{\teleportProofii}{\inline{%
\begin{tikzpicture}[quanto]
\node [style=boundary vertex] (0) at (-4, 3) {};
  \node [style=red vertex] (1) at (1.5, 3) {};
  \node [style=green vertex] (2) at (-2, 2) {};
  \node [style=green vertex] (3) at (-2, 1) {};
  \node [style=red vertex] (4) at (1.5, 1) {};
  \node [style=red vertex] (5) at (-4, 0) {};
  \node [style=green vertex] (6) at (-2, 0) {};
  \node [style=green vertex] (7) at (1.5, 0) {};
  \node [style=red vertex] (8) at (-4, -1) {};
  \node [style=green vertex] (9) at (-2, -1) {};
  \node [style=red vertex] (10) at (1.5, -1) {};
  \node [style=boundary vertex] (11) at (1.5, -3) {};
  \draw  (10) to (11);
  \draw  (4) to (1);
  \draw  (6) to (9);
  \draw  (4) to (7);
  \draw  (0) to (5);
  \draw  (2) to (3);
  \draw  (3) to (4);
  \draw  (8) to (5);
  \draw  (3) to (6);
  \draw  (7) to (10);
  \draw  (5) to (6);
  \node [red angle below] at (8) {$\pi,\{v_1\}$};
  \node [green angle] at (9) {$\pi,\{v_2\}$};
  \node [red angle] at (10) {$\pi,\{v_1\}$};
  \node [green angle] at (7) {$\pi,\{v_2\}$};
\end{tikzpicture}
}}

\newcommand{\teleportProofiii}{\inline{%
\begin{tikzpicture}[quanto]
  \node [style=boundary vertex] (0) at (-4, 3) {};
  \node [style=red vertex] (1) at (-4, 2) {};
  \node [style=green vertex] (2) at (-4, 1) {};
  \node [style=green vertex] (3) at (-2, -1) {};
  \node [style=red vertex] (4) at (-2, -2) {};
  \node [style=boundary vertex] (5) at (-2, -3) {};
  \draw  (4) to (5);
  \draw  (3) to (4);
  \draw [out=270, in=450] (1) to (2);
  \draw [out=270, in=450] (0) to (1);
  \draw [out=-90, looseness=0.75, in=90] (2) to (3);
  \node [red angle] at (1) {$\pi,\{v_1\}$};
  \node [red angle left] at (4) {$\pi,\{v_1\}$};
  \node [green angle] at (2) {$\pi,\{v_2\}$};
  \node [green angle left] at (3) {$\pi,\{v_2\}$};
\end{tikzpicture}
}}

\newcommand{\teleportProofiv}{\inline{%
\begin{tikzpicture}[quanto]
  \node [style=boundary vertex] (0) at (-4, 3) {};
  \node [style=red vertex] (1) at (-4, 2) {};
  \node [style=red vertex] (2) at (-2, -2) {};
  \node [style=boundary vertex] (3) at (-2, -3) {};
  \draw  (2) to (3);
  \draw [out=267, looseness=0.50, in=90] (1) to (2);
  \draw [out=270, in=450] (0) to (1);
  \node [red angle] at (1) {$\pi,\{v_1\}$};
  \node [red angle left] at (2) {$\pi,\{v_1\}$};
\end{tikzpicture}
}}

\newcommand{\teleportProofv}{\inline{%
\begin{tikzpicture}[quanto]
  \node [style=boundary vertex] (0) at (-4, 3) {};
  \node [style=boundary vertex] (1) at (-2, -3) {};
  \draw [out=270, looseness=0.50, in=90] (0) to (1);
\end{tikzpicture}
}}

\newcommand{\greenSpiderLHS}{%
  \begin{tikzpicture}[quanto]
    \upspider{green vertex}{spideri}{0,0}
    \downspider{green vertex}{spiderii}{0,-1.5}
    \draw [] (spideri) to (spiderii) ;
    \node [green angle] at (spideri) {$\alpha,\mathcal{U}$};
    \node [green angle] at (spiderii) {$\beta, \mathcal{U}$};
  \end{tikzpicture}
}

\newcommand{\greenSpiderRHS}{%
  \begin{tikzpicture}[quanto]
    \spider{green vertex}{a}{0,0}
    \node [green angle] at (a) {$\alpha + \beta,\mathcal{U}$};
  \end{tikzpicture}
}

\newcommand{\greenAntiLoopLHS}{%
  \begin{tikzpicture}[quanto]
    \downspider{green vertex}{c}{1.75,-1.0}
    \draw (c) to [in=40,out=140,loop] ();
    \node [green angle] at (c) {$\alpha,\mathcal{U}$};
  \end{tikzpicture}
}

\newcommand{\greenAntiLoopRHS}{%
  \begin{tikzpicture}[quanto]
    \downspider{green vertex}{c}{1.75,-1.0}
    \node [green angle] at (c) {$\alpha,\mathcal{U}$};
  \end{tikzpicture}
}

\newcommand{\greenIdentityLHS}{%
\begin{tikzpicture}[quanto]
\node [boundary vertex] (c) at (1.05,-3.525) {};
\node [boundary vertex] (a) at (1.05,-1.0) {};
\node [green vertex] (b) at (1.05,-2.25) {};
\draw [] (a) to (b);
\draw [] (b) to (c);
\node [green angle] at (b) {$0,\mathcal{U}$};
\end{tikzpicture}
}

\newcommand{\greenIdentityRHS}{%
\begin{tikzpicture}[quanto]
\node [boundary vertex] (c) at (1.05,-3.525) {};
\node [boundary vertex] (a) at (1.05,-1.0) {};
\draw [] (a) to (c);
\end{tikzpicture}
}

\newcommand{\greenCommutesLHS}{%
\begin{tikzpicture}[quanto]
\downspider{green vertex}{b}{2.5,-3.525}
\node [green vertex] (a) at (2.5,-2.25) {};
\node [boundary vertex] (e) at (2.5,-1.0) {};
\node [boundary vertex] (f) at (2.5,-4.75) {};
\draw [] (e) to (a);
\draw [] (a) to (b);
\draw [] (f) to (b);
\node [green angle] at (b) {$\beta,\mathcal{U}'$};
\node [green angle] at (a) {$\alpha,\mathcal{U}$};
\end{tikzpicture}
}

\newcommand{\greenCommutesRHS}{%
\begin{tikzpicture}[quanto]
\downspider{green vertex}{b}{2.5,-2.25}
\node [green vertex] (a) at (2.5,-3.525) {};
\node [boundary vertex] (e) at (2.5,-1.0) {};
\node [boundary vertex] (f) at (2.5,-4.75) {};
\draw [] (e) to (b);
\draw [] (b) to (a);
\draw [] (f) to (a);
\node [green angle] at ($(b)+(0,+0.2)$) {$\beta,\mathcal{U}'$};
\node [green angle] at ($(a)+(0,-0.4)$) {$\alpha,\mathcal{U}$};
\end{tikzpicture}
}

\newcommand{\piCommutesLHS}{%
\begin{tikzpicture}[quanto]
\node [boundary vertex] (d) at (1.05,-4.8) {};
\node [boundary vertex] (e) at (2.5,-4.8) {};
\node [green vertex] (b) at (1.75,-3.525) {};
\node [boundary vertex] (c) at (1.75,-1.0) {};
\node [red vertex] (a) at (1.75,-2.25) {};
\node [ellipses] (f) at ($ 0.5*(d) + 0.5*(e) + (0,0.4)$) {} ;
\draw [] (b) to (d);
\draw [] (a) to (b);
\draw [] (c) to (a);
\draw [] (b) to (e);
\node [green angle] at (b) {$\alpha,\mathcal{U}$};
\node [red angle] at (a) {$\pi,\mathcal{U}$};
\end{tikzpicture}
}

\newcommand{\piCommutesRHS}{%
\begin{tikzpicture}[quanto]
\node [red vertex] (c) at (1.175,-3.525) {};
\node [boundary vertex] (f) at (2.85,-4.8) {};
\node [green vertex] (a) at (2.0,-2.25) {};
\node [red vertex] (b) at (2.85,-3.525) {};
\node [boundary vertex] (e) at (1.175,-4.8) {};
\node [boundary vertex] (d) at (2.0,-1.0) {};
\node [ellipses] (g) at ($ 0.5*(f) + 0.5*(e) + (0,0.33)$) {} ;
\draw [] (c) to (e);
\draw [] (d) to (a);
\draw [] (b) to (f);
\draw [] (a) to (c);
\draw [] (a) to (b);
\node [red angle,left] at ($(c)-(0.8,0)$) {$\pi,\mathcal{U}$};
\node [green angle] at (a) {$-\alpha,\mathcal{U}$};
\node [red angle] at (b) {$\pi,\mathcal{U}$};
\end{tikzpicture}
}

\newcommand{\bialgebraLHS}{%
\begin{tikzpicture}[quanto]
\node [boundary vertex] (f) at (1.05,-1.0) {};
\node [boundary vertex] (h) at (2.5,-4.8) {};
\node [green vertex] (d) at (1.05,-2.25) {};
\node [boundary vertex] (g) at (1.05,-4.8) {};
\node [boundary vertex] (e) at (2.5,-1.0) {};
\node [red vertex] (c) at (1.05,-3.525) {};
\node [green vertex] (b) at (2.5,-2.25) {};
\node [red vertex] (a) at (2.5,-3.525) {};
\draw [] (e) to (b);
\draw [] (a) to (h);
\draw [] (c) to (g);
\draw [] (d) to (c);
\draw [] (b) to (a);
\draw [] (f) to (d);
\draw [] (b) to (c);
\draw [] (d) to (a);
\end{tikzpicture}
}

\newcommand{\bialgebraRHS}{%
\begin{tikzpicture}[quanto]
\node [red vertex] (e) at (1.75,-2.25) {};
\node [boundary vertex] (a) at (1.05,-1.0) {};
\node [green vertex] (f) at (1.75,-3.525) {};
\node [boundary vertex] (c) at (1.05,-4.8) {};
\node [boundary vertex] (d) at (2.5,-4.8) {};
\node [boundary vertex] (b) at (2.5,-1.0) {};
\draw [] (a) to (e);
\draw [] (f) to (c);
\draw [] (e) to (f);
\draw [] (f) to (d);
\draw [] (b) to (e);
\end{tikzpicture}
}

\newcommand{\copyingLHS}{%
\begin{tikzpicture}[quanto]
\node [boundary vertex] (d) at (1.05,-3.525) {};
\node [red vertex] (b) at (1.75,-1.0) {};
\node [green vertex] (a) at (1.75,-2.25) {};
\node [boundary vertex] (c) at (2.5,-3.525) {};
\node [ellipses] (e) at ($ 0.5*(d) + 0.5*(c) + (0,0.4)$) {} ;
\draw [] (a) to (c);
\draw [] (a) to (d);
\draw [] (b) to (a);
\node [green angle] at (a) {$\alpha , \mathcal{U}$};
\end{tikzpicture}
}

\newcommand{\copyingRHS}{%
\begin{tikzpicture}[quanto]
\node [red vertex] (b) at (1.05,-1.0) {};
\node [boundary vertex] (d) at (1.05,-2.25) {};
\node [boundary vertex] (c) at (2.5,-2.25) {};
\node [red vertex] (a) at (2.5,-1.0) {};
\node [ellipses] (e) at ($ 0.5*(d) + 0.5*(c) + (0,0.5)$) {} ;
\draw [] (a) to (c);
\draw [] (b) to (d);
\end{tikzpicture}
}

\newcommand{\hopfLawLHS}{%
\begin{tikzpicture}[quanto]
\downspider{red vertex}{a}{1.05,-3.525}
\upspider{green vertex}{b}{1.05,-2}
\draw [,bend left=30] (b) to (a);
\draw [,bend left=-30] (b) to (a);
\node [green angle] at (b) {$\alpha , \mathcal{U}$};
\node [red angle] at (a) {$\beta , \mathcal{U}'$};
\end{tikzpicture}
}

\newcommand{\hopfLawRHS}{%
\begin{tikzpicture}[quanto]
\downspider{red vertex}{a}{1.05,-3.525}
\upspider{green vertex}{b}{1.05,-2}
\node [green angle] at (b) {$\alpha , \mathcal{U}$};
\node [red angle] at (a) {$\beta , \mathcal{U}'$};
\end{tikzpicture}
}

\newcommand{\removeGreenLHS}{%
\begin{tikzpicture}[quanto]
\node [boundary vertex] (b) at (1.05,-4.8) {};
\node [boundary vertex] (c) at (2.5,-4.8) {};
\node [green vertex] (e) at (1.75,-3.525) {};
\node [hadamard vertex] (d) at (1.75,-2.25) {};
\node [boundary vertex] (a) at (1.75,-1.0) {};
\node [ellipses] (f) at ($ 0.5*(b) + 0.5*(c) + (0,0.4)$) {} ;
\draw [] (e) to (c);
\draw [] (a) to (d);
\draw [] (e) to (b);
\draw [] (d) to (e);
\node [green angle] at (e) {$\alpha,\mathcal{U}$};
\end{tikzpicture}
}

\newcommand{\removeGreenRHS}{%
\begin{tikzpicture}[quanto]
\node [boundary vertex] (a) at (1.75,-1.0) {};
\node [hadamard vertex] (f) at (1.05,-3.525) {};
\node [boundary vertex] (b) at (1.05,-4.8) {};
\node [boundary vertex] (c) at (2.5,-4.8) {};
\node [red vertex] (d) at (1.75,-2.25) {};
\node [hadamard vertex] (e) at (2.5,-3.525) {};
\node [ellipses] (g) at ($ 0.5*(b) + 0.5*(c) + (0,0.33)$) {} ;
\draw [] (f) to (b);
\draw [] (a) to (d);
\draw [] (d) to (f);
\draw [] (d) to (e);
\draw [] (e) to (c);
\node [red angle] at (d) {$\alpha,\mathcal{U}$};
\end{tikzpicture}
}

\newcommand{\HsquaredLHS}{%
\begin{tikzpicture}[quanto]
\node [boundary vertex] (a) at (1.05,-1.0) {};
\node [hadamard vertex] (c) at (1.05,-2.25) {};
\node [hadamard vertex] (d) at (1.05,-3.525) {};
\node [boundary vertex] (b) at (1.05,-4.8) {};
\draw [] (d) to (b);
\draw [] (a) to (c);
\draw [] (c) to (d);
\end{tikzpicture}
}

\newcommand{\HsquaredRHS}{%
\begin{tikzpicture}[quanto]
\node [boundary vertex] (c) at (1.05,-3.525) {};
\node [boundary vertex] (a) at (1.05,-1.0) {};
\draw [] (a) to (c);
\end{tikzpicture}
}

\newcommand{\greenSpiderLHSa}{%
  \begin{tikzpicture}[quanto]
    \upspider{green vertex}{spideri}{0,0}
    \downspider{green vertex}{spiderii}{0,-1.5}
    \draw [] (spideri) to (spiderii) ;
    \node [green angle] at (spideri) {$\alpha$};
    \node [green angle] at (spiderii) {$\beta$};
  \end{tikzpicture}
}

\newcommand{\greenSpiderRHSa}{%
  \begin{tikzpicture}[quanto]
    \spider{green vertex}{a}{0,0}
    \node [green angle] at (a) {$\alpha + \beta$};
  \end{tikzpicture}
}

\newcommand{\greenAntiLoopLHSa}{%
  \begin{tikzpicture}[quanto]
    \downspider{green vertex}{c}{1.75,-1.0}
    \draw (c) to [in=40,out=140,loop] ();
    \node [green angle] at (c) {$\alpha$};
  \end{tikzpicture}
}

\newcommand{\greenAntiLoopRHSa}{%
  \begin{tikzpicture}[quanto]
    \downspider{green vertex}{c}{1.75,-1.0}
    \node [green angle] at (c) {$\alpha$};
  \end{tikzpicture}
}

\newcommand{\greenIdentityLHSa}{%
\begin{tikzpicture}[quanto]
\node [boundary vertex] (c) at (1.05,-3.525) {};
\node [boundary vertex] (a) at (1.05,-1.0) {};
\node [green vertex] (b) at (1.05,-2.25) {};
\draw [] (a) to (b);
\draw [] (b) to (c);
\node [green angle] at (b) {$0$};
\end{tikzpicture}
}

\newcommand{\greenIdentityRHSa}{%
\begin{tikzpicture}[quanto]
\node [boundary vertex] (c) at (1.05,-3.525) {};
\node [boundary vertex] (a) at (1.05,-1.0) {};
\draw [] (a) to (c);
\end{tikzpicture}
}

\newcommand{\piCommutesLHSa}{%
\begin{tikzpicture}[quanto]
\node [boundary vertex] (d) at (1.05,-4.8) {};
\node [boundary vertex] (e) at (2.5,-4.8) {};
\node [green vertex] (b) at (1.75,-3.525) {};
\node [boundary vertex] (c) at (1.75,-1.0) {};
\node [red vertex] (a) at (1.75,-2.25) {};
\node [ellipses] (f) at ($ 0.5*(d) + 0.5*(e) + (0,0.4)$) {} ;
\draw [] (b) to (d);
\draw [] (a) to (b);
\draw [] (c) to (a);
\draw [] (b) to (e);
\node [green angle] at (b) {$\alpha$};
\node [red angle] at (a) {$\pi$};
\end{tikzpicture}
}

\newcommand{\piCommutesRHSa}{%
\begin{tikzpicture}[quanto]
\node [red vertex] (c) at (1.175,-3.525) {};
\node [boundary vertex] (f) at (2.85,-4.8) {};
\node [green vertex] (a) at (2.0,-2.25) {};
\node [red vertex] (b) at (2.85,-3.525) {};
\node [boundary vertex] (e) at (1.175,-4.8) {};
\node [boundary vertex] (d) at (2.0,-1.0) {};
\node [ellipses] (g) at ($ 0.5*(f) + 0.5*(e) + (0,0.33)$) {} ;
\draw [] (c) to (e);
\draw [] (d) to (a);
\draw [] (b) to (f);
\draw [] (a) to (c);
\draw [] (a) to (b);
\node [red angle,left] at ($(c)-(0.8,0)$) {$\pi$};
\node [green angle] at (a) {$-\alpha$};
\node [red angle] at (b) {$\pi$};
\end{tikzpicture}
}

\newcommand{\bialgebraLHSa}{%
\begin{tikzpicture}[quanto]
\node [boundary vertex] (f) at (1.05,-1.0) {};
\node [boundary vertex] (h) at (2.5,-4.8) {};
\node [green vertex] (d) at (1.05,-2.25) {};
\node [boundary vertex] (g) at (1.05,-4.8) {};
\node [boundary vertex] (e) at (2.5,-1.0) {};
\node [red vertex] (c) at (1.05,-3.525) {};
\node [green vertex] (b) at (2.5,-2.25) {};
\node [red vertex] (a) at (2.5,-3.525) {};
\draw [] (e) to (b);
\draw [] (a) to (h);
\draw [] (c) to (g);
\draw [] (d) to (c);
\draw [] (b) to (a);
\draw [] (f) to (d);
\draw [] (b) to (c);
\draw [] (d) to (a);
\end{tikzpicture}
}

\newcommand{\bialgebraRHSa}{%
\begin{tikzpicture}[quanto]
\node [red vertex] (e) at (1.75,-2.25) {};
\node [boundary vertex] (a) at (1.05,-1.0) {};
\node [green vertex] (f) at (1.75,-3.525) {};
\node [boundary vertex] (c) at (1.05,-4.8) {};
\node [boundary vertex] (d) at (2.5,-4.8) {};
\node [boundary vertex] (b) at (2.5,-1.0) {};
\draw [] (a) to (e);
\draw [] (f) to (c);
\draw [] (e) to (f);
\draw [] (f) to (d);
\draw [] (b) to (e);
\end{tikzpicture}
}

\newcommand{\copyingLHSa}{%
\begin{tikzpicture}[quanto]
\node [boundary vertex] (d) at (1.05,-3.525) {};
\node [red vertex] (b) at (1.75,-1.0) {};
\node [green vertex] (a) at (1.75,-2.25) {};
\node [boundary vertex] (c) at (2.5,-3.525) {};
\node [ellipses] (e) at ($ 0.5*(d) + 0.5*(c) + (0,0.4)$) {} ;
\draw [] (a) to (c);
\draw [] (a) to (d);
\draw [] (b) to (a);
\node [green angle] at (a) {$\alpha$};
\end{tikzpicture}
}

\newcommand{\copyingRHSa}{%
\begin{tikzpicture}[quanto]
\node [red vertex] (b) at (1.05,-1.0) {};
\node [boundary vertex] (d) at (1.05,-2.25) {};
\node [boundary vertex] (c) at (2.5,-2.25) {};
\node [red vertex] (a) at (2.5,-1.0) {};
\node [ellipses] (e) at ($ 0.5*(d) + 0.5*(c) + (0,0.5)$) {} ;
\draw [] (a) to (c);
\draw [] (b) to (d);
\end{tikzpicture}
}

\newcommand{\hopfLawLHSa}{%
\begin{tikzpicture}[quanto]
\downspider{red vertex}{a}{1.05,-3.525}
\upspider{green vertex}{b}{1.05,-2}
\draw [,bend left=30] (b) to (a);
\draw [,bend left=-30] (b) to (a);
\node [green angle] at (b) {$\alpha$};
\node [red angle] at (a) {$\beta$};
\end{tikzpicture}
}

\newcommand{\hopfLawRHSa}{%
\begin{tikzpicture}[quanto]
\downspider{red vertex}{a}{1.05,-3.525}
\upspider{green vertex}{b}{1.05,-2}
\node [green angle] at (b) {$\alpha$};
\node [red angle] at (a) {$\beta$};
\end{tikzpicture}
}

\newcommand{\removeGreenLHSa}{%
\begin{tikzpicture}[quanto]
\node [boundary vertex] (b) at (1.05,-4.8) {};
\node [boundary vertex] (c) at (2.5,-4.8) {};
\node [green vertex] (e) at (1.75,-3.525) {};
\node [hadamard vertex] (d) at (1.75,-2.25) {};
\node [boundary vertex] (a) at (1.75,-1.0) {};
\node [ellipses] (f) at ($ 0.5*(b) + 0.5*(c) + (0,0.4)$) {} ;
\draw [] (e) to (c);
\draw [] (a) to (d);
\draw [] (e) to (b);
\draw [] (d) to (e);
\node [green angle] at (e) {$\alpha$};
\end{tikzpicture}
}

\newcommand{\removeGreenRHSa}{%
\begin{tikzpicture}[quanto]
\node [boundary vertex] (a) at (1.75,-1.0) {};
\node [hadamard vertex] (f) at (1.05,-3.525) {};
\node [boundary vertex] (b) at (1.05,-4.8) {};
\node [boundary vertex] (c) at (2.5,-4.8) {};
\node [red vertex] (d) at (1.75,-2.25) {};
\node [hadamard vertex] (e) at (2.5,-3.525) {};
\node [ellipses] (g) at ($ 0.5*(b) + 0.5*(c) + (0,0.33)$) {} ;
\draw [] (f) to (b);
\draw [] (a) to (d);
\draw [] (d) to (f);
\draw [] (d) to (e);
\draw [] (e) to (c);
\node [red angle] at (d) {$\alpha$};
\end{tikzpicture}
}

\newcommand{\HsquaredLHSa}{%
\begin{tikzpicture}[quanto]
\node [boundary vertex] (a) at (1.05,-1.0) {};
\node [hadamard vertex] (c) at (1.05,-2.25) {};
\node [hadamard vertex] (d) at (1.05,-3.525) {};
\node [boundary vertex] (b) at (1.05,-4.8) {};
\draw [] (d) to (b);
\draw [] (a) to (c);
\draw [] (c) to (d);
\end{tikzpicture}
}

\newcommand{\HsquaredRHSa}{%
\begin{tikzpicture}[quanto]
\node [boundary vertex] (c) at (1.05,-3.525) {};
\node [boundary vertex] (a) at (1.05,-1.0) {};
\draw [] (a) to (c);
\end{tikzpicture}
}

\title{A graphical approach to measurement-based quantum computing}
\author{Ross Duncan}
\date{\today}

\begin{document}
\maketitle


\begin{abstract}
  Quantum computations are easily represented in the graphical
  notation known as the \zxcalculus, a.k.a.\! the red-green calculus.
  We demonstrate its use in reasoning about measurement-based quantum
  computing, where the graphical syntax directly captures the
  structure of the entangled states usesd to represent computations,
  and show that the notion of information flow within the entangled
  states gives rise to rewriting strategies for proving the correctness of
  quantum programs.
\end{abstract}

Quantum computation, at least for the finite dimensional systems
usually considered, lives in the setting of finite dimensional Hilbert
spaces.  Even ignoring the possibly enormous dimension of the spaces
involved, a Hilbert space is a very rich mathematical
environment which often hides the structure of the states and conceals
the behaviour of their
maps,  making it difficult to analyse quantum programs.  Can this
difficulty be circumvented? 

In this chapter, we will present an abstract formulation of quantum
theory, based on algebraic features present in the Hilbert space
theory, but making no reference to Hilbert spaces themselves.  The
reader will perhaps be unsurprised to learn that the tool of choice for this
reformulation of quantum mechanics is category theory, and in
particular the theory of symmetric monoidal categories (\smcs).

In a seminal paper \cite{AbrCoe:CatSemQuant:2004}, Abramsky and Coecke
introduced the notions of $\dag$-symmetric monoidal category (\dsmc)
and $\dag$-compact category, and, exploiting the fact that the category
of finite dimensional Hilbert spaces and linear maps 
(henceforth called \fdhilb) forms a $\dag$-compact category, gave a
high-level proof of correctness of the quantum teleportation
protocol \cite{BBCJW:1993:teleport}.  In so doing, they showed that
quantum protocols do not necessarily rely upon the full apparatus of
Hilbert spaces: a more abstract presentation of quantum mechanics can
suffice.  We will use such a high level presentation to analyse
measurement-based quantum programs.

As discussed earlier in this volume, $\dag$-compact categories admit a
graphical notation where the morphisms of the category are represented
by diagrams.  Sequential composition of morphisms is represented
by plugging together diagrams, and parallel composition (\ie the
tensor product) is represented by juxtaposition of diagrams.  The
crucial point, fully elaborated by Selinger \cite{Selinger:2009aa}, is
that the equations of the category are fully captured by homotopic
transformation of diagrams.  In other words, two morphisms are equal,
according to the axioms of \dsmcs, if and only if the corresponding
diagrams can be continuously deformed into each other.  Hence, by
transcribing a morphism into the graphical notation a large
amount of equational structure is incorporated directly into the
syntax of the  formal system.

The axioms of $\dag$-compact categories will not be sufficient,
however, to study the quantum systems of interest in this chapter.  We
will introduce more structure by choosing specific generators for the
category of diagrams, and imposing certain equations between diagrams
involving those generators.  These equations induce an equivalence
relation between diagrams based on substitution.  More concretely we
view each equation, say $L = R$, as a rewrite rule, and whenever we
find $L$ occurring as a subdiagram in some larger diagram $D$, we may
perform the rewrite $D[L] \to D[R]$, as shown in
Figure~\ref{fig:rewrite-subs-example}.  Since all diagrams are typed
this substitution is always possible.  We will treat diagrams and
rewriting informally here, but the interested reader can find a
detailed account in the paper of Dixon and Kissinger
\cite{Dixon:2011fk}.

\begin{figure}
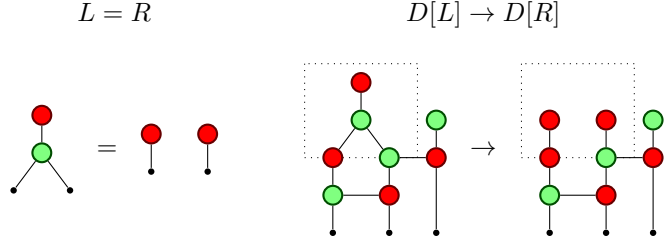

  \centering
  \[
  \begin{array}{ccc}
  L = R && D[L]\to D[R] \\ \\
 \inline{%
\beginpgfgraphicnamed{rw-example-i}
\InputIfFileExists{rw-example-i.tikz}{}{\input{./tikz/rw-example-i.tikz}}
\endpgfgraphicnamed} = \inline{%
\beginpgfgraphicnamed{rw-example-ii}
\InputIfFileExists{rw-example-ii.tikz}{}{\input{./tikz/rw-example-ii.tikz}}
\endpgfgraphicnamed}  
 && 
 \inline{%
\beginpgfgraphicnamed{rw-example-iiii}
\InputIfFileExists{rw-example-iiii.tikz}{}{\input{./tikz/rw-example-iiii.tikz}}
\endpgfgraphicnamed} \to \inline{%
\beginpgfgraphicnamed{rw-example-iv}
\InputIfFileExists{rw-example-iv.tikz}{}{\input{./tikz/rw-example-iv.tikz}}
\endpgfgraphicnamed}
  \end{array}
  \]
  \caption{Rewriting as substitution}
  \label{fig:rewrite-subs-example}
\end{figure}

The additional equations imposed on diagrams are those corresponding
to two different algebra structures found on the underlying Hilbert
space.  Thanks to the theorem of Coecke, Pavlovic and Vicary
\cite{PavlovicD:MSCS08}, there is a bijective correspondence between
orthonormal bases for the Hilbert space --- which from our point of
view represent quantum observables --- and special commutative
$\dag$-Frobenius algebras.  Therefore we encode each observable by an
algebra, and its associated equations give the first collection of
rewrite rules.  For the purpose of analysing measurement-based quantum
programs, we'll only need to consider two different algebras, namely
those corresponding to the $X$ and $Z$ spin observables.  These
observables have, in addition, a further property: they are
complementary, and in a particularly strong sense.  Intuitively,
complementarity means that perfect knowledge of one observable
implies complete ignorance of the other.  In previous work, Coecke and
the author \cite{Coecke:2009aa} showed that this kind of complementarity
can also be formalised in terms of algebras.  Strongly complementary
observables form a bialgebra, in fact a Hopf algebra.
We also impose the defining equations of these structures onto the
diagrammatic language to get another family of rewrite rules.

In summary, the graphical calculus consists of diagrams generated by
two Frobenius algebras for the $X$ and $Z$ observables; the equations
imposed by the monoidal structure may be effectively forgotten
because, thanks to the context of $\dag$-compact categories, 
the notion of equality of diagrams already contains all of them.
We then  impose rewrite rules corresponding to 
equations for the Frobenius structure, and then further equations
stating that these generators, when combined appropriately, form a
Hopf algebra.  This setup, combined with additional elements to be
introduced later, forms the \zxcalculus, also known as the red-green
calculus.

The \zxcalculus is known to be weaker than the full theory of Hilbert
spaces, but it is sound, meaning that any equation derived in it will
also hold when translated back to the Hilbert space formalism.
It replaces matrices with a structured and discrete notation
which exposes the relation between different parts of a quantum
system.  Furthermore, its graphical nature allows quantum circuits to
be represented very easily, and more importantly, the ad hoc notation
for graph states used in measurement-based quantum computation can be
derived from algebraic considerations alone.  Hence we can see the
beautiful interplay between the structure of an entangled state and
the algebraic objects which would represent this state.  The rewrite
rules then allow transformations between, \eg, the circuit model and the
measurement-based model, and expose how information flows within the
entangled state during the process of executing a measurement-based
program.  This last fact will at the heart of our analysis:  we will
demonstrate how the non-determinism induced by quantum measurements
may be tracked through the graph structure of an entangled state and
verify that a given computation is in fact deterministic.

\paragraph{Background and related work}
\label{sec:backgr-relat-work}

We assume that the reader is familiar with the  basics of monoidal
category theory; aside from other chapters in this volume, the
articles of Abramsky and Tzevelekos\cite{Abramsky:2009fx}, and Coecke
and Paquette \cite{BobEric2011cats} provide suitable introductions,
while \cite{MacLane:CatsWM:1971} is the standard text.  Compact closed
categories  were introduced by Kelly and Laplaza
\cite{KelLap:comcl:1980}, and the notion of dagger-compactness first
arose in \cite{AbrCoe:CatSemQuant:2004}.  Diagrammatic notation for
monoidal categories has a long history going back to work of Kelly
\cite{Kelly1972An-Abstract-App,Kelly1972Many-variable-f} and Penrose
\cite{Penrose1971Applications-of}; one can also view the proof-net
syntax of linear logic in this light
\cite{Girard:ProofNets:1996,Blute91naturaldeduction}.  The essential
reference on this subject is Selinger's survey \cite{Selinger:2009aa},
which pulls together a great deal of material scattered throughout the
literature.

The study of quantum mechanics through categorical eyes was initiated
by the paper of Abramsky and Coecke \cite{AbrCoe:CatSemQuant:2004},
and the explicit use of diagrams was emphasised by Coecke
\cite{Coecke:2009rt}.  The notion of classical structure, here
referred to as \emph{observable structure} was introduced by Coecke
and Pavlovic \cite{Coecke2006Quantum-Measure}, and further developed
by those authors in collaboration with Paquette
\cite{Coecke2006POVMs-and-Naima,Coecke:2009db}.  The key theorem, that
any classical structure in finite dimensional Hilbert space is
equivalent to a basis for that space, was shown by Coecke, Pavlovic and
Vicary \cite{PavlovicD:MSCS08}.  The that complementarity could also
be formalised in terms of interacting algebras was introduced by
Coecke and the author \cite{Coecke:2008nx,Coecke:2009aa}:  these
papers introduced several fundamental ideas which have since found
application in areas as diverse as quantum foundations
\cite{edwards-thesis-2009,Bob-Coecke:2008xr,billbobrob2011},
topological quantum computation \cite{Horsman:2011lr}, and
measurement-based quantum computing, which is our main concern here. 

One can view the quantum teleportation protocol
\cite{BBCJW:1993:teleport} as the first measurement-based quantum
computation; albeit the program computes the identity function.
Gottesman and Chuang showed later that this idea could be generalised
to a universal computation model \cite{Daniel-Gottesman:1999nx}.  The
model of interest here is not the teleportation model, but rather the
one-way model introduced by Raussendorf and Briegel
\cite{Raussendorf-2001,RauBri:OnewayQC:2003,raussendorf-2003-68}.  Our
work here is based on the work carried out by Danos, Kashefi and
Panangaden \cite{DanosV:meac} to provide this model with a formal
syntax, the measurement calculus.  Danos and Kashefi
\cite{DanKas:determinism:2005} introduced the concept of \emph{flow},
later renamed \emph{causal} flow to study the problem of determinism
in the one-way model.  Mhalla and Perdrix demonstrated an efficient
algorithm for finding optimal flows \cite{Mhalla:2008kx}, while
Browne, Kashefi, Mhalla, and Perdrix introduced the notion of
\emph{generalised flow} \cite{D.E.-Browne2007Generalized-Flo}.

This chapter mainly draws on the author's
joint work with Perdrix \cite{Duncan:2010aa}, although work relating
the graphical/categorical approach to \mbqc goes back rather further
\cite{Duncan:thesis:2006}.

\paragraph{Outline of the chapter}
\label{sec:outline-chapter}

The next section is a primer on the basics of quantum mechanics.
Section \ref{sec:observ-strong-compl} introduces the algebraic
framework of interacting observables (as presented in
\cite{Coecke:2009aa}) in full generality;  Section
\ref{sec:zxcalculus} presents the \zxcalculus, which is the specific
instance of that framework as a formal graphical calculus based on
qubits with the spin observables $X$ and $Z$.  Section
\ref{sec:graph-stat-meas} introduces the one-way model and the
measurement calculus, and shows how to represent them in the
\zxcalculus.  Finally, Section \ref{sec:determinism-flow} examines
determinism in the one-way model, and shows how to use the property of
flow to generate rewrite sequences to prove the correctness of the
measurement calculus programs.

\paragraph{Notation}
\label{sec:notation}

We will use the Dirac notation throughout: vectors are denoted by
\emph{kets} $\ket{\psi}$, and their duals by \emph{bras} $\bra{\phi}$.
The inner product is written $\innp{\phi}{\psi}$, and is taken to be
linear in the second component and anti-linear in the first.  We will
usually denote $\mathbb{C}^2$ by $Q$, since it is the state space of
\emph{qubits}.  The vectors comprising the standard basis for $Q$, sometime
called the or $Z$-basis, are written $\ket{0}$ and $\ket{1}$; we
denote the $X$-basis by
\[
\ket{+} = \frac{1}{\sqrt{2}}(\ket{0} + \ket{1})
\qquad\quad
\ket{-} = \frac{1}{\sqrt{2}}(\ket{0} - \ket{1})
\]
When writing tensor products of qubits we will usually suppress the
tensor symbol, and write \eg $\ket{00}$ in place of $\ket{0}\otimes
\ket{0}$.  The base field $\mathbb{C}$ will often be written simply as
$I$ since it is the unit of the monoidal category structure.

If $u$ and $v$ are vertices of an undirected graph, then $v \sim u$
means that they are adjacent.

When drawing diagrams, we use the pessimistic convention: diagrams
should be read from top to bottom.

\section{The rudiments of quantum computing}
\label{sec:rudim-quant-comp}

This section is necessarily rather brief; for a more complete
treatment we suggest the excellent books by Mermin
\cite{Mermin:2007fk}, and Kaye, Laflamme and Mosca \cite{Kaye:2007qy}.

Whereas a classical bit has only two values, its quantum analogue ---
the qubit --- is a unit vector in a two-dimensional Hilbert
space.  It is impossible to distinguish two states which differ only
by a global phase, so we quotient the state space by the relation
$\ket{\psi} \sim e^{i\alpha}\ket{\psi}$.
The state space of a compound quantum system, \ie one formed
by combining individual systems, is given by the tensor product of the
constituent state spaces, so a state consisting of $n$ qubits is a
vector in a Hilbert space of dimension $2^n$.  Such a state space
necessarily contains states that cannot be decomposed into a product
of $n$ individual qubits.  For example, the Bell state,
\[
\ket{\Phi_+} =  \frac{\ket{00}+\ket{11}}{\sqrt{2}}\;,
\] 
is a perfectly valid state of two qubits, but there are no single qubit
states such that $\ket{\phi}\otimes \ket{\psi} = \ket{\Phi_+}$.
This simple mathematical fact underlies the
the  phenomenon of \emph{entanglement}. These indecomposable states
reflect non-local correlations between the subsystems, and form the
main building block of the paradigm called \emph{measurement-based quantum
computation}.  

The evolution of an undisturbed quantum system is given by a unitary
operator:
\[
\ket{\psi(t)} = U_t \ket{\psi(t_0)}
\]
A quantum computation is typically described as a quantum circuit:
this is just a sequence of unitary operations acting on some number of
qubits.  Unitarity implies that quantum computations are reversible,
since the unitaries form a group.  The exception is quantum
measurement.

Quantum measurements have two properties which run contrary to classical
intuition.  Firstly their outcomes are \emph{probabilistic}:  in almost all
quantum states the outcome of a given measurement cannot be know with
certainty.  Secondly, they have \emph{side-effects}, so that state
after a measurement will usually be different to that before it.
Mathematically speaking we identify a quantum measurement with a
self-adjoint operator,
\[
M = \sum_i \lambda_i \outp{v_i}{v_i}\;.
\]
The possible observed values are the eigenvalues $\lambda_i$.  We are
only concerned with non-degenerate measurements here, so we assume
that all the $\lambda_i$ are distinct and non-zero.  Given a quantum
state $\ket{\psi}$, the probability of observing $\lambda_i$ is
given by the inner product 
\[
p(\lambda_i) = \sizeof{\innp{v_i}{\psi}}^2\;.
\]
Most importantly, the new state of the system after the measurement is
the corresponding eigenvector $\ket{v_i}$.  In other words, observing
$\lambda_i$ is effectively the same as acting on the state with the
projection operator $\outp{v_i}{v_i}$, or, if the measured
system is destroyed by the measurement---which will be the case for
the systems of interest here---simply $\bra{v_i}$, .  The actual values of the
measurements are not important here, so we will regard them simply as
labels for the outcomes.

If one part of an entangled quantum state is measured the effect of
that measurement can be observed in other parts of the state.  For
example, consider again the Bell state $\ket{\Phi_+}$. If its first
qubit is measured in the $\ket{0},\ket{1}$ basis then the two outcomes
are equally likely; suppose that $\ket{0}$ is observed.  The resulting
effect is to act on the joint state with the operator $\outp{0}{0}
\otimes \id{}$.  Ignoring normalisation, we have
\begin{multline*}
  (\outp{0}{0} \otimes \id{})\ket{\Phi_+} 
  = (\outp{0}{0} \otimes \id{})\ket{00} + (\outp{0}{0} \otimes
  \id{})\ket{11} 
  \\ 
  = \innp{0}{0}\ket{00} + \innp{0}{1}\ket{11} = \ket{00}\;.
\end{multline*}
Hence, now performing the same measurement on the second qubit will
produce outcome $\ket{0}$ with probability one.  Despite acting on
only one part of the system, we have produced a global change.
(Notice also that the new joint state is no longer entangled.)

This, in a nutshell, is the concept behind measurement-based quantum
computation: we begin with a large entangled state, and by performing
carefully chosen measurements upon it, the unmeasured parts are driven
toward the desired result.  As a first approximation, we could say
that the structure of entangled quantum states defines the desired
computation, while the measurements themselves function more to
``push'' information through this structure.  From this point of of
view, the measurements play a role similar to the evaluation maps in
functional programming, effectively ``applying'' their outcomes to the
function defined by the rest of the state.  (This is not entirely
accurate; as we shall see, the choice of measurements does play a role
in defining the computation.)  

It is frequently useful to generalise the notion of quantum state to
admit probabilistic mixtures of states.  In this setting, states
comprising a single state vector as
described above are called \emph{pure states}, while the others are
called \emph{mixed states}.  These more general states are represented
by \emph{density operators}: that is, trace one Hermitian matrices 
of the form
\[
\rho = \sum_i p_i \outp{\psi_i}{\psi_i}
\]
where $0 \leq p_i\leq 1$ and $\sum_i p_i = 1$.  Note that the
components $\ket{\psi_i}$ need not be orthogonal.  For pure states we
have $p_1 = 1$ and $p_i = 0$ for $i > 1$.  The decomposition
of a mixed state is not unique.  For example, the maximally mixed
qubit can arise by preparing either $\ket{0}$ or $\ket{1}$ with
equal probability, or equivalent by preparing either $\ket{+}$ and
$\ket{-}$:
\[
\frac{\outp{0}{0} + \outp{1}{1}}{2} = 
  \begin{pmatrix}
    1/2 & 0 \\ 0 & 1/2
  \end{pmatrix}
= \frac{\outp{+}{+} + \outp{-}{-}}{2}.
\]
The most general class of operations that are possible in the density
matrix formalism are \emph{completely positive maps} also called
\emph{superoperators}.  The action of such a map $\cal E$ on a state
$\rho$ is given by:
\[
\rho' = \mathcal{E}\rho\mathcal{E}^\dag.
\]
Since $\mathcal{E}$ is positive $\rho'$ is again a density matrix.
The most basic examples of superoperators are
unitary maps and quantum measurements.  For example, given a
measurement $M = \sum_i \lambda_i \outp{v_i}{v_i} = \sum_i \lambda_i P_i$,
the effect of performing the measurement on state $\rho$ is
\[
\rho' = \sum_i P_i  \rho P_i 
= \sum_{i,j} \mathrm{Tr}[\rho P_i] \outp{v_i}{v_i}
\]
where $\mathrm{Tr}[\rho P_i]$ gives the probability of observing
outcome $i$.  
While mixed states arise for a variety of reasons in quantum
computation, in this chapter the randomness introduced by measurement
will by the only source of uncertainty.

\section{Observables and strong complementarity}
\label{sec:observ-strong-compl}

\subsection{Observables and observable structures}
\label{sec:observ-observ-struct}

Given a quantum system whose state space is the Hilbert space $A$, we
will assume that any orthonormal basis $\{\ket{a_i}\}_i$ for $A$
defines an observable on that system.  Furthermore, these will be the only
observables of interest.

The no-cloning \cite{Wootters1982A-single-quantu} and no-deleting
\cite{Pati2000Impossibility-o} theorems state that it is impossible to
perfectly copy or erase an unknown quantum state.  However, it is
possible to perform both of these operations if the state is
guaranteed to be an outcome of a known observable; that is, if it is a
member of some given basis.  We can therefore view each quantum
observable as determining a classical data type, whose elements are possible
outcomes, and whose operations are copying and deleting.  For example,
the copying and deleting operations for the standard basis are given
by the linear maps

\begin{align*}
  \delta_Z : Q & \to Q \otimes Q &
  \epsilon_Z : Q & \to I
\\
  \delta_Z : \ket{i} & \mapsto \ket{ii} &
  \epsilon_Z : \ket{i}  & \mapsto 1
\end{align*}
Note that $\epsilon_Z$ is an unnormalised bra, namely $\sqrt{2}
\bra{+}$.  Graphically we will denote these operations by
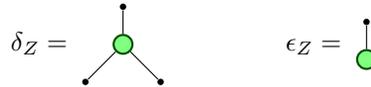
\[
\delta_Z = \greendelta
\qquad\qquad
\epsilon_Z = \greencounit
\]
What axioms should such operations obey?  Informally we may say that
if we copy something, and then copy one of the copies, it should not 
matter which copy we copied; that if we copy something and immediately
erase one of the copies, the combined operation should have no effect;
and, if we copy something, the two copies may be exchanged without
making any difference.  The same thing stated formally is that
$(\delta_Z,\epsilon_Z)$ should form a cocommutative comonoid on $Q$.
Presented graphically we have the following:
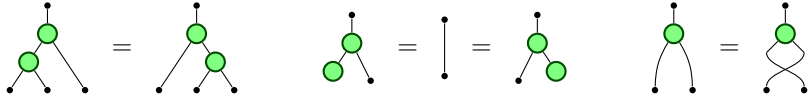
\[
{\inline{%
\beginpgfgraphicnamed{coassoc-lhs}
\InputIfFileExists{coassoc-lhs.tikz}{}{\input{./tikz/coassoc-lhs.tikz}}
\endpgfgraphicnamed}} = 
{\inline{%
\beginpgfgraphicnamed{coassoc-rhs}
\InputIfFileExists{coassoc-rhs.tikz}{}{\input{./tikz/coassoc-rhs.tikz}}
\endpgfgraphicnamed}}
\qquad
\inline{%
\beginpgfgraphicnamed{counit-lhs}
\InputIfFileExists{counit-lhs.tikz}{}{\input{./tikz/counit-lhs.tikz}}
\endpgfgraphicnamed} =
\inline{%
\beginpgfgraphicnamed{short-id}
\begin{tikzpicture}
	\begin{pgfonlayer}{nodelayer}
		\node [style=boundary vertex] (0) at (0, 0.75) {};
		\node [style=boundary vertex] (1) at (0, -0.75) {};
	\end{pgfonlayer}
	\begin{pgfonlayer}{edgelayer}
		\draw [style=plain] (0) to (1);
	\end{pgfonlayer}
\end{tikzpicture}}
\endpgfgraphicnamed} =
\inline{%
\beginpgfgraphicnamed{counit-lhs-bis}
\InputIfFileExists{counit-lhs-bis.tikz}{}{\input{./tikz/counit-lhs-bis.tikz}}
\endpgfgraphicnamed}
\qquad
{\inline{%
\beginpgfgraphicnamed{cocomm-lhs}
\InputIfFileExists{cocomm-lhs.tikz}{}{\input{./tikz/cocomm-lhs.tikz}}
\endpgfgraphicnamed}} = 
{\inline{%
\beginpgfgraphicnamed{cocomm-rhs}
\InputIfFileExists{cocomm-rhs.tikz}{}{\input{./tikz/cocomm-rhs.tikz}}
\endpgfgraphicnamed}}
\]
Since we operate in a $\dag$-category we may also consider the adjoint
operations
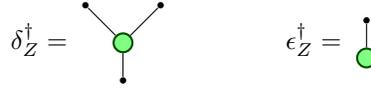
\[
\delta_Z^\dag = \greenmu
\qquad\qquad
\epsilon_Z^\dag = \greencounit
\]
which automatically form a commutative monoid, and obey the same
pictorial equations as $(\delta_Z ,\epsilon_Z)$ but flipped upside
down:
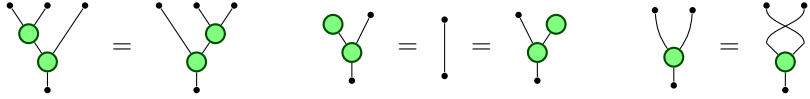
\[
{\inline{%
\beginpgfgraphicnamed{assoc-lhs}
\InputIfFileExists{assoc-lhs.tikz}{}{\input{./tikz/assoc-lhs.tikz}}
\endpgfgraphicnamed}} = 
{\inline{%
\beginpgfgraphicnamed{assoc-rhs}
\InputIfFileExists{assoc-rhs.tikz}{}{\input{./tikz/assoc-rhs.tikz}}
\endpgfgraphicnamed}}
\qquad
\inline{%
\beginpgfgraphicnamed{unit-lhs}
\InputIfFileExists{unit-lhs.tikz}{}{\input{./tikz/unit-lhs.tikz}}
\endpgfgraphicnamed} =
\inline{%
\beginpgfgraphicnamed{short-id}
\begin{tikzpicture}
	\begin{pgfonlayer}{nodelayer}
		\node [style=boundary vertex] (0) at (0, 0.75) {};
		\node [style=boundary vertex] (1) at (0, -0.75) {};
	\end{pgfonlayer}
	\begin{pgfonlayer}{edgelayer}
		\draw [style=plain] (0) to (1);
	\end{pgfonlayer}
\end{tikzpicture}}
\endpgfgraphicnamed} =
\inline{%
\beginpgfgraphicnamed{unit-lhs-bis}
\InputIfFileExists{unit-lhs-bis.tikz}{}{\input{./tikz/unit-lhs-bis.tikz}}
\endpgfgraphicnamed}
\qquad
{\inline{%
\beginpgfgraphicnamed{comm-lhs}
\InputIfFileExists{comm-lhs.tikz}{}{\input{./tikz/comm-lhs.tikz}}
\endpgfgraphicnamed}} = 
{\inline{%
\beginpgfgraphicnamed{comm-rhs}
\InputIfFileExists{comm-rhs.tikz}{}{\input{./tikz/comm-rhs.tikz}}
\endpgfgraphicnamed}}
\]
Notice that $\epsilon_Z^\dag$ is an unnormalised ket, $\sqrt{2}\ket{+}$.

Taken together, the 4-tuple
$(\delta_Z,\epsilon_Z,\delta_Z^\dag,\epsilon_Z^\dag)$ forms a special
commutative $\dag$-Frobenius algebra; this amounts to saying that in
addition to  the above, the following equations also hold:
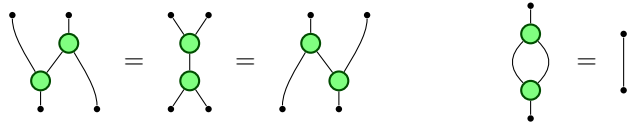
\[
\inline{%
\beginpgfgraphicnamed{frob-law-lhs}
\InputIfFileExists{frob-law-lhs.tikz}{}{\input{./tikz/frob-law-lhs.tikz}}
\endpgfgraphicnamed}
=
\inline{%
\beginpgfgraphicnamed{frob-law-rhs}
\InputIfFileExists{frob-law-rhs.tikz}{}{\input{./tikz/frob-law-rhs.tikz}}
\endpgfgraphicnamed}
=
\inline{%
\beginpgfgraphicnamed{frob-law-lhs-bis}
\InputIfFileExists{frob-law-lhs-bis.tikz}{}{\input{./tikz/frob-law-lhs-bis.tikz}}
\endpgfgraphicnamed}
\qquad\qquad
\inline{%
\beginpgfgraphicnamed{special-lhs}
\InputIfFileExists{special-lhs.tikz}{}{\input{./tikz/special-lhs.tikz}}
\endpgfgraphicnamed}
=
\inline{%
\beginpgfgraphicnamed{short-id}
\begin{tikzpicture}
	\begin{pgfonlayer}{nodelayer}
		\node [style=boundary vertex] (0) at (0, 0.75) {};
		\node [style=boundary vertex] (1) at (0, -0.75) {};
	\end{pgfonlayer}
	\begin{pgfonlayer}{edgelayer}
		\draw [style=plain] (0) to (1);
	\end{pgfonlayer}
\end{tikzpicture}}
\endpgfgraphicnamed}
\]
These equations, the Frobenius law and the special condition
respectively, will not be motivated here; see
\cite{Coecke2006POVMs-and-Naima,Coecke2006Quantum-Measure}.

The preceding discussion of algebras is motivated by the following theorem:

\begin{theorem}[\cite{PavlovicD:MSCS08}]\label{thm:frobs-bases}
  There is a bijective correspondence between orthonormal bases for a
  finite-dimensional Hilbert space $A$, and special commutative
  Frobenius algebras on $A$.
\end{theorem}

Since we have assumed that all observables are non-degenerate, Theorem
\ref{thm:frobs-bases} permits us to treat Frobenius algebras of the
above type as an abstract version of quantum observables.  This
definition  moreover makes no reference to the fact that the
underlying object is a Hilbert space, and hence it can be used in any
\dsmc.  For this reason, we will henceforward refer to special commutative
$\dag$-Frobenius algebras by the term \emph{observable
  structure}\footnote{The same object has also been called a \emph{classical
  structure} \cite{Coecke:2009db} and a \emph{basis structure} \cite{edwards-thesis-2009}}.

\begin{remark}\label{rem:observable-vs-measurement}
  Note that this representation of observables by algebras is not
  a representation of the \emph{measurement} of an observable.
  The additional formal apparatus required to  account for the
  non-deterministic aspect of the measurement will be introduced in
  Section \ref{sec:semantics}.
\end{remark}

\subsection{Unbiasedness and the phase group}
\label{sec:unbi-phase-group}

Given an orthonormal basis $\{\ket{a_i}\}_i$ for a $d$-dimensional
Hilbert space $A$, a vector $\ket{\psi}$ is unbiased for
$\{\ket{a_i}\}_i$ if for all $i$, we have $\sizeof{\innp{a_i}{\psi}} =
\frac{1}{\sqrt{d}}$.  For example, $\ket{+_\alpha} =
\frac{1}{\sqrt{2}}(\ket{0} + e^{i\alpha}\ket{1})$ is unbiased for the
standard basis on $Q$, for all values of $\alpha$;  indeed these are
the only unbiased states for the standard basis.  Incorporating this
concept into our diagrammatic language yields a surprising amount of
power.

Recall that in $\dag$-category  a morphism $f : A \to B$ is unitary if
$f^\dag\circ f = \id{A}$  and $f\circ f^\dag = \id{B}$;
diagrammatically this is written,
\[
\inline{%
\beginpgfgraphicnamed{f-fdag}
\InputIfFileExists{f-fdag.tikz}{}{\input{./tikz/f-fdag.tikz}}
\endpgfgraphicnamed} 
=
\inline{%
\beginpgfgraphicnamed{long-id}
\begin{tikzpicture}
	\begin{pgfonlayer}{nodelayer}
		\node [style=boundary vertex] (0) at (0, 1.75) {};
		\node [style=boundary vertex] (1) at (0, -1.75) {};
	\end{pgfonlayer}
	\begin{pgfonlayer}{edgelayer}
		\draw [style=plain] (0) to (1);
	\end{pgfonlayer}
\end{tikzpicture}}
\endpgfgraphicnamed}
=
\inline{%
\beginpgfgraphicnamed{fdag-f}
\InputIfFileExists{fdag-f.tikz}{}{\input{./tikz/fdag-f.tikz}}
\endpgfgraphicnamed}
\]
where the picture for $f^\dag$ is obtained by flipping the picture for
$f$ upside down.
This notion of unitarity agrees with usual one in \fdhilb. Now we can
make:

\begin{definition}\label{def:Lambda}
  Let $(\delta,\epsilon)$ be an observable structure on $A$, and let
  $\alpha:I\to A$ be a point of $A$.  Define a map $\Lambda(\alpha) :
  A\to A$ by
  \[
  \Lambda(\alpha) \vcentcolon = \delta^\dag \circ (\alpha \otimes
  \id{A}) = \inline{%
\beginpgfgraphicnamed{lambda-alpha}
\InputIfFileExists{lambda-alpha.tikz}{}{\input{./tikz/lambda-alpha.tikz}}
\endpgfgraphicnamed}
  \]
\end{definition}

\begin{definition}\label{def:unbiased}
  Let 
  $\alpha:I\to A$ be a point of $A$, and $\Lambda(\cdot)$ as in
  Definition~\ref{def:Lambda}.  We say $\alpha$ is unbiased for
  $(\delta,\epsilon)$ if $\Lambda(\alpha)$ is  unitary.
  \[
  \inline{%
\beginpgfgraphicnamed{lambda-alpha}
\InputIfFileExists{lambda-alpha.tikz}{}{\input{./tikz/lambda-alpha.tikz}}
\endpgfgraphicnamed} \text{ unbiased } 
  \quad \Longleftrightarrow \quad 
  \inline{%
\beginpgfgraphicnamed{lambda-alpha-lambda-alpha-dag}
\InputIfFileExists{lambda-alpha-lambda-alpha-dag.tikz}{}{\input{./tikz/lambda-alpha-lambda-alpha-dag.tikz}}
\endpgfgraphicnamed}  = \inline{%
\beginpgfgraphicnamed{long-id}
\begin{tikzpicture}
	\begin{pgfonlayer}{nodelayer}
		\node [style=boundary vertex] (0) at (0, 1.75) {};
		\node [style=boundary vertex] (1) at (0, -1.75) {};
	\end{pgfonlayer}
	\begin{pgfonlayer}{edgelayer}
		\draw [style=plain] (0) to (1);
	\end{pgfonlayer}
\end{tikzpicture}}
\endpgfgraphicnamed}
  \]
\end{definition}

When $\alpha$ is unbiased, the map $\Lambda(\alpha)$
is called a \emph{phase map} for $(\delta,\epsilon)$. 
According to the unit law for the monoid structure,
$\Lambda(\epsilon^\dag)$ yields the identity map, which is unitary.
Therefore every observable structure has at least one unbiased point, namely
$\epsilon^\dag$. Further since they are unitary, the phase maps form a
group, indeed an abelian group as the following calculation shows:
\[
\inline{%
\beginpgfgraphicnamed{phase-group-abelian-i}
\InputIfFileExists{phase-group-abelian-i.tikz}{}{\input{./tikz/phase-group-abelian-i.tikz}}
\endpgfgraphicnamed} =
\inline{%
\beginpgfgraphicnamed{phase-group-abelian-ii}
\InputIfFileExists{phase-group-abelian-ii.tikz}{}{\input{./tikz/phase-group-abelian-ii.tikz}}
\endpgfgraphicnamed} =
\inline{%
\beginpgfgraphicnamed{phase-group-abelian-iii}
\InputIfFileExists{phase-group-abelian-iii.tikz}{}{\input{./tikz/phase-group-abelian-iii.tikz}}
\endpgfgraphicnamed} =
\inline{%
\beginpgfgraphicnamed{phase-group-abelian-iv}
\InputIfFileExists{phase-group-abelian-iv.tikz}{}{\input{./tikz/phase-group-abelian-iv.tikz}}
\endpgfgraphicnamed}
\]
Since all the phase maps commute, we make the following notational
convention:
\[
\greenphase{\alpha} \coloneqq \inline{%
\beginpgfgraphicnamed{lambda-alpha}
\InputIfFileExists{lambda-alpha.tikz}{}{\input{./tikz/lambda-alpha.tikz}}
\endpgfgraphicnamed}\!,
\]
for which we have the equations
\[
\left(\greenphase{\alpha}\right)^\dag = \greenphase{-\alpha}
\qquad\text{and}\qquad
\twogreenphases{\alpha}{\beta} = \greenphase{\alpha+\beta}\!.
\]
The alert reader will have noted that the unbiased points themselves
form an abelian group, isomorphic to the phase group, under the
multiplication $\delta^\dag$;  one can equivalently define the phase group
via this route.

\begin{remark}\label{rem:phasegroup-or-monoid}
  While we will be exclusively interested in the case where
  $\alpha:I\to A$ is an unbiased point for some observable, most of
  the above still holds when $\alpha$ is an arbitrary point of $A$.
  In that case we get a commutative monoid rather than a group.  Note
  especially that Theorem~\ref{thm:gen-spider}, below, still applies.
\end{remark}

Returning to the example of $(\delta_Z,\epsilon_Z)$ on the qubit, the
vectors $\ket{+_\alpha}$, when multiplied by $\sqrt{2}$, yield the
phase maps 
\[
\Lambda_Z(\sqrt{2}\ket{+_\alpha}) = Z_\alpha  =
\begin{pmatrix*}
  1 & 0 \\ 0 & e^{i\alpha}
\end{pmatrix*}\;,
\]
comprising rotations around the $Z$ axis of the Bloch sphere.

We can now state the key theorem for the diagrammatic treatment of
observable structures and their phase groups:

\begin{theorem}[\cite{Coecke:2009aa}]\label{thm:gen-spider}
  Let $D$ be a connected diagram generated by an observable structure
  $(\delta,\epsilon)$ and its phase group;  then $D$ is determined
  completely by its number of inputs, its number of outputs, and the
  sum $\sum_i \alpha_i$ of phase group elements occurring in it.
\end{theorem}

The above result is effectively a normal form theorem for observable
structures, however we will use it instead to justify a new notational
convention, and simply collapse any connected diagram down to a single
vertex, which we refer to as a \emph{spider}:
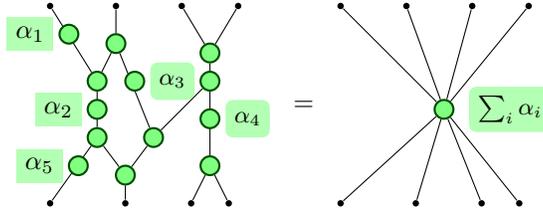
\[
\inline{%
\beginpgfgraphicnamed{complicated-spider}
\InputIfFileExists{complicated-spider.tikz}{}{\input{./tikz/complicated-spider.tikz}}
\endpgfgraphicnamed} = \inline{%
\beginpgfgraphicnamed{spider-sum}
\InputIfFileExists{spider-sum.tikz}{}{\input{./tikz/spider-sum.tikz}}
\endpgfgraphicnamed}
\]
The label will be omitted when $\alpha = 0$.  We can therefore adopt
spiders as the generators of the diagrammatic language, governed by a
single equational scheme, the spider rule:
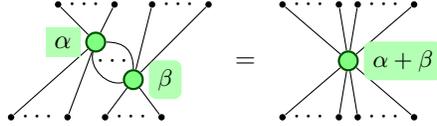
\[
\inline{%
\beginpgfgraphicnamed{spider-rule-lhs}
\InputIfFileExists{spider-rule-lhs.tikz}{}{\input{./tikz/spider-rule-lhs.tikz}}
\endpgfgraphicnamed} = \inline{%
\beginpgfgraphicnamed{spider-rule-rhs}
\InputIfFileExists{spider-rule-rhs.tikz}{}{\input{./tikz/spider-rule-rhs.tikz}}
\endpgfgraphicnamed}
\]

\begin{example}\label{ex:spider-identity}
  What is the spider with one input, one output, and $\alpha = 0$?
  The answer is provided by the unit law of the observable
  structure: it must be the identity, as shown below.
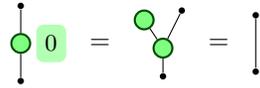
\[
\greenphase{0} = \inline{%
\beginpgfgraphicnamed{unit-lhs}
\InputIfFileExists{unit-lhs.tikz}{}{\input{./tikz/unit-lhs.tikz}}
\endpgfgraphicnamed} = \inline{%
\beginpgfgraphicnamed{short-id}
\begin{tikzpicture}
	\begin{pgfonlayer}{nodelayer}
		\node [style=boundary vertex] (0) at (0, 0.75) {};
		\node [style=boundary vertex] (1) at (0, -0.75) {};
	\end{pgfonlayer}
	\begin{pgfonlayer}{edgelayer}
		\draw [style=plain] (0) to (1);
	\end{pgfonlayer}
\end{tikzpicture}}
\endpgfgraphicnamed}
\]
\end{example}

Once we have made the above convention with respect to
the identity,  all the earlier equations are included in the spider rule,
hence this formulation is equivalent to the definition  in terms of
$\delta$, $\epsilon$ and $\Lambda(\alpha)$.

\begin{example}\label{ex:spider-compact}
  For any given observable structure $(\delta,\epsilon)$ we can
  produce a bipartite state $d:I \to A \otimes A$ by $d = \delta \circ
  \epsilon$:
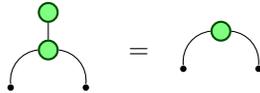
  \[
  \etaproofi = \etaproofii
  \]
  If we partially compose this state with its adjoint we obtain, via
  the spider rule, the identity:
  \[
  \inline{%
\beginpgfgraphicnamed{compact-proof-i}
\InputIfFileExists{compact-proof-i.tikz}{}{\input{./tikz/compact-proof-i.tikz}}
\endpgfgraphicnamed}
  =
  \inline{%
\beginpgfgraphicnamed{compact-proof-ii}
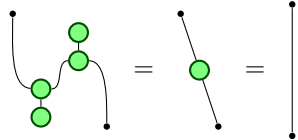
\begin{tikzpicture}
	\begin{pgfonlayer}{nodelayer}
		\node [style=green vertex] (0) at (0, 0) {};
		\node [style=boundary vertex] (1) at (-0.5, 1.5) {};
		\node [style=boundary vertex] (2) at (0.5, -1.5) {};
	\end{pgfonlayer}
	\begin{pgfonlayer}{edgelayer}
		\draw [style=plain] (1) to (0);
		\draw [style=plain] (0) to (2);
	\end{pgfonlayer}
\end{tikzpicture}}
\endpgfgraphicnamed}
  =
  \inline{%
\beginpgfgraphicnamed{long-id}
\begin{tikzpicture}
	\begin{pgfonlayer}{nodelayer}
		\node [style=boundary vertex] (0) at (0, 1.75) {};
		\node [style=boundary vertex] (1) at (0, -1.75) {};
	\end{pgfonlayer}
	\begin{pgfonlayer}{edgelayer}
		\draw [style=plain] (0) to (1);
	\end{pgfonlayer}
\end{tikzpicture}}
\endpgfgraphicnamed}
  \]
  Hence, the object $A$ bears a self-dual compact structure.  It is
  straight forward to construct observable structures for $A \otimes A$
  given one on $A$, so the monoidal category generated by $A$ is compact
  closed.
\end{example}

\subsection{Strong complementarity}
\label{sec:strong-compl}

In quantum theory, two observables are said to be \emph{complementary}
if measuring one of them reveals no information about the other, for
example the $X$ and $Z$ spins.
Notice that both elements of $X$ basis, $\ket{+}$ and $\ket{-}$, are
unbiased with respect to the $Z$ basis, and vice versa; these bases
are said to be \emph{mutually unbiased}.  Mutually unbiased bases
correspond to complementary observables: given an
eigenstate of $Z$, the inner product with either eigenstate of $X$ has
the same absolute value, and hence both outcomes are equiprobable when
an $X$ measurement is performed.  In this section we will present,
though not justify, a characterisation of complementarity in terms of
observable structures rather than bases.  In fact, we will present the
axioms for observable structures which are \emph{strongly
  complementary}, a property enjoyed by well-behaved pairs of
observables.  While the observables we are most interested in---the $X$
and $Z$ spins---are strongly complementary, the material of this
section is completely general;  the special features
of the $X$ and $Z$ observables are treated in the next section.

Since we are now dealing with two observables, we will have two
observable structures, $(\delta_g,\epsilon_g)$ and
$(\delta_r,\epsilon_r)$, which are represented by green and red
coloured spiders.  (For those reading without the benefit of colour,
green will appear as light grey, red as dark grey.)
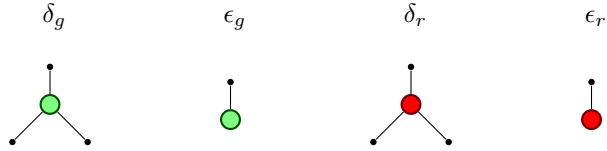
\[
\begin{array}{ccccccc}
  \delta_g & \qquad & \epsilon_g & \qquad & \delta_r & \qquad &
  \epsilon_r
  \\
  \\
  \greendelta && \greencounit && \reddelta && \redcounit
\end{array}
\]
The map $\delta$
was originally introduced as copying operation, however it was
axiomatised without any reference to the objects it copies.  We now
correct this.

\begin{definition}\label{def:classical-point}
  A point $k:I\to A$ is called \emph{classical} for an observable
  structure $(\delta,\epsilon)$ if it satisfies $\delta\circ k = k
  \otimes k$.
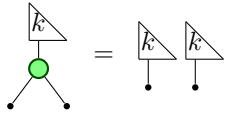
  \[
  \inline{%
\beginpgfgraphicnamed{classical-point-lhs}
\InputIfFileExists{classical-point-lhs.tikz}{}{\input{./tikz/classical-point-lhs.tikz}}
\endpgfgraphicnamed} = 
  \inline{%
\beginpgfgraphicnamed{classical-point-rhs}
\InputIfFileExists{classical-point-rhs.tikz}{}{\input{./tikz/classical-point-rhs.tikz}}
\endpgfgraphicnamed} 
  \]
\end{definition}

\noindent
In the language of coalgebras, classical points are called
\emph{set-like elements}.

\begin{lemma}\label{lem:classical-eigenpoints}
  Let $k$ denote a classical point for $(\delta_g,\epsilon_g)$, and
  let $\alpha$ be any  unbiased point for $(\delta_g,\epsilon_g)$; $k$
  is an eigenpoint of the corresponding phase map $\Lambda_g(\alpha)$.
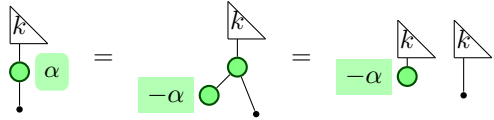
  \[
  \inline{%
\beginpgfgraphicnamed{eigenpt-i}
\InputIfFileExists{eigenpt-i.tikz}{}{\input{./tikz/eigenpt-i.tikz}}
\endpgfgraphicnamed} =
  \inline{%
\beginpgfgraphicnamed{eigenpt-ii}
\InputIfFileExists{eigenpt-ii.tikz}{}{\input{./tikz/eigenpt-ii.tikz}}
\endpgfgraphicnamed} =
  \inline{%
\beginpgfgraphicnamed{eigenpt-iii}
\InputIfFileExists{eigenpt-iii.tikz}{}{\input{./tikz/eigenpt-iii.tikz}}
\endpgfgraphicnamed}
  \]
\end{lemma}

Note the appearance of a scalar element here -- the eigenvalue of $k$.

\begin{remark}\label{rem:scalars}
  Any diagram with no inputs or outputs represents an arrow of type
  $I\to I$.  When interpreted in \fdHilb these are simply complex
  numbers.  Since quantum mechanics does not distinguish states that
  differ by a scalar factor we will ignore these whenever they
  appear.  Further, \emph{many of the equations presented below hold only up
  to scalar normalising factor}.  We omit these in order to simplify
  the presentation;  if needed they can easily be reconstructed.
\end{remark}

\begin{definition}\label{def:strong-comp-obs}
  Two observable structures $(\delta_g,\epsilon_g)$ and
  $(\delta_r,\epsilon_r)$ on $A$ are called \emph{strongly complementary} if
  \begin{enumerate}
  \item For every point $k:I \to A$, if $k$ is classical for
    $(\delta_g,\epsilon_g)$ then it is unbiased for
    $(\delta_r,\epsilon_r)$, and vice versa.\label{item:strong-comp-i}
  \item $\epsilon_r^\dag$ is classical for $(\delta_g,\epsilon_g)$ and
    $\epsilon_g^\dag$ is classical for $(\delta_r,\epsilon_r)$, \ie:
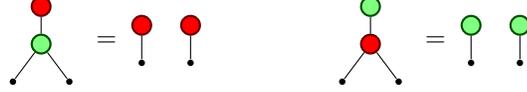
    \[
    \inline{%
\beginpgfgraphicnamed{green-copies-red-lhs}
\InputIfFileExists{green-copies-red-lhs.tikz}{}{\input{./tikz/green-copies-red-lhs.tikz}}
\endpgfgraphicnamed} = \redunit\,\redunit
      \qquad\qquad
    \inline{%
\beginpgfgraphicnamed{red-copies-green-lhs}
\InputIfFileExists{red-copies-green-lhs.tikz}{}{\input{./tikz/red-copies-green-lhs.tikz}}
\endpgfgraphicnamed} = \greenunit\,\greenunit
    \]\label{item:strong-comp-ii}
  \item The equation $(\delta_r^\dag \otimes \delta_r^\dag) \circ
    (\id{A} \otimes \sigma \otimes \id{A}) \circ (\delta_g \otimes
    \delta_g) = \delta_g \circ \delta_r^\dag$ holds, \ie:
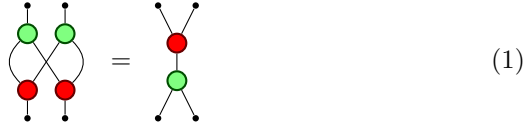
    \begin{equation}\label{eq:bialg-law}
      \inline{%
\beginpgfgraphicnamed{bialgebra-law-lhs}
\InputIfFileExists{bialgebra-law-lhs.tikz}{}{\input{./tikz/bialgebra-law-lhs.tikz}}
\endpgfgraphicnamed}
      =
      \inline{%
\beginpgfgraphicnamed{bialgebra-law-rhs}
\InputIfFileExists{bialgebra-law-rhs.tikz}{}{\input{./tikz/bialgebra-law-rhs.tikz}}
\endpgfgraphicnamed}
    \end{equation}
    where $\sigma$ denotes the symmetry of the monoidal
    structure. \label{item:strong-comp-iii}
  \end{enumerate}
\end{definition}
Since we are operating in \dsmc, conditions~\ref{item:strong-comp-ii}
and~\ref{item:strong-comp-iii} also imply flipped versions of the same
equations.  Given this, these two conditions could be replaced by a
unified condition:
\begin{itemize}
\item The 4-tuple $(\delta_g,\epsilon_g, \delta_r^\dag,
  \epsilon_r^\dag)$ forms a \emph{bialgebra} on $A$.
\end{itemize}
In fact, as well being a bialgebra, a pair of strongly complementary
observable structures is in addition a Hopf algebra.  Recall $d =
\delta \circ \epsilon^\dag$, and define $s :A\to 
A$ by $s = (d_r^\dag \otimes \id{A}) \circ (\id{A} \otimes d_g)$.  We
introduce a new element of the graphical notation for $s$:
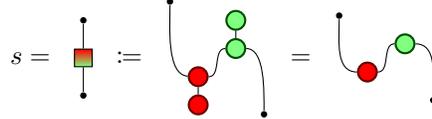
\[
 s = \inline{%
\beginpgfgraphicnamed{antipode-def}
\begin{tikzpicture}
	\begin{pgfonlayer}{nodelayer}
		\node [style=boundary vertex] (0) at (0, 1) {};
		\node [style=boundary vertex] (1) at (0, -1) {};
		\node [style=antipode] (2) at (0, 0) {};
	\end{pgfonlayer}
	\begin{pgfonlayer}{edgelayer}
		\draw [style=plain] (0) to (2);
		\draw [style=plain] (2) to (1);
	\end{pgfonlayer}
\end{tikzpicture}}
\endpgfgraphicnamed}
\coloneqq \inline{%
\beginpgfgraphicnamed{antipode-i}
\InputIfFileExists{antipode-i.tikz}{}{\input{./tikz/antipode-i.tikz}}
\endpgfgraphicnamed}
= \inline{%
\beginpgfgraphicnamed{antipode-ii}
\InputIfFileExists{antipode-ii.tikz}{}{\input{./tikz/antipode-ii.tikz}}
\endpgfgraphicnamed}
\]
Now we have:

\begin{lemma}\label{lem:hopf-law}
The 5-tuple $(\delta_g,\epsilon_g, \delta_r^\dag,
\epsilon_r^\dag,s)$ forms a Hopf algebra on $A$, \ie 
\begin{equation}\label{eq:hopf-law}
\begin{split}
  \delta_r^\dag \circ (s \otimes \id{A}) \circ \delta_g & =
  \epsilon_r^\dag \circ \epsilon_g.\\
  \inline{%
\beginpgfgraphicnamed{hopf-law-lhs}
\InputIfFileExists{hopf-law-lhs.tikz}{}{\input{./tikz/hopf-law-lhs.tikz}}
\endpgfgraphicnamed} & =   \inline{%
\beginpgfgraphicnamed{hopf-law-rhs}
\InputIfFileExists{hopf-law-rhs.tikz}{}{\input{./tikz/hopf-law-rhs.tikz}}
\endpgfgraphicnamed}
\end{split}
\end{equation}
\end{lemma}

\begin{remark}\label{rem:hopf-is-comp}
  We have stated Lemma \ref{lem:hopf-law} as a consequence of the
  bialgebra structure;  in fact, under a mild side condition, equation
  (\ref{eq:hopf-law}) can be  shown to be equivalent to condition
  \ref{item:strong-comp-i} of Definition \ref{def:strong-comp-obs}.
  See \cite{Coecke:2009aa} for full details.
\end{remark}

The classical points have some useful additional properties, which we
will now state;  the reader can find proofs in \cite{Coecke:2009aa}.

Thanks to Definition~\ref{def:strong-comp-obs}, if $k$ is classical
for $(\delta_g,\epsilon_g)$ then $\Lambda_r(k)$ is an element of the
phase group for the strongly complementary observable
$(\delta_r,\epsilon_r)$.  We draw the classical points in the colour
of the observable with respect to which they are unbiased, and rely on
the label to indicate that it is in fact a classical point: Latin
letters will indicate classical points, while Greek letters will
denote arbitrary unbiased points.

\begin{proposition}\label{prop:classical-points-properties}
  Let $k,k'$ be classical points for $(\delta_g,\epsilon_g)$, and let
  $h$ be classical for $(\delta_r,\epsilon_r)$; then:
  \begin{enumerate}
  \item The phase map $\Lambda_r(k)$ is a comonoid homomorphism of
    $(\delta_g,\epsilon_g)$:
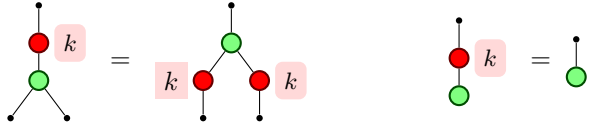
    \[
    \inline{%
\beginpgfgraphicnamed{classical-pt-comonoid-homo-i}
\InputIfFileExists{classical-pt-comonoid-homo-i.tikz}{}{\input{./tikz/classical-pt-comonoid-homo-i.tikz}}
\endpgfgraphicnamed}
    =
    \inline{%
\beginpgfgraphicnamed{classical-pt-comonoid-homo-ii}
\InputIfFileExists{classical-pt-comonoid-homo-ii.tikz}{}{\input{./tikz/classical-pt-comonoid-homo-ii.tikz}}
\endpgfgraphicnamed}
    \qquad\qquad
    \inline{%
\beginpgfgraphicnamed{classical-pt-comonoid-homo-iii}
\begin{tikzpicture}
	\begin{pgfonlayer}{nodelayer}
		\node [style=green vertex] (0) at (0, 0) {};
		\node [style=red vertex] (1) at (0, 1) {};
		\node [style=boundary vertex] (2) at (0, 2) {};
		\node [style=red angle] (3) at (0, 1) {$k$};
	\end{pgfonlayer}
	\begin{pgfonlayer}{edgelayer}
		\draw [style=plain] (0) to (1);
		\draw [style=plain] (2) to (1);
	\end{pgfonlayer}
\end{tikzpicture}}
\endpgfgraphicnamed}
    =
    \greencounit
    \]
  \item The phase maps $\Lambda_g(h)$ and $\Lambda_r(k)$ commute, up to
    a scalar factor:
    \[
    \inline{%
\beginpgfgraphicnamed{classical-phases-commute-lhs}
\InputIfFileExists{classical-phases-commute-lhs.tikz}{}{\input{./tikz/classical-phases-commute-lhs.tikz}}
\endpgfgraphicnamed} = 
    \inline{%
\beginpgfgraphicnamed{classical-phases-commute-rhs}
\InputIfFileExists{classical-phases-commute-rhs.tikz}{}{\input{./tikz/classical-phases-commute-rhs.tikz}}
\endpgfgraphicnamed}
    \]
  \item The point $\delta_r^\dag (k \otimes k')$ is also classical for
    $(\delta_g,\epsilon_g)$:
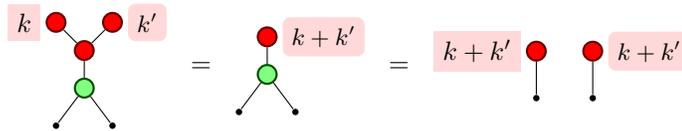
    \[
    \inline{%
\beginpgfgraphicnamed{classical-pts-closed-i}
\InputIfFileExists{classical-pts-closed-i.tikz}{}{\input{./tikz/classical-pts-closed-i.tikz}}
\endpgfgraphicnamed}
    =
    \inline{%
\beginpgfgraphicnamed{classical-pts-closed-ii}
\InputIfFileExists{classical-pts-closed-ii.tikz}{}{\input{./tikz/classical-pts-closed-ii.tikz}}
\endpgfgraphicnamed}
    =
    \inline{%
\beginpgfgraphicnamed{classical-pts-closed-iii}
\InputIfFileExists{classical-pts-closed-iii.tikz}{}{\input{./tikz/classical-pts-closed-iii.tikz}}
\endpgfgraphicnamed}
    \]
  \end{enumerate}
\end{proposition}

\begin{corollary}\label{cor:classical-subgroup}
  If $(\delta_g,\epsilon_g)$ has finitely many classical points, then
  they form a subgroup among the group of unbiased points of
  $(\delta_r,\epsilon_r)$.
\end{corollary}

When we consider observable structures over Hilbert spaces, having
finitely many classical points is just the statement the underlying
space is finite dimensional.  Since this is the case for all the
situations of interest for this chapter we will henceforth assume that
the \emph{classical phases} always form a subgroup.

There is also an important interaction between the classical points
and the phase group.  

\begin{proposition}\label{prop:classical-action-on-points}
  Let $k$ be a classical point for $(\delta_g,\epsilon_g)$,
  let $\alpha$ be an unbiased point for $(\delta_g,\epsilon_g)$, and
  define $k\bullet \alpha \vcentcolon= \Lambda_r(k) \circ \alpha$; then:
  \begin{enumerate}
  \item $k\bullet \alpha$ is again unbiased for
    $(\delta_g,\epsilon_g)$;
  \item $\Lambda_r(k) \circ \Lambda_g(\alpha) = \Lambda_g(k\bullet
    \alpha) \circ \Lambda_r(k)$. 
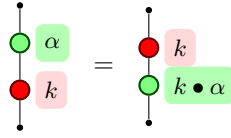
    \[
    \inline{%
\beginpgfgraphicnamed{classical-phases-action-lhs}
\InputIfFileExists{classical-phases-action-lhs.tikz}{}{\input{./tikz/classical-phases-action-lhs.tikz}}
\endpgfgraphicnamed}
    =     \inline{%
\beginpgfgraphicnamed{classical-phases-action-rhs}
\InputIfFileExists{classical-phases-action-rhs.tikz}{}{\input{./tikz/classical-phases-action-rhs.tikz}}
\endpgfgraphicnamed}
    \]
  \item $\Lambda_r(k)$ is a group automorphism of the unbiased points
    of $(\delta_g,\epsilon_g)$, and conjugation by $\Lambda_r(k)$ is
    an automorphism of the corresponding phase group.
  \end{enumerate}
\end{proposition}

To restate some of the preceding: the classical points of one
observable structure always form a subgroup among the unbiased points
of the strongly complementary observable; this subgroup in turn
acts as an automorphism group upon the unbiased points of the first
observable structure.

While all the preceding results will be (sometimes implicitly) used in
the subsequent sections, we will be able to make this all rather more
concrete by focusing on the specific case of the $Z$ and $X$ spin
observables.

\section{The \zxcalculus}
\label{sec:zxcalculus}

\subsection{The $Z$ and $X$ observables}
\label{sec:z-and-x-spins}

In this section, and in the rest of the chapter, we'll represent the
$Z$ and $X$ spin observables by following two strongly complementary
observable structures on $\mathbb{C}^2$:
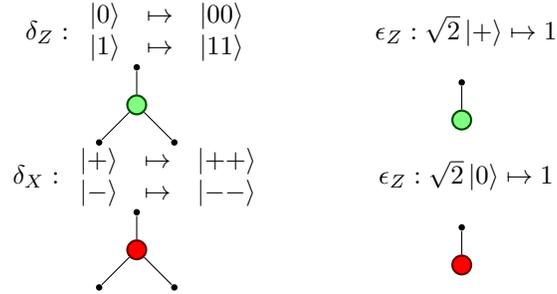
\[
\begin{array}{ccc}
\delta_Z :
\begin{array}{ccc}
  \ket{0} &\mapsto &\ket{00} \\
  \ket{1} &\mapsto &\ket{11} \\
\end{array}
& \qquad &
\epsilon_Z : \sqrt{2}\ket{+} \mapsto 1
\\
\greendelta && \greencounit 
\\
\delta_X :
\begin{array}{ccc}
  \ket{+} &\mapsto &\ket{++} \\
  \ket{-} &\mapsto &\ket{--} \\
\end{array}
& \qquad &
\epsilon_Z : \sqrt{2}\ket{0} \mapsto 1
\\
\reddelta
&&
\redcounit
\end{array}
\]
One of the most significant simplifications that occurs when working
with the $Z$ and $X$ observables is that they both generate the same
compact structure; that is, we have the equation
\[
\begin{array}{ccccc}
  d_Z & \quad = \quad & d_X & \quad = \quad & \ket{00} + \ket{11}\\
\greeneta &=& \redeta &=& \etapic
\end{array}
\]
Since there is no need to distinguish between a green or red cup (or
cap), we will drop the dot from diagrammatic notation whenever possible.
Since the category bears a single compact structure, 
we can treat the internal structure of any diagram as an
undirected graph and appeal to the principle of  diagrammatic
equivalence described earlier:  if two diagrams are
isomorphic as labelled graphs, they are equal.

In direct consequence, the antipode map of the Hopf algebra
structure is trivial: 
\begin{gather*}
  s  = (d_Z^\dag \otimes \id{A}) \circ (\id{A} \otimes d_X)  =
  \id{Q} 
  \\
  \inline{%
\beginpgfgraphicnamed{antipode-def}
\begin{tikzpicture}
	\begin{pgfonlayer}{nodelayer}
		\node [style=boundary vertex] (0) at (0, 1) {};
		\node [style=boundary vertex] (1) at (0, -1) {};
		\node [style=antipode] (2) at (0, 0) {};
	\end{pgfonlayer}
	\begin{pgfonlayer}{edgelayer}
		\draw [style=plain] (0) to (2);
		\draw [style=plain] (2) to (1);
	\end{pgfonlayer}
\end{tikzpicture}}
\endpgfgraphicnamed} 
  = \inline{%
\beginpgfgraphicnamed{antipode-ii}
\InputIfFileExists{antipode-ii.tikz}{}{\input{./tikz/antipode-ii.tikz}}
\endpgfgraphicnamed} 
  = \inline{%
\beginpgfgraphicnamed{cup-cap}
\InputIfFileExists{cup-cap.tikz}{}{\input{./tikz/cup-cap.tikz}}
\endpgfgraphicnamed}
  = \inline{%
\beginpgfgraphicnamed{short-id}
\begin{tikzpicture}
	\begin{pgfonlayer}{nodelayer}
		\node [style=boundary vertex] (0) at (0, 0.75) {};
		\node [style=boundary vertex] (1) at (0, -0.75) {};
	\end{pgfonlayer}
	\begin{pgfonlayer}{edgelayer}
		\draw [style=plain] (0) to (1);
	\end{pgfonlayer}
\end{tikzpicture}}
\endpgfgraphicnamed}
\end{gather*}
The defining equation of the Hopf algebra structure can therefore be
simplified:
\[
\inline{%
\beginpgfgraphicnamed{hopf-law-simple-lhs}
\InputIfFileExists{hopf-law-simple-lhs.tikz}{}{\input{./tikz/hopf-law-simple-lhs.tikz}}
\endpgfgraphicnamed}
=
\inline{%
\beginpgfgraphicnamed{hopf-law-rhs}
\InputIfFileExists{hopf-law-rhs.tikz}{}{\input{./tikz/hopf-law-rhs.tikz}}
\endpgfgraphicnamed}
\]

Recall that a point $\ket{\alpha_Z}$ is unbiased for the standard
basis $\ket{0},\ket{1}$ if and only if it has the form
$\frac{1}{\sqrt{2}}(\ket{0} + e^{i\alpha}\ket{1})$.  Hence the phase
group for the $Z$ observable is just the circle group, \ie, the
interval $[0,2\pi)$ under addition modulo $2\pi$.  The $X$
observable's phase group is isomorphic, so we represent the unbiased
points and phase maps as shown below.
\[
\greenphase{\alpha} \quad 
\greenpoint{\alpha'} \quad
\redphase{\beta} \quad 
\redpoint{\beta'} \quad
\alpha,\alpha',\beta,\beta' \in [0,2\pi)
\]
The phase maps are unitary, so the dagger sends each element to its
inverse, \ie. it negates the angle:
\[
\left( \greenphase{\alpha} \right)^\dag =
\greenphase{-\alpha}
\qquad
\left( \redphase{\alpha} \right)^\dag =
\redphase{-\alpha}
\]
Since we are operating in dimension 2, there are two classical points,
corresponding to the angles $0$ and $\pi$.  The action of the non-trivial
classical map is to negate the phase:
\[
\inline{%
\beginpgfgraphicnamed{red-pi-phases-action-lhs}
\InputIfFileExists{red-pi-phases-action-lhs.tikz}{}{\input{./tikz/red-pi-phases-action-lhs.tikz}}
\endpgfgraphicnamed} =
\inline{%
\beginpgfgraphicnamed{red-pi-phases-action-rhs}
\InputIfFileExists{red-pi-phases-action-rhs.tikz}{}{\input{./tikz/red-pi-phases-action-rhs.tikz}}
\endpgfgraphicnamed} \qquad
\inline{%
\beginpgfgraphicnamed{green-pi-phases-action-lhs}
\InputIfFileExists{green-pi-phases-action-lhs.tikz}{}{\input{./tikz/green-pi-phases-action-lhs.tikz}}
\endpgfgraphicnamed} =
\inline{%
\beginpgfgraphicnamed{green-pi-phases-action-rhs}
\InputIfFileExists{green-pi-phases-action-rhs.tikz}{}{\input{./tikz/green-pi-phases-action-rhs.tikz}}
\endpgfgraphicnamed}
\]

Given any two bases for a Hilbert space there is a unitary isomorphism
that maps one basis to the other; in the case of the $Z$ and $X$ bases
this map is the familiar Hadamard matrix.  We'll introduce an extra element
into the diagrammatic language to represent this map:
\[
H = \frac{1}{\sqrt{2}}
\begin{pmatrix*}[r]
  1 & 1 \\ 1 & -1
\end{pmatrix*}
=
\hgate
\]
The Hadamard is a self-adjoint unitary, hence we have the equations:
\[
\left( \hgate \right)^\dag  = \hgate 
\qquad \qquad
\inline{%
\beginpgfgraphicnamed{h-squared}
\InputIfFileExists{h-squared.tikz}{}{\input{./tikz/h-squared.tikz}}
\endpgfgraphicnamed} = \inline{%
\beginpgfgraphicnamed{long-id}
\begin{tikzpicture}
	\begin{pgfonlayer}{nodelayer}
		\node [style=boundary vertex] (0) at (0, 1.75) {};
		\node [style=boundary vertex] (1) at (0, -1.75) {};
	\end{pgfonlayer}
	\begin{pgfonlayer}{edgelayer}
		\draw [style=plain] (0) to (1);
	\end{pgfonlayer}
\end{tikzpicture}}
\endpgfgraphicnamed}
\]
Since the Hadamard maps one basis to the other, we could use it to
\emph{define} one observable structure in terms of the other:
\[
\inline{%
\beginpgfgraphicnamed{h-delta}
\InputIfFileExists{h-delta.tikz}{}{\input{./tikz/h-delta.tikz}}
\endpgfgraphicnamed} = \reddelta \qquad \inline{%
\beginpgfgraphicnamed{h-epsilon}
\begin{tikzpicture}
	\begin{pgfonlayer}{nodelayer}
		\node [style=hadamard vertex] (0) at (0, 0) {};
		\node [style=boundary vertex] (1) at (0, 1) {};
		\node [style=green vertex] (2) at (0, -1) {};
	\end{pgfonlayer}
	\begin{pgfonlayer}{edgelayer}
		\draw [style=plain] (0) to (2);
		\draw [style=plain] (1) to (0);
	\end{pgfonlayer}
\end{tikzpicture}}
\endpgfgraphicnamed} = \redcounit
\]
Strictly speaking, it is redundant to continue with both 
the $Z$ and $X$ observable structures: we could eliminate, for example,
the $X$ vertices.  However, for some purposes it is more convenient to
use both $Z$ and $X$, while for other purposes it is easier to work
with just $Z$ and $H$, so we maintain all three elements in the
syntax, and endorse the following \emph{colour duality principle}.

\begin{proposition}\label{prop:colour-duality}
  Every statement made in the diagrammatic language also holds with
  the colours reversed
\end{proposition}

We will switch freely between the two-coloured presentation, and the one-colour
and Hadamard view depending on which is most convenient at any given time.

\subsection{Syntax and semantics I}
\label{sec:syntax}

The \zxcalculus is a formal graphical notation, based on the notion of
an \emph{open graph}.

\begin{definition}\label{def:ograph}
  An \emph{open graph} is a triple $(G,I,O)$ consisting of an
  undirected graph $G =(V,E)$ and distinguished subsets $I,O \subseteq
  V$ of \emph{input} and \emph{output} vertices $I$ and $O$.  The set
  of vertices $I\cup O$ is called the \emph{boundary} of $G$, and
  $V\setminus (I\cup O)$ is the \emph{interior} of $G$.
\end{definition}

A term of the \zxcalculus is called a \emph{diagram}; this is an open
graph with some additional properties and structure.

\begin{definition}\label{def:diagram}
  A \emph{diagram} is an open graph $(G,I,O)$, where (i) all the
  boundary vertices are of degree one; (ii) the set of inputs $I$ and
  the set of outputs $O$ are both totally ordered; and (iii) whose
  interior vertices are restricted to the following types:
  \begin{itemize}
  \item $Z$ vertices with $m$ inputs and $n$ outputs, labelled by an
    angle $\alpha \in [0,2\pi)$; these are denoted $Z^m_n(\alpha)$,
    and shown graphically as (light) green circles,
  \item $X$ vertices with $m$ inputs and $n$ outputs, labelled by an
    angle $\alpha \in [0,2\pi)$; these are these are denoted
    $X^{m}_{n}(\alpha)$, and shown graphically as 
    (dark) red circles,
  \item $H$ (or Hadamard) vertices, restricted to degree 2; shown as
    squares.
  \end{itemize}
  If a $X$ or $Z$ vertex has $\alpha = 0$ then the label is entirely
  omitted.  The allowed vertices are shown in
  Figure~\ref{fig:components}.
\end{definition}
Since the inputs and outputs of of a diagram are totally ordered, we
can identify them with natural numbers and speak of the $k$th input, etc.

\begin{figure}[tb]
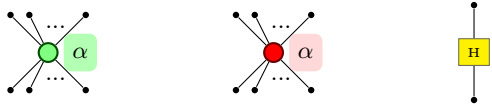

  \centering
  \greenspider{$\alpha$}\hspace{1cm}
  \redspider{$\alpha$}\hspace{1cm}
  \hgate
  \caption{Interior vertices of diagrams}
  \label{fig:components}
\end{figure}

\begin{remark}\label{rem:outputs-and-inputs}
  When a vertex occurs inside the graph, the distinction between
  inputs and outputs is purely conventional: one can view them simply
  as vertices of degree $n+m$; however, this distinction allows the
  semantics to be stated more directly, see below.
\end{remark}

The collection of diagrams forms a compact category in the
obvious way: the objects are natural numbers and the arrows $m\to n$
are those diagrams with $m$ inputs and $n$ outputs; composition
$g\circ f$  is formed by identifying the inputs of $g$ with the
outputs of $f$ and erasing the corresponding vertices; $f\otimes g$ is
the diagram formed by the disjoint union of $f$ and $g$ with $I_f$
ordered before $I_g$, and similarly for the outputs.  This is
basically the free (self-dual) compact category generated by the
arrows shown in Figure~\ref{fig:components}.

We can make this category $\dag$-compact by specifying that $f^\dag$
is the same diagram as $f$, but with the inputs and outputs
exchanged, and all the angles negated.

This construction yields a category that does not incorporate the
algebraic structure of strongly complementary observables.  To obtain
the desired category  we must quotient by the equations shown in
Figure~\ref{fig:rewrite-rules}.  We denote the category  so-obtained
by $\DD$.

\begin{figure}[th!b]
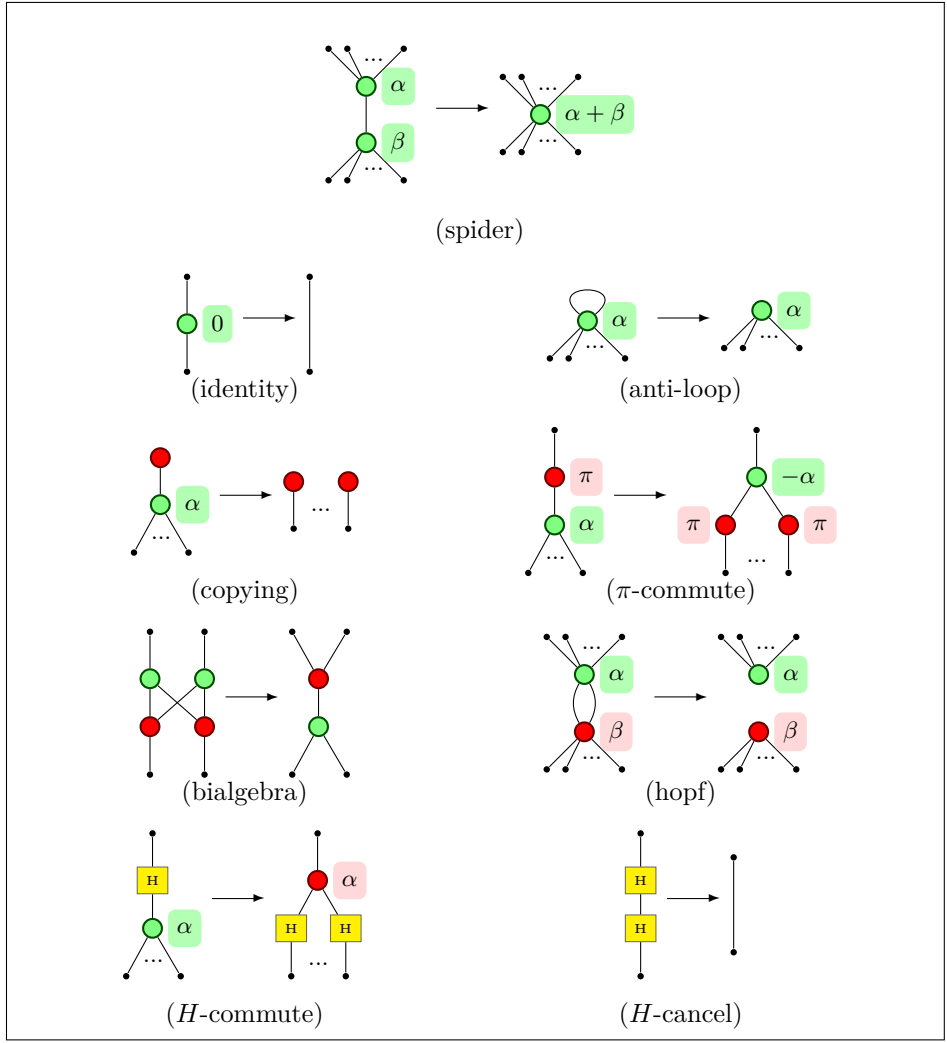

  \centering
  \fbox{
  \begin{minipage}[c]{1.0\linewidth}
  \centering
\[
\greenSpiderLHSa
\rTo
\greenSpiderRHSa
\]
(spider)
\[
\begin{array}{ccc}
\greenIdentityLHSa
\rTo
\greenIdentityRHSa
& 
\qquad \quad 
&
\greenAntiLoopLHSa
\rTo
\greenAntiLoopRHSa
\\
\text{(identity)}&&\text{(anti-loop)}
\\ \\[-1ex]
\copyingLHSa
\rTo
\copyingRHSa
&
\qquad \qquad
&
\piCommutesLHSa
\rTo
\piCommutesRHSa
\\ 
\text{(copying)}  &&  \text{($\pi$-commute)}
\\ \\[-1ex]
\bialgebraLHSa
\rTo
\bialgebraRHSa
&
\qquad \qquad
&
\hopfLawLHSa
\rTo
\hopfLawRHSa
\\
\text{(bialgebra)} &&\text{(hopf)} 
\\ \\[-1ex]
\removeGreenLHSa
\rTo
\removeGreenRHSa
& 
\qquad \qquad
&
\HsquaredLHSa
\rTo
\HsquaredRHSa
\\ \\[-1ex]
\text{($H$-commute)}  &&\text{($H$-cancel)} 
\end{array}
\]
\end{minipage}}
\caption{Rewrite rules for the \zxcalculus.  We present the rules for the
  $Z$ subsystem;  to obtain the complete set of rules exchange the
  colours in the rules shown above.
}
\label{fig:rewrite-rules}
\end{figure}

\begin{remark}\label{rem:comment-on-rules}
  The equations shown in Figure~\ref{fig:rewrite-rules} are not
  exactly those described in Sections~\ref{sec:observ-strong-compl}
  and \ref{sec:z-and-x-spins}, however they are equivalent to them.
  We shall therefore, on occasion, use properties discussed earlier as
  derived rules in computations.
\end{remark}

Since \DD is a monoidal category we can assign an interpretation to
any diagram by providing a monoidal functor from \DD to any other
monoidal category.  Since we interested in quantum mechanics, the
obvious target category  is \fdhilb.

\begin{definition}\label{def:semantic-pt1}
    Let $\denote{\cdot} : \DD \to \fdHilb$ be a symmetric monoidal
  functor defined on objects by
  \[
  \denote{1} = \mathbb{C}^2
  \]
  and on the generators by:
  \[
  \begin{array}{rcccl}
    \denote{Z^m_n(\alpha)} & = & \denote {{\greenspider{$\alpha$}}}
    & = &
    \left\{
    \begin{array}{ccl}
      \ket{0}^{\otimes m} & \mapsto &   \ket{0}^{\otimes n}\\
      \ket{1}^{\otimes m} & \mapsto &   e^{i \alpha}\ket{1}^{\otimes n}
    \end{array}\right.\,,
\\
    \denote{X^m_n(\alpha)} & = & \denote{{\redspider{$\alpha$}}}
    &=& 
    \left\{
    \begin{array}{ccl}
      \ket{+}^{\otimes m} & \mapsto &   \ket{+}^{\otimes n}\\
      \ket{-}^{\otimes m} & \mapsto &   e^{i \alpha}\ket{-}^{\otimes n}
    \end{array}\right.\,,
    \\
    \denote{H} & = & \denote{{\hgate}}
    &=& 
    \frac{1}{\sqrt{2}}
    \left(
      \begin{array}{cc}
        1&1\\1&-1
      \end{array}
    \right).
  \end{array}  
\]
The value of $\denote{\cdot}$ on all other objects and arrows is then
fixed by the requirement that it be a symmetric monoidal
functor\footnote{The full details of this construction
  regarding cyclic graphs and traces can be found in
  \cite{Duncan:thesis:2006}.}.
\end{definition}

\begin{theorem}[Soundness]\label{thm:soundness-i}
  For any diagrams $D$ and $D'$ in \DD, if $D = D'$ then $\denote{D} =
  \denote{D'}$ in \fdhilb.
  \begin{proof}
    Notice that the compact closed structure is preserved
    automatically because $\denote{\cdot}$ is a monoidal functor.  It
    just remains to check that all the equations of
    Figure~\ref{fig:rewrite-rules} hold in the image of
    $\denote{\cdot}$.
  \end{proof}
\end{theorem}

\begin{remark}\label{rem:not-complete}
  While Theorem~\ref{thm:soundness-i} shows that every equation
  provable in the \zxcalculus is true in Hilbert spaces, the converse
  does not hold:  there are diagrams $D$ and $D'$ such that
  $\denote{D} = \denote{D'}$ but the equation $D = D'$ cannot be
  derived from the rules of the calculus.  See \cite{Duncan:2009ph}
  for details.
\end{remark}

\begin{proposition}\label{prop:dag-preserved}
  For any diagram $D$ in \DD, we have $\denote{D^\dag} = \denote{D}^\dag$.
\end{proposition}

\subsection{Quantum circuits}
\label{sec:quantum-circuits}

The quantum circuit model is simple and intuitive quantum
computational model.  Analogous to traditional Boolean circuits, a
quantum circuit consists of a register of qubits, to which quantum
logic gates---that is one- or two-qubit unitary operations---are
applied, in sequence and in parallel\footnote{The circuit model
  usually incorporates measurements too, but this will not be necessary
  here.  See, \eg, chapter 4 of
  {\cite{NieChu:QuantComp:2000}}.}.  A fairly typical set of logic
gates is shown in Figure~\ref{fig:quantum-logic-gates}, however these
are not all necessary, as the following theorem states.

\begin{figure}[th]
  \centering
  \[
  \begin{array}{ccc}
    Z =
    \begin{pmatrix*}[r]
      1 & 0 \\ 0 & -1
    \end{pmatrix*}
& \qquad &
    X =
    \begin{pmatrix*}[r]
      0 & 1 \\ 1 & 0
    \end{pmatrix*}
\\ \\
    Z_\alpha =
    \begin{pmatrix*}[r]
      1 & 0 \\ 0 & e^{i\alpha}
    \end{pmatrix*}
& \qquad &
    H =\frac{1}{\sqrt{2}}
    \begin{pmatrix*}[r]
      1 & 1 \\ 1 & -1
    \end{pmatrix*}
\\ \\
\CX = \begin{pmatrix*}[r]
      1 & 0 & 0 & 0  \\
      0 & 1 & 0 & 0  \\
      0 & 0 & 0 & 1  \\
      0 & 0 & 1 & 0 
    \end{pmatrix*}
&&
\CZ = \begin{pmatrix*}[r]
      1 & 0 & 0 & 0  \\
      0 & 1 & 0 & 0  \\
      0 & 0 & 1 & 0  \\
      0 & 0 & 0 & -1 
    \end{pmatrix*}
  \end{array}
  \]
  \caption{Quantum logic gates}
  \label{fig:quantum-logic-gates}
\end{figure}

\begin{theorem}[\cite{Adriano-Barenco:1995qy}]\label{thm:divincenzo-universality}
The set $\{Z_\alpha,H,\CX\}$ suffices to generate all unitary matrices
on $Q^n$.  
\end{theorem}

\begin{corollary}\label{cor:zx-is-universal}
  The \zxcalculus can represent all unitary matrices on $Q^n$.
  \begin{proof}
    It suffices to show that there are \zxcalculus terms for the
    matrices $Z_\alpha$, $H$ and $\CX$.  We have
    \[
    \denote{\;\hgate\;} = H, 
    \qquad 
    \denote{\;\greenphase{\alpha}\;} = Z_\alpha
    \quad\text{and}\quad
    \denote{\;\cex\;} = \CX 
    \]
    which can be verified by direct calculation.  Note that 
    \[
    \denote{\;\ctrlXskewed\;} \quad  =  \quad \denote{\;\ctrlXskewedbis\;}
    \]
    so the presentation of \CX is unambiguous.
  \end{proof}
\end{corollary}

\begin{example}[The \CZ-gate]\label{ex:CZ}
  The \CZ-gate can be obtained by using a Hadamard ($H$) gate to transform
  the second qubit of a \CX gate.  We obtain a simpler representation
  using the colour-change rule
  \[
  \CZviaCX \quad = \quad \czed
  \]
  From the presentation of \CZ in the \zxcalculus, we can immediately
  read off that it is symmetric in its inputs.  Furthermore, we can
  prove one of the basic properties of the \CZ gate, namely that is
  self-inverse.
  \[
  \CZselfinvi =
  \CZselfinvii =
  \CZselfinviii =
  \CZselfinviv =
  \CZselfinvv =
  \CZselfinvvi
  \]
\end{example}

\begin{example}[Bell state]\label{ex:preparing-bell-state}
  The following is a \zxcalculus term representing a quantum circuit
  which produces a Bell state, $\ket{00}+\ket{11}$.  We can verify
  this fact by the equations of the calculus.
  \[
  \makeBellPfi =
  \makeBellPfii =
  \makeBellPfiii =
  \etapic
  \]
  The corresponding \zxcalculus derivation is a proof of
  the correctness of this circuit.
\end{example}

The \zxcalculus can represent many things which do not correspond to
quantum circuits.  We now present a criterion to recognise which
diagrams do correspond to quantum circuits.

\begin{definition}\label{def:circuit-like-diagram}
    A diagram is called \emph{circuit-like} if: 
  \renewcommand{\theenumi}{(C\arabic{enumi})}
  \begin{enumerate}
  \item \label{circit-like:1}all of its vertices can be covered by a
    set $\mathcal{P}$ of disjoint directed paths, each of which ends in
    an output;
  \item \label{circuit-like:2}for every oriented cycle $\gamma$ in the
    diagram, if $\gamma$ contains at least 2 edges from different
    paths in $\mathcal{P}$, then it traverses at least one of them in
    the direction opposite to that induced by the path; and,
  \item \label{circuit-like:3} it is a simple graph, and is 
    3-coloured.
  \end{enumerate}
  \renewcommand{\labelenumi}{\theenumi.}
\end{definition}

Intuitively, the paths of $\mathcal{P}$ represent the trajectories of
the individual qubits through the circuit, whereas those edges not
included in any path represent two-qubit gates.  Condition
\ref{circuit-like:2} guarantees that the diagram has a causally
consistent temporal order, while condition \ref{circuit-like:3} forces
the diagram to be minimal with respect to  the spider, anti-loop, and Hopf
rules.

\begin{remark}\label{rem:circuit-like-with-inputs}
  Definition~\ref{def:circuit-like-diagram} requires each path to end
  in an output vertex, but does not demand that the initial vertex is
  an input.  This allows the representation of  quantum circuits with
  some or all inputs fixed; see Example~\ref{ex:preparing-bell-state}
  above.
\end{remark}

\begin{example}\label{ex:not-circuit-like}
The following diagram is not circuit-like, since  condition
\ref{circuit-like:2} fails.
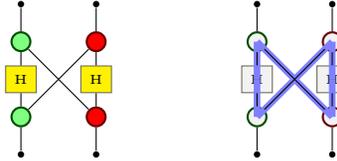
\[
\exnotcircuitlike \qquad\qquad \exnotcircuitlikebis
\]
There is only one possible path covering of this diagram;  the
indicated cycle runs contrary to the path.
\end{example}

\begin{example}\label{ex:really-circuit-like}
  The following diagram is circuit-like, and, as shown, is equivalent
  to something which clearly \emph{looks} like a circuit.
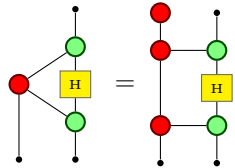
\[
\excircuitlikei = \excircuitlikeii
\]
Of course, the right-hand diagram is \emph{not} circuit-like because
it does not satisfy condition \ref{circuit-like:3}; indeed, it reduces
to the left-hand diagram by two applications of the spider rule.
\end{example}

As the preceding example shows, the definition  of 
circuit-like---in particular the technical 
third condition---is rather stronger than strictly necessary to ensure
that a diagram corresponds to a valid circuit.  For example, it forces
the two-qubit gates to be \CX rather than \CZ.  However it is not hard
to prove the following result:

\begin{proposition}
  \label{prop:circuit-like-implies-unitary-embedding}
  If $D:n\to m$ is a circuit-like diagram, then $\denote{D}$ is a
  unitary embedding $Q^n \to Q^m$; conversely, if $\denote{D}$ is a unitary
  embedding, then the exists some circuit-like $D'$ such that  $D =
  D'$ by the rules of the \zxcalculus.
\end{proposition}

\subsection{Syntax and semantics II: measurements}
\label{sec:semantics}

While the version of the \zxcalculus we have presented so far can
represent the projection onto some measurement outcome, for example we
have $\denote{\greencounit} = \bra{+}$, the physical process of
measurement, including its non-deterministic aspect, cannot be
represented.  To address this we now introduce the notion of a
\emph{$\mathcal{V}$-labelled diagram}.  This will require a
modification to the syntax, and a new interpretation based on
superoperators, rather than \fdhilb.

\begin{definition}\label{def:conditional-diagram}
  Let $\mathcal{V}$ be some set of variables.  A \emph{conditional
    diagram} is a diagram (cf Definition~\ref{def:diagram}) where each
  $Z$ or $X$ vertex with $\alpha \neq 0$ is additionally labelled by a
  set $\mathcal{U} \subseteq \mathcal{V}$.
\end{definition}

If a vertex is labelled by a $\mathcal{U} \neq \emptyset$ then it is
called \emph{conditional}, otherwise it is \emph{unconditional}.  A
diagram with no conditional vertices is called unconditional.  The 
allowed vertices of $\mathcal{V}$-labelled diagrams are shown in
Figure~\ref{fig:vertex-types-with-variables}.  

The equational rules must also be modified to take account the labels:
certain rewrites are only allowed when the variable sets agree, in a
sense that will be made clear below.  The updated rules are shown in
Figure~\ref{fig:rewrite-rules-bis}.

\begin{figure}
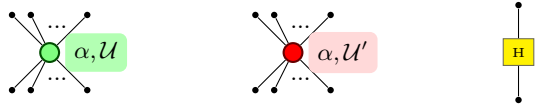

  \centering
    \greenspider{$\alpha,\mathcal{U}$}\hspace{1cm}
    \redspider{$\alpha,\mathcal{U'}$}\hspace{1cm}
    \hgate
  \caption{Interior vertices of $\mathcal{V}$-labelled diagrams}
  \label{fig:vertex-types-with-variables}
\end{figure}

\begin{figure}
  \centering
  \fbox{
  \begin{minipage}[c]{1.0\linewidth}
  \centering
\[
\begin{array}{ccc}
\greenSpiderLHS
\rTo
\greenSpiderRHS
& 
\qquad \quad 
&
\greenAntiLoopLHS
\rTo
\greenAntiLoopRHS
\\
\text{(spider)}&&\text{(anti-loop)}
\\ \\[-1ex]
\greenCommutesLHS
\rTo
\greenCommutesRHS
&
\qquad \quad 
&
\greenIdentityLHS
\rTo
\greenIdentityRHS
\\
\text{($\alpha$-commute)} && \text{(identity)} 
\\ \\[-1ex]
\copyingLHS
\rTo
\copyingRHS
&
\qquad \qquad
&
\piCommutesLHS
\rTo
\piCommutesRHS
\\ 
\text{(copying)}  &&  \text{($\pi$-commute)}
\\ \\[-1ex]
\bialgebraLHS
\rTo
\bialgebraRHS
&
\qquad \qquad
&
\hopfLawLHS
\rTo
\hopfLawRHS
\\
\text{(bialgebra)} &&\text{(hopf)} 
\\ \\[-1ex]
\removeGreenLHS
\rTo
\removeGreenRHS
& 
\qquad \qquad
&
\HsquaredLHS
\rTo
\HsquaredRHS
\\ \\[-1ex]
\text{($H$-commute)}  &&\text{($H$-cancel)} 
\end{array}
\]
\end{minipage}}
\caption{Rewrite rules for the \zxcalculus with conditional vertices.
  We present the rules for the $Z$ subsystem; to obtain the complete
  set of rules exchange the colours in the rules shown above.  }
\label{fig:rewrite-rules-bis}
\end{figure}

For any given $\mathcal{V}$, the $\mathcal{V}$-labelled diagrams,
quotiented by the equations of Figure~\ref{fig:rewrite-rules-bis} form
a $\dag$-compact category denoted \DDV; $\DD(\emptyset)$ is exactly
the category \DD defined earlier.

\begin{definition}\label{def:valuation}
  A function $v:\mathcal{V}\to \{0,1\}$ is called a \emph{valuation}
  of $\mathcal{V}$; for each valuation $v$, we define a functor
  $\hat{v}:\DDV \to \DD$ which produces a new diagram by relabelling
  the $Z$ and $X$ vertices.  If a vertex $z$ is labelled by $\alpha$
  and $\mathcal{U}$, then $\hat{v}(z)$ is labelled by $0$ if
  $\sum_{s\in \mathcal{U}} v(s) = 0$ and $\alpha$ otherwise.
\end{definition}

The modified rewrite rules for $\mathcal{V}$-labelled diagrams are
simply the original equations, with the constraint that they should be
true in all valuations.

\begin{definition}\label{def:semantic-pt2}
  Let $D$ be a diagram in \DDV such that every variable of
  $\mathcal{V}$ occurs in $D$.  Define a symmetric monoidal functor
  $\denoteV{\cdot}: \DDV \to \superop$ by setting $\denoteV{1} =
  \mathbb{C}^2\times\mathbb{C}^2$ and, for every diagram $D$, defining
  $\denoteV{D}$ as the superoperator constructed by summing over all
  the valuations of $\mathcal{V}$:
  \[
  \rho \mapsto \sum_{v\in 2^\mathcal{V}}
  \denote{\hat{v}(D)}\rho\denote{\hat{v}(D)}^\dag\;.
  \]
\end{definition}

\begin{proposition}\label{prop:dag-preserved}
  For any diagram $D$ in \DDV, we have $\denoteV{D^\dag} = \denoteV{D}^\dag$.
\end{proposition}

\begin{example}\label{ex:measurement}
  The following diagram represents the measurement of a single qubit
  in the basis $\ket{\pm_\alpha} = \ket{0} \pm e^{i\alpha}\ket{1}$.  
  \[
  \meas
  \]
  The variable $i$
  encodes which of the two possible outcomes occurred, as can be seen
  by computing the  denotation:
  \[
  \begin{split}
   \rho \mapsto  \sum_{v = 0,1} \denote{\hat{v}(PIC)}\rho
    \denote{\hat{v}(PIC)}^\dag & = \bra{+_\alpha}\rho\ket{+_\alpha}
    +\bra{+_\alpha}Z\rho Z\ket{+_\alpha} \\ &=
    \bra{+_\alpha}\rho\ket{+_\alpha} +\bra{-_\alpha}\rho
    \ket{-_\alpha}
  \end{split}
  \]
\end{example}

\begin{example}\label{ex:classically-controlled-Z}
  A classically controlled Pauli-$Z$ operation is represented by the
  following diagram:
  \[
  \greenphasesig{\alpha}{\{i\}}
  \]
\end{example}

\begin{example}  \label{ex:teleportation}
  Combining the two previous examples, we present the teleportation
  protocol \cite{BBCJW:1993:teleport}.
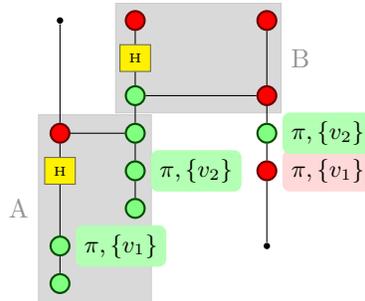
  \[
  \teleportLabelled
  \]
  The block labelled ``B'' is the circuit to prepare a Bell state,
  while ``A'' represents the Bell basis measurement of
  two qubits.  Notice that the same variables label the measurements as the
  corresponding correction operators, indicating that these operations
  must be correlated.  We can rewrite this diagram to prove the
  correctness of the protocol.
  \begin{multline*}
    \teleportProofi = \teleportProofii \\ = \teleportProofiii =
    \teleportProofiv = \teleportProofv
  \end{multline*}
  Since all the conditional vertices are removed by the final step,
  we can conclude that the teleportation protocol is
  \emph{deterministic}.
\end{example}

\begin{example}\label{ex:correlated-measurements}
It is also possible to write down diagrams which correspond to quite
unphysical operations.  For example the diagram
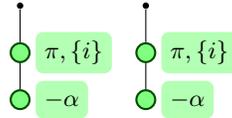
\[
\meas \meas
\]
represents two one-qubit measurements whose outcomes are always
perfectly correlated, regardless of the input.  This is of course
plainly impossible in quantum mechanics.

To avoid such situations, each single qubit measurement must be
labelled by a fresh variable; any other vertex labelled by the same
variable must be interpreted as an operation which is classically
controlled by the outcome of that measurement.  
\end{example}

\section{The measurement calculus}
\label{sec:graph-stat-meas}

Now we turn our attention to the details of \emph{measurement-based quantum
computation} (\mbqc).  While there are several approaches to \mbqc, we
will be concerned only with the \emph{one-way model} (\owqc)
introduced by Raussendorf and
Briegel
\cite{Raussendorf-2001,RauBri:OnewayQC:2003,H.-J.-Briegel:2009hb}.

Whereas the quantum circuit model consists of reversible unitary
gates, the \owqc carries out computation via the irreversible state
changes induced by quantum measurements.  The computation begins with
some input qubits coupled to a large entangled resource state, called
a cluster state or \emph{graph state}. Single qubit measurements are
performed upon the state. Since the measured qubits are no longer
entangled with the rest of the resource we can view the measurement process as
consuming the resource, hence the name \emph{one-way}.  The
computation proceeds by a number of rounds of measurement, where the
choice of measurement performed in later rounds may depend on the
observed outcomes of earlier measurements, until finally only the
output qubits remain unmeasured.  It may then be necessary to apply
some single-qubit unitary corrections to obtain the desired result.
Aside from the initial creation of the resource state, all the
operations of the \owqc act locally on a single qubit, so it is
perhaps surprising that the \owqc is universal for quantum computing:
any unitary operation on $n$ qubits may be computed by the \owqc.

We shall formally describe the \owqc using the syntax of the
\emph{measurement calculus} of Danos, Kashefi and Panangaden
\cite{DanosV:meac}.  Measurement calculus programs, called
\emph{patterns}, consist of a (finite) set of qubits, a (finite) set
of Boolean variables called \emph{signals}, and a
sequence of \emph{commands} $C_n \ldots C_1 C_0$, read
from right to left.  The possible commands are:
\begin{itemize}
\item $N_i$ : initialise qubit $i$ in the state $\ket{+}$.
\item $E_{ij}$ : entangle qubits $i$ and $j$ by apply applying a \CZ
  operation.
\item $M^\alpha_i$ : measure qubit $i$ in the basis
  $\ket{\pm_\alpha} \coloneqq \frac{1}{\sqrt{2}}(\ket{0} \pm
  e^{i\alpha}\ket{1})$.  We assume the measurements are destructive,
  so qubit $i$ will play no further role in the computation.
\item $X_i$ and $Z_i$ : apply a 1-qubit Pauli $X$ (resp. $Z$) operator
  to qubit $i$.  These are called \emph{corrections}.
\end{itemize}
The measurement and correction commands may be classically controlled
by the value of one or more signals.
\begin{itemize}
\item ${}^s[M^\alpha_i]^t$ : perform the measurement $M^{(-1)^s\alpha +t\pi}_i$.
\item $X_i^s$ and $Z_i^s$ : if $s = 1$ perform the command $X_i$
  (resp. $Z_i$), otherwise do nothing.
\end{itemize}
In principle the signals could obtain their values from any source,
however they will always be associated to the outcome of a measurement
already performed in the same pattern\footnote{ It is easy to
  generalise the conditional commands to allow classical control by
  some arithmetic expression of signals, however this is not required
  here, and indeed does not increase the expressiveness of the
  measurement calculus.  }.  If the $\ket{+_\alpha}$ outcome is
obtained, the corresponding signal is set to zero; otherwise it set to
one.  This introduces the third of three determinacy conditions:
\begin{enumerate}
\item The initialisation of a qubit is the first command acting on it.
\item The measurement of a qubit is the last command acting on it.
\item No command depends upon a measurement not already performed.
\end{enumerate}
Any qubits not initialised are \emph{inputs}; any qubit not measured
is an \emph{output}.

\begin{example}\label{ex:mc-simple-example}
  Consider the 2-qubit pattern
  \[
  \patP_H \coloneqq X_2^{s_1} M^0_1 E_{12} N_2\;.
  \]
  Since qubit 1 is not initialised, it must be an input; similarly
  qubit 2 is an output.  The only signal is associated to the
  measurement of qubit 1.  Suppose that the input qubit is in state
  $\ket\psi = a\ket0 + b\ket 1$.  The execution proceeds as follows:
  \begin{multline*}
    \ket\psi \rTo^{N_2} \ket\psi\otimes\ket+ 
     = a\ket{00} + a\ket{01} + b\ket{10} + b\ket{11} \\
     \rTo^{E_{12}} a\ket{00} + a\ket{01} +  b\ket{10} - b\ket{11} \\
    = \ket+(a\ket+ + b\ket-) + \ket-(a\ket+ - b\ket-)
  \end{multline*}
  So far all we have done is construct the initial entanglement.  The
  next step is the measurement $M^0_1$.  Suppose that the result of
  the measurement is 0, \ie the projection onto $\ket{+}$; then qubit
  1 is eliminated and we have the new state $a\ket+ + b\ket-$.  Since
  the signal $s_1$ is zero, there is no need to perform the final
  correction.  On the other hand, should the outcome of the
  measurement be 1, the resulting state will be $a\ket+ - b\ket-$, and
  since $s_1 = 1$ the $X_2$ correction must be applied, again producing a
  final state of $a\ket+ + b\ket-$.  Hence the overall effect of
  $\patP_H$ is independent of the outcome of the measurement: in
  either case the computer applies a   Hadamard gate upon its input.
\end{example}

This example illustrates a key feature of the \owqc.  The measurement
introduces a branch in the execution where one outcome corresponds to
the `correct' behaviour, in the sense
that the projection achieves the desired computational effect, and one
branch contains  an `error' which must be corrected later
in the pattern.  More generally, a pattern with $n$ measurements
potentially performs $2^n$ different linear maps on its input, which
are called the \emph{branch maps};  the branch where all the
measurements output 0 is called the \emph{positive branch}, and we
make the convention that this branch is the computation that we intend
to carry out.  The pattern is deterministic if all the branch maps
have the same effect on the input as the positive branch.

\begin{remark}\label{rem:native-mc-sementics}
  The concept of branch map is used in \cite{DanosV:meac} to define
  the semantics of the measurement calculus.  We omit this, because we
  shortly provide a semantics via a translation into the \zxcalculus.
  The interpretation presented here is equivalent to
  that of Danos, Kashefi, and Panangaden \cite{DanosV:meac}.
\end{remark}

\begin{theorem}\label{thm:emc-form}
  For any pattern \patP, there exists an equivalent pattern---in the
  sense of having the same semantics---whose command
  sequence has the form 
  \[
  \patP^* = C M E N
  \]
  where $C$, $M$, $E$, and $N$ are sequences of commands consisting
  exclusively of corrections, measurements, entangling operations, and
  initialisations respectively.  
\end{theorem}

The pattern $\patP^*$ is said to be in \emph{standard form}; 
any pattern may be transformed into an equivalent by a rewrite
procedure presented in \cite{DanosV:meac}.  There are two main ingredients
to this procedure.  The first are the commuting relations between the
Pauli matrices and the \CZ.  The second is the following identity,
allowing corrections to be absorbed into conditional measurements:
\begin{equation}\label{eq:no-cond-meas-in-mc}
  {}^s[M^\alpha_i]^t  = M^\alpha_i Z_i^t X_i^s \;.
\end{equation}
However, this will also function as the \emph{definition} of the
conditional measurement, in which corrections are  only conditional
commands.  Needless to say, with
this convention, Theorem~\ref{thm:emc-form} no longer applies.  We
will adopt the following convention for the rest of the chapter:
patterns are always assumed to be in standard form, including
conditional measurements, which are then replaced with unconditional
measurements, as per (\ref{eq:no-cond-meas-in-mc}).

\begin{definition}\label{def:pattern-diag-translation}
  Let \patP be a pattern with qubits ranged over by $V$, and let $I,O
  \subseteq V$ denote its input and output qubits respectively.  We
  define a diagram $\Dp : \sizeof{I} \to \sizeof{O}$ in
  $\DD(V\setminus I)$, by composing the subdiagrams corresponding to
  the command sequence of \patP, as shown in table
  \ref{tab:translation}.
\end{definition}
\begin{table}[h]
  \centering
  \begin{tabular}{c|c|c|c|c}
    $N_i$ & $E_{ij}$ & $M_i^\alpha$ & $X_i^s$ & $Z_i^s$ \\
    \hline
    \rule{17mm}{0mm}&
    \rule{17mm}{0mm}&
    \rule{17mm}{0mm}&
    \rule{17mm}{0mm}&
    \rule{17mm}{0mm}\\
    \inline{\greenunit}&
    \inline{\czed}&
    \meas &
    \inline{\redphasesig{\pi}{\{s\}}}&
    \inline{\greenphasesig{\pi}{\{s\}}}
  \end{tabular}
  \caption{Translation from pattern to diagram.}
  \label{tab:translation}
\end{table}

Since the semantics of the diagram $\Dp$ capture the  meaning of the
pattern \patP, reasoning in the \zxcalculus allows properties of the
the pattern to be derived by graphical manipulations of \Dp.  Let's
reconsider the Example~\ref{ex:mc-simple-example}.

\begin{example}\label{ex:mc-simple-diag}
  The pattern $\patP_H$ produces the following diagram:
  \[
  D(\patP_H) \coloneqq \inline{%
\beginpgfgraphicnamed{pat-hadamard-i}
\InputIfFileExists{pat-hadamard-i.tikz}{}{\input{./tikz/pat-hadamard-i.tikz}}
\endpgfgraphicnamed} \;.
  \]
  Note that the vertical wires correspond to the physical qubits; we
  shall stick to this convention whenever convenient.
  Now we can deduce that $\patP_H$ computes the Hadamard gate
  using purely diagrammatic reasoning:
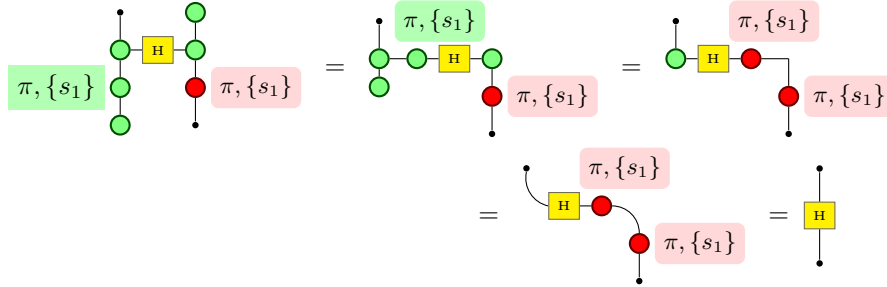
\begin{multline*}
  \inline{%
\beginpgfgraphicnamed{pat-hadamard-i}
\InputIfFileExists{pat-hadamard-i.tikz}{}{\input{./tikz/pat-hadamard-i.tikz}}
\endpgfgraphicnamed} = \inline{%
\beginpgfgraphicnamed{pat-hadamard-ii}
\InputIfFileExists{pat-hadamard-ii.tikz}{}{\input{./tikz/pat-hadamard-ii.tikz}}
\endpgfgraphicnamed} =
  \inline{%
\beginpgfgraphicnamed{pat-hadamard-iii}
\InputIfFileExists{pat-hadamard-iii.tikz}{}{\input{./tikz/pat-hadamard-iii.tikz}}
\endpgfgraphicnamed} \\ = \inline{%
\beginpgfgraphicnamed{pat-hadamard-iv}
\InputIfFileExists{pat-hadamard-iv.tikz}{}{\input{./tikz/pat-hadamard-iv.tikz}}
\endpgfgraphicnamed} = \hgate
\end{multline*}
\end{example}

\begin{example}\label{ex:define-mbqc-cnot}
  The ubiquitous CNOT operation can be computed by the pattern 
  $\patP_{\CX} = X^{3}_4 Z^{2}_4Z^{2}_1 M^0_3 M^0_2 E_{13}E_{23} E_{34} N_3
  N_4$ \cite{DanosV:meac}.  This yields the diagram
  \[
  D(\patP_{\CX}) = \cnoti,
  \]
  with qubit 1 the leftmost, and qubit 4 is the rightmost.  Now, we
  can prove the correctness of the pattern by rewriting:
  
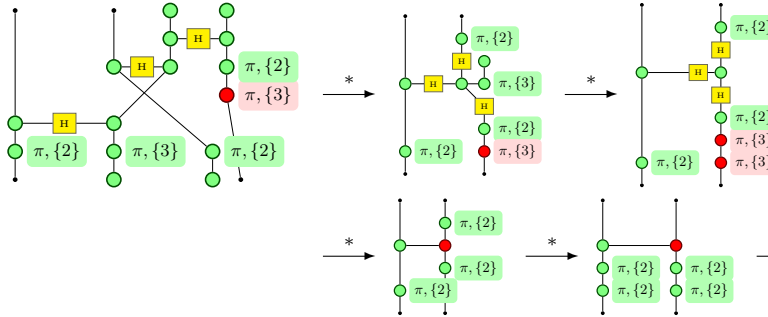
\[
  \begin{split}
    \cnoti[0.75] &\rTo^*
    \cnotii[0.6] \rTo^*
    \cnotiii[0.6] \\
    {}& \rTo^* 
    \cnotiv[0.6] \rTo^*
    \cnotv[0.6] \rTo^*
    \cnotvi[0.75] 
  \end{split}
  \]
  One can clearly see in this example how the non-determinism introduced
  by measurements is corrected by conditional operations later in the
  pattern.
  The possibility of performing such corrections depends on the
  \emph{geometry} of the pattern, the entanglement graph implicitly defined
  by the pattern.  This will be the main concern of the next section.
\end{example}

\section{Determinism and flow}
\label{sec:determinism-flow}

Consider the pattern $\mathfrak{N} = Z^{s_2}_1 M^0_2 M^\alpha_3 E_{12} E_{23}
N_2 N_3$.  Working in the \zxcalculus, it can be rewritten as follows:
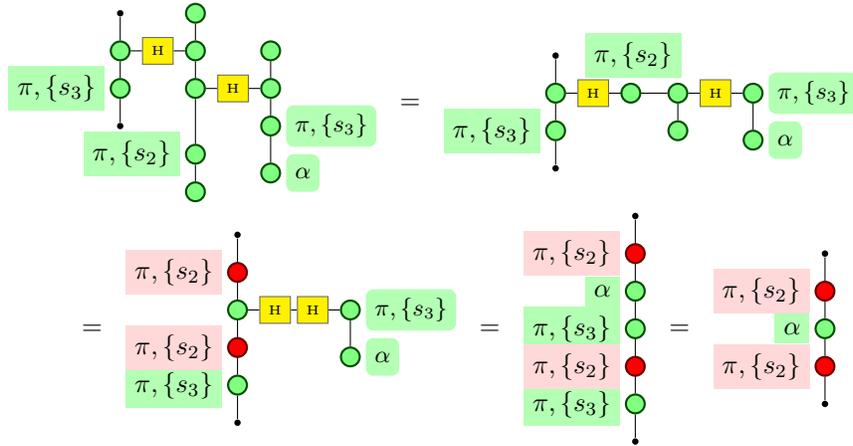
\begin{multline*}
  \inline{%
\beginpgfgraphicnamed{no-flow-example-i}
\InputIfFileExists{no-flow-example-i.tikz}{}{\input{./tikz/no-flow-example-i.tikz}}
\endpgfgraphicnamed} = \inline{%
\beginpgfgraphicnamed{no-flow-example-ii}
\InputIfFileExists{no-flow-example-ii.tikz}{}{\input{./tikz/no-flow-example-ii.tikz}}
\endpgfgraphicnamed} \\ = 
  \inline{%
\beginpgfgraphicnamed{no-flow-example-iii}
\InputIfFileExists{no-flow-example-iii.tikz}{}{\input{./tikz/no-flow-example-iii.tikz}}
\endpgfgraphicnamed} = \inline{%
\beginpgfgraphicnamed{no-flow-example-iv}
\InputIfFileExists{no-flow-example-iv.tikz}{}{\input{./tikz/no-flow-example-iv.tikz}}
\endpgfgraphicnamed} =
  \inline{%
\beginpgfgraphicnamed{no-flow-example-v}
\InputIfFileExists{no-flow-example-v.tikz}{}{\input{./tikz/no-flow-example-v.tikz}}
\endpgfgraphicnamed}
\end{multline*}

Hence we have a pattern which is non-deterministic, acting as a
$Z$-rotation by either $\alpha$ or $-\alpha$ depending on the outcome of
measurement 2.  Furthermore, unlike the previous examples, there is no
way to remove the dependence on $s_2$ by correction at qubit 1, or
by conditional measurement at qubit 3.

As this example shows, not every pattern performs a deterministic
computation, and this possibility depends not just upon the
conditional operations introduced by the programmer, but also the
structure of the entangled resource used in the computation.  This
structure is called the \emph{geometry} of the pattern.  A pattern
can perform a deterministic computation if its geometry has a graph
theoretic property called \emph{flow}.
Examples~\ref{ex:mc-simple-example} and \ref{ex:define-mbqc-cnot} have
flow, while the pattern $\mathfrak{N}$ above does not.

Flow is a sufficient property for determinism, but not a necessary
one: it guarantees \emph{strong} and \emph{uniform} determinism.
Strong determinism means that all the branch maps of the
pattern are equal, while uniformity means that the pattern is
deterministic for all possible choices of its measurement angles.
Before moving on, we make the following obvious observation:

\begin{theorem}\label{thm:circuit-like-implies-deterministic}
  If \Dp can be rewritten to an unconditional diagram then \patP is
  strongly deterministic.
\end{theorem}

Hence, to show that a pattern is deterministic, it suffices to find
some rewrite sequence which removes all the conditional vertices.
Typically this is done by `pushing' the conditional vertex introduced
by measurement through the diagram until it meets a matching
corrector.  The two conditional vertices then cancel each other out.
The rest of this section will explore when this is possible.

\begin{definition}\label{def:geometry-of-pattern}
  The \emph{geometry} of a pattern \patP, denoted \Gp, is the open graph
  $((V,E),I,O)$ defined by taking the qubits of \patP as vertices $V$, the
  input and output qubits as the sets $I$ and $O$, and defining the
  edge relation by  $v \sim u $ if and only if the  command $E_{vu}$
  occurs in \patP.
\end{definition}

\begin{example}\label{ex:geomtrey-of-cnot}
Consider $\patP_{\CX} =  X^{3}_4 Z^{2}_4Z^{2}_1 M^0_3 M^0_2 E_{13}E_{23}
E_{34} N_3$ as in Example~\ref{ex:define-mbqc-cnot}.  We then have 
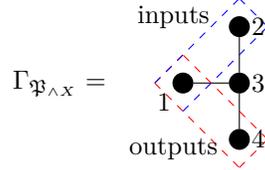
\[
\G{\patP_{\CX}} = \inline{%
\beginpgfgraphicnamed{geom-cnot}
\InputIfFileExists{geom-cnot.tikz}{}{\input{./tikz/geom-cnot.tikz}}
\endpgfgraphicnamed}
\]
\end{example}

\begin{definition}\label{def:flow}
  Let $G = ((V,E),I,O)$ be an open graph; a \emph{flow} on $G$
  is a pair $(f,\prec)$, where $f$ is a function $V\setminus O \to
  V\setminus I$ and $\prec$ is a partial order on $V$, satisfying
  \renewcommand{\theenumi}{(F\arabic{enumi})}
  \begin{enumerate}
  \item $f(u) \sim u$; \label{item:flow1}
  \item $u < f(u)$;\label{item:flow2}
  \item  If $f(u) \sim v$ and $u \neq v$ then $u < v$.\label{item:flow3}    
  \end{enumerate}
  \renewcommand{\labelenumi}{\theenumi.}
\end{definition}

Intuitively, the function $f$ specifies a causal successor for every
measured qubit.  Should the measurement at qubit $i$ give the `wrong'
answer, then a conditional operation at qubit $f(i)$ can be used to
correct the resulting error.  The partial order ensures that no
causal loops can form: that is, it is not necessary to apply a correction
to a qubit that has already been measured.  We have the
following:

\begin{theorem}[\cite{Danos2006Determinism-in-}]\label{thm:original-flow}
  If $G$ is an open graph with flow then there exists a strongly
  and uniformly deterministic pattern \patP such that $G = \Gp$.
\end{theorem}

It must be emphasised that flow is a property of the geometry
\emph{not} the pattern itself.  In order for a pattern to be deterministic,
correct placement of the conditional operations is still required.
Danos and Kashefi give the pattern \patP explicitly in
\cite{Danos2006Determinism-in-}; we will reconstruct this later.  

\begin{definition}\label{def:diag-for-geometry}
  Let $\Gamma =((V,E),I,O)$ be an open graph; we define an
  unconditional diagram \DG as follows:
  \begin{itemize}
  \item The vertices of \DG are given by the disjoint union $V+E+I+O$.
    Should $\Gamma$ have vertex $v$ contained in both $I$
    and $O$ then $\DG$ contains \emph{three} corresponding vertices
    in \DG; we use subscripts $v_V, v_I, v_O$ to disambiguate.
  \item The vertices are typed depending which disjoint subset they
    originate in: those from $V$ (the original vertices of $\Gamma$)
    have type $Z$, without any label; those from $E$ have type $H$;
    and those from $I+O$ are boundary vertices, with $I$ providing the
    inputs and $O$ the outputs.
  \item If $e$ is an edge in $\Gamma$ connecting vertices $u$ and $v$,
    then we have $u_V \sim e$ and $e\sim v_V$ in $\DG$. For the
    boundary vertices we have $v_I \sim v_V$ and $v_O \sim v_V$.
  \end{itemize}
\end{definition}

\begin{example}  \label{ex:diag-for-geomtery}
  Consider again the $\CX$ pattern, or rather its geometry.
  \[
  \G{\patP_{\CX}} = \inline{%
\beginpgfgraphicnamed{geom-cnot}
\InputIfFileExists{geom-cnot.tikz}{}{\input{./tikz/geom-cnot.tikz}}
\endpgfgraphicnamed} 
  \qquad 
  \D{\G{\patP_{\CX}}} = \inline{%
\beginpgfgraphicnamed{pat-geom-cnot}
\InputIfFileExists{pat-geom-cnot.tikz}{}{\input{./tikz/pat-geom-cnot.tikz}}
\endpgfgraphicnamed}
  \]
\end{example}

\begin{theorem}\label{thm:flow-implies-circuit-like}
  Let \patP be a pattern;  if \Gp has flow then \DGp is
  equivalent to a circuit-like diagram.
\end{theorem}

\begin{proof}
  Suppose that \Gp has a flow $(f,\prec)$.  Let $J$ be the vertices of
  \Gp which are minimal with respect to $\prec$.  For each vertex
  $j\in J$ we can define a finite sequence $p_j = j, f(j),
  f^2(j),\ldots,f^n(j)$ where the last element of the sequence is an
  output qubit. By definition of $f$, the collection $\cup_{j\in J}
  p_j$ contains all the vertices of \Gp.  Each $p_j$ defines a path in
  \DGp, and the collection of these paths covers all the $Z$ vertices
  of \DGp; we can trivially extend these paths to include the boundary
  vertices adjacent to their end points.  The collection
  $\{p_j\}_{j\in J}$ provides the path covering required by the
  definition of circuit-like
  (cf. Definition~\ref{def:circuit-like-diagram}).   \DGp
  satisfies condition \ref{circuit-like:2}, because of the partial
  order structure of the flow, and condition \ref{circuit-like:3} by
  construction, however it does not satisfy condition \ref{circit-like:1},
  since some vertices are not covered by the path.  Specifically,
  those $H$ vertices $e$ where $v \sim e \sim u$ and $f(u) \neq v$ are
  not covered by the path.  At each such vertex we perform the
  following rewrite:
  \[
  \inline{%
\beginpgfgraphicnamed{circuitify-rewrite-lhs}
\InputIfFileExists{circuitify-rewrite-lhs.tikz}{}{\input{./tikz/circuitify-rewrite-lhs.tikz}}
\endpgfgraphicnamed} 
  = \inline{%
\beginpgfgraphicnamed{circuitify-rewrite-rhs}
\InputIfFileExists{circuitify-rewrite-rhs.tikz}{}{\input{./tikz/circuitify-rewrite-rhs.tikz}}
\endpgfgraphicnamed} 
  \]
  Now the $H$ vertices can be removed via the rewrite shown below:
  \[
  \czed = \inline{%
\beginpgfgraphicnamed{circuitify-rewrite-ii-rhs}
\InputIfFileExists{circuitify-rewrite-ii-rhs.tikz}{}{\input{./tikz/circuitify-rewrite-ii-rhs.tikz}}
\endpgfgraphicnamed}
  \]
  After which any pairs of adjacent $H$ vertices may be cancelled, and
  the spider rule can be used to guarantee that the diagram is three coloured.
\end{proof}

The converse to Theorem~\ref{thm:flow-implies-circuit-like} does not
hold; in \cite{Duncan:2010aa} it is shown that geometries which have
\emph{generalised flow}, discussed below, can be rewritten to
circuit-like diagrams.  However, we can give a weaker result which
holds for flow.

\begin{definition}\label{def:weakly-circuit-like}
  Let $D$ be a diagram, and let $U \subseteq V$ be a set of its
  vertices satisfying 
  \begin{itemize}
  \item $u\in U$ implies that $u$ has type $H$;
  \item $v_1 \sim u \sim v_2$ implies that $v_1$ and $v_2$ are
    either both of type $Z$ or both of $X$.
  \end{itemize}
  Then $D$ is called \emph{weakly circuit-like} if the following
  conditions hold.
  \renewcommand{\theenumi}{(W\arabic{enumi})}
  \begin{enumerate}
  \item \label{weak-circit-like:1} The vertices $V \setminus U$ can be
    covered by a set $\mathcal{P}$ of disjoint directed paths, each of
    which ends in an output;
  \item \label{weak-circuit-like:2}for every oriented cycle $\gamma$ in the
    diagram, if $\gamma$ contains at least 2 edges from different
    paths in $\mathcal{P}$, then it traverses at least one of them in
    the direction opposite to that induced by the path; and,
  \item \label{weak-circuit-like:3} it is a simple graph, and is 
    3-coloured.
  \end{enumerate}
  \renewcommand{\labelenumi}{\theenumi.}
\end{definition}

Weakly circuit-like diagrams correspond to circuits where
the 2-qubit gates may be \CZ gates as well as \CX gates.

\begin{theorem}\label{thm:weak-circuit-like-implies-flow}
  Let \patP be a pattern; if \DGp is weakly circuit-like then \Gp has
  flow.
\end{theorem}
\begin{proof}
  By construction, the $Z$ vertices are in bijective correspondence with
  the qubits of \patP; hence we need only define a flow $(f,\prec)$
  over the $Z$ vertices.  Let $p$ be one of the paths of the path
  covering of \DGp; $p$ then defines a linear order over the $Z$
  vertices it covers.  Define $f$ by $f(u) = v$ whenever $v$ is the
  successor of $u$ in this order.  Since the paths are disjoint, the
  same procedure can be carried out for every path, and since the
  paths cover all the $Z$ vertices this defines the required function
  $f$.  The union of these linear orders gives a partial order over
  the $Z$ vertices; to obtain the required $\prec$ this order can be
  completed by imposing  condition \ref{item:flow3}.  The acyclicity
  condition \ref{weak-circuit-like:2} guarantees that this is
  possible.
\end{proof}

Evidently, the circuit-like diagram produced from \DGp is not
equivalent to \Dp, however they are closely related.

\begin{definition}\label{def:pat-diag-from-DGp}
  Let \patP be pattern; construct a new diagram \DGpp from \DGp as
  follows:
  \begin{itemize}
  \item If ${}^s[M^\alpha_i]^t$ occurs in \patP then modify \DGp by
    adjoining the subdiagram corresponding to the measurement, as
    shown below:
    \[
    \inline{%
\beginpgfgraphicnamed{adjoin-meas-lhs}
\InputIfFileExists{adjoin-meas-lhs.tikz}{}{\input{./tikz/adjoin-meas-lhs.tikz}}
\endpgfgraphicnamed} \to
    \inline{%
\beginpgfgraphicnamed{adjoin-meas-rhs}
\InputIfFileExists{adjoin-meas-rhs.tikz}{}{\input{./tikz/adjoin-meas-rhs.tikz}}
\endpgfgraphicnamed}
    \]
  \item If $X^s_i$ or $Z^s_i$ occurs in \patP then adjoin the subdiagram
    corresponding to the correction as shown below:
    \[
    \inline{%
\beginpgfgraphicnamed{adjoin-corr-lhs}
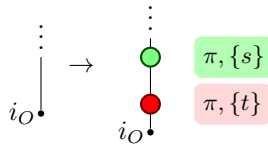
\begin{tikzpicture}
	\begin{pgfonlayer}{nodelayer}
		\node [style=none] (0) at (-1.5, -0.5) {$i_O$};
		\node [style=none] (1) at (-1, 1) {};
		\node [style=none] (2) at (-1, 1.75) {$\vdots$};
		\node [style=boundary vertex] (3) at (-1, -0.5) {};
	\end{pgfonlayer}
	\begin{pgfonlayer}{edgelayer}
		\draw [style=plain] (3) to (1.center);
	\end{pgfonlayer}
\end{tikzpicture}}
\endpgfgraphicnamed} \to
    \inline{%
\beginpgfgraphicnamed{adjoin-corr-rhs}
\InputIfFileExists{adjoin-corr-rhs.tikz}{}{\input{./tikz/adjoin-corr-rhs.tikz}}
\endpgfgraphicnamed}
    \]
    Note that since the \patP is in standard form, corrections can
    only appear at an output qubit.
  \end{itemize}
\end{definition}

\begin{example}\label{ex:pat-diag-from-DGp}
  Recall the pattern $\patP_{\CX} = X^{3}_4 Z^{2}_4Z^{2}_1 M^0_3 M^0_2
  E_{13}E_{23} E_{34} N_3 N_4$.  We have:
  \[
  \D{\G{\patP_{\CX}}} = \inline{%
\beginpgfgraphicnamed{pat-geom-cnot}
\InputIfFileExists{pat-geom-cnot.tikz}{}{\input{./tikz/pat-geom-cnot.tikz}}
\endpgfgraphicnamed} \qquad
  \D{\G{\patP_{\CX}}}^* = \inline{%
\beginpgfgraphicnamed{pat-geom-cnot-star}
\InputIfFileExists{pat-geom-cnot-star.tikz}{}{\input{./tikz/pat-geom-cnot-star.tikz}}
\endpgfgraphicnamed} 
  \]
\end{example}

\begin{lemma}\label{lem:reduce-dp-to-dgpp}
  For any pattern \patP we have $\Dp \to \DGpp$
\end{lemma}
\begin{proof}
  All the $Z$ vertices in \Dp which are introduced by $N_i$ and
  $E_{ij}$ commands form a connected subgraph, hence they can all be
  contracted together via the spider rule;  this gives \DGpp. 
\end{proof}

\begin{corollary}\label{cor:Dp-rewrites-to-almost-circuit-like}
  If \patP has flow then the subgraph of \Dp excluding the
  measurements rewrites to a circuit-like diagram.
\end{corollary}

If \patP has a circuit-like geometry then
Theorem~\ref{thm:original-flow} shows that it could be deterministic,
if there are corrections in the appropriate places.  The \zxcalculus can
be used to determine where the corrections must be placed.  Suppose
that we have the configuration shown below.  (The dotted boxes
represent either measurements or outputs.)
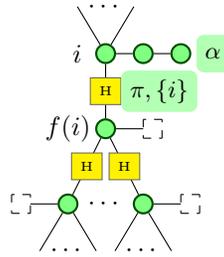
\[
    \inline{%
\beginpgfgraphicnamed{corrector-placement-i}
\InputIfFileExists{corrector-placement-i.tikz}{}{\input{./tikz/corrector-placement-i.tikz}}
\endpgfgraphicnamed}
\]
The measurement at qubit $i$ introduces an error term which must be
cancelled at a later qubit.  We can perform the following rewrite
sequence:
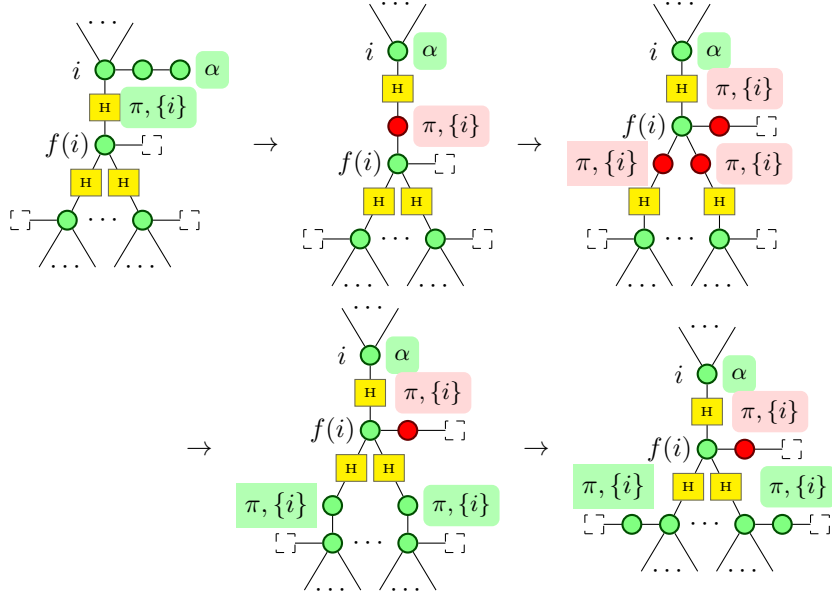
\begin{multline*}
    \inline{%
\beginpgfgraphicnamed{corrector-placement-i}
\InputIfFileExists{corrector-placement-i.tikz}{}{\input{./tikz/corrector-placement-i.tikz}}
\endpgfgraphicnamed} \to
    \inline{%
\beginpgfgraphicnamed{corrector-placement-ii}
\InputIfFileExists{corrector-placement-ii.tikz}{}{\input{./tikz/corrector-placement-ii.tikz}}
\endpgfgraphicnamed} \to
    \inline{%
\beginpgfgraphicnamed{corrector-placement-iii}
\InputIfFileExists{corrector-placement-iii.tikz}{}{\input{./tikz/corrector-placement-iii.tikz}}
\endpgfgraphicnamed} \\ \to
    \inline{%
\beginpgfgraphicnamed{corrector-placement-iv}
\InputIfFileExists{corrector-placement-iv.tikz}{}{\input{./tikz/corrector-placement-iv.tikz}}
\endpgfgraphicnamed} \to
    \inline{%
\beginpgfgraphicnamed{corrector-placement-v}
\InputIfFileExists{corrector-placement-v.tikz}{}{\input{./tikz/corrector-placement-v.tikz}}
\endpgfgraphicnamed}
\end{multline*}
Hence in order to correct for the measurement at qubit $i$, we must
perform a conditional $X$ at qubit $f(i)$, and a condition $Z$ at all
its neighbours except $i$ itself.    Since the geometry
is circuit-like all of these qubits are later in the execution of the
pattern than the measurement of $i$ itself.  Hence we can conclude: 

\begin{theorem}\label{thm:flow-plus-good-correctors-implies-circuit}
  Let \patP be a pattern such that \Gp has flow and, for every
  measured qubit $i$, the command sequence contains correctors
  $X^i_{f(i)}$ and $Z_j^i$ for all $i \neq j \sim f(i)$; then \Dp
  rewrites to an unconditional circuit-like diagram.
\end{theorem}

Of course, if $j$  is a  measured qubit the
corrections can be absorbed into the measurements.

The $Z$ correction at qubit $j$ is effectively the same as an error
introduced by measuring $j$, hence this operation can be deferred
by another step, and the correction performed at $f(j)$ and its
neighbours.  However the same is not true for the corrector $X$.
Consider the following diagrams:
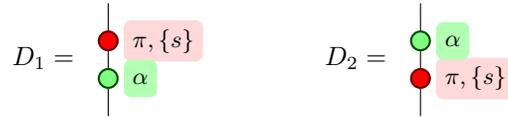
\[
D_1 = \inline{%
\beginpgfgraphicnamed{x-dont-commute-i}
\InputIfFileExists{x-dont-commute-i.tikz}{}{\input{./tikz/x-dont-commute-i.tikz}}
\endpgfgraphicnamed} 
\qquad\qquad 
D_2 = \inline{%
\beginpgfgraphicnamed{x-dont-commute-ii}
\InputIfFileExists{x-dont-commute-ii.tikz}{}{\input{./tikz/x-dont-commute-ii.tikz}}
\endpgfgraphicnamed}
\]
In this case we have $\denote{D_1} \neq \denote{D_2}$ unless $\alpha =
0$ or $\alpha = \pi$, hence these
diagrams are cannot be rewritten to one another. Therefore uniform
determinism requires that the $X$ correction be performed at  $f(i)$
and not later in the computation.  On the other hand,  if $\alpha \in
\{0,\pi\}$ then these diagrams are equal.  This points to an important
advantage of the \zxcalculus.

The diagram \Dp contains all the  information of the pattern itself,
not just its geometry. Therefore by rewriting as discussed above the
correctness of \patP can be verified directly, and in particular
should the \mbqc programmer have made an error in the placement of the
corrections this error will be revealed.  Further, since the
\zxcalculus is sensitive to the values of the angles, it can also be
used to show that a pattern is deterministic even when it is not
uniformly so.  Consider the pattern $\patP = M^0_3 M^\alpha_2 E_{23}
E_{12} N_2 N_3$.  This pattern does not have flow.  However it is
deterministic:
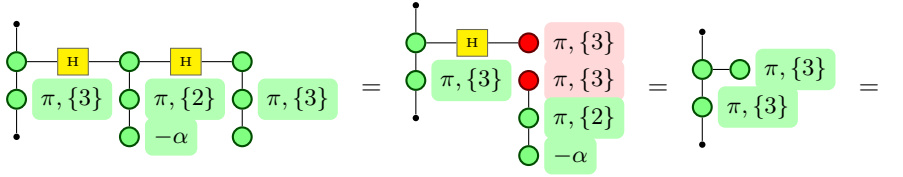
\[
\inline{%
\beginpgfgraphicnamed{non-uniform-i}
\InputIfFileExists{non-uniform-i.tikz}{}{\input{./tikz/non-uniform-i.tikz}}
\endpgfgraphicnamed} =
\inline{%
\beginpgfgraphicnamed{non-uniform-ii}
\InputIfFileExists{non-uniform-ii.tikz}{}{\input{./tikz/non-uniform-ii.tikz}}
\endpgfgraphicnamed} =
\inline{%
\beginpgfgraphicnamed{non-uniform-iii}
\InputIfFileExists{non-uniform-iii.tikz}{}{\input{./tikz/non-uniform-iii.tikz}}
\endpgfgraphicnamed} = \inline{%
\beginpgfgraphicnamed{long-id}
\begin{tikzpicture}
	\begin{pgfonlayer}{nodelayer}
		\node [style=boundary vertex] (0) at (0, 1.75) {};
		\node [style=boundary vertex] (1) at (0, -1.75) {};
	\end{pgfonlayer}
	\begin{pgfonlayer}{edgelayer}
		\draw [style=plain] (0) to (1);
	\end{pgfonlayer}
\end{tikzpicture}}
\endpgfgraphicnamed} 
\]
The disconnected component is just a scalar factor which we drop since
it has no bearing on the computation.

In this chapter we have focused on flow but flow is not a necessary
condition for strong and uniform determinism.  With flow each
measurement has a single successor where the correction must be
performed.  If each qubit has instead a set of correcting qubits this
yields the notion of \emph{generalised flow}
\cite{D.E.-Browne2007Generalized-Flo}.  A computation is called
\emph{stepwise} deterministic if after each measurement the
non-determinism can be removed by correction.  
Generalised flow is both
necessary and sufficient for strong, stepwise, uniform determinism.
The \zxcalculus can also be used to handle generalised flow; in fact
using the rewrite rules, especially the bialgebra rule, any geometry
which has generalised flow can be rewritten to an equivalent pattern
which has flow.  The details can be found in \cite{Duncan:2010aa}.

\section{Conclusions}
\label{sec:conclusions}

Complementarity has long been recognised as one of the fundamental
ingredients of quantum mechanics, although it is usually understood
negatively: as the failure of certain classical properties.  Here we
have demonstrated a positive characterisation of complementarity, in
terms of the existence of certain algebraic structures.  The
\zxcalculus provide a very rich language for reasoning about quantum
systems that fully exploits this algebraic structure.

Due to its graphical nature, the \zxcalculus is extremely legible,
exposing the close parallels between quantum circuits and \owqc
patterns with flow.  In the preceding section we have
seen how the almost metaphorical property of flow actually defines a
trajectory within a diagram  along which
information---in this the outcome of some quantum measurement---must
travel in order to produce a deterministic computation.  This analysis
can be taken further in the analysis of \emph{generalised flow}
\cite{Duncan:2010aa}, where the role of the bialgebra rule is crucial,
functioning as a kind `interference'  principle, where different paths
cancel each out.

The graphical syntax we have employed here is not simply a gimmick.
Since the \zxcalculus is based on algebraic first principles, it can
unify differing computational paradigms such quantum circuits and the
one-way model in a single setting.  Further, since it is based on
graphs, it is amenable to automation, opening the door to mechanised
reasoning about quantum programs \cite{quantomatichome}.

\paragraph{Acknowledgements}
\label{sec:acknowledgements-1}

This chapter is based on work originally carried in collaboration with
Bob Coecke \cite{Coecke:2008nx,Coecke:2009aa} and Simon Perdrix
\cite{Duncan:2009ph,Duncan:2010aa}.  The author is supported
financially by the FRS-FNRS.

\bibliography{flowincats}
\end{document}